\documentclass[12pt,draftclsnofoot,onecolumn]{IEEEtran}
\usepackage{amsfonts}
\usepackage{graphicx}
\usepackage{color}
\usepackage{amsmath,amsfonts,amssymb,amsthm,epsfig,epstopdf,url,array}
\usepackage{url,textcomp}
\usepackage{mdwmath}
\usepackage{mdwtab}
\usepackage{times}
\usepackage{authblk}
\usepackage{cite}
\usepackage{amsbsy}
\usepackage{mathtools}
\usepackage{subfigure}
\usepackage{algorithmic}
\usepackage{stfloats}
\newcommand{\bs}{\boldsymbol}
\DeclareMathAlphabet{\mathpzc}{OT1}{pzc}{m}{it}
\newtheorem{theorem}{Theorem}
\newtheorem{lemma}{Lemma}

\newtheorem{remark}{Remark}
\newtheorem{definition}{Definition}

\begin{document}

\title{Asymptotic Outage Analysis of Spatially Correlated Rayleigh MIMO Channels}
\author{Guanghua~Yang,
Huan Zhang,
Zheng~Shi,
Shaodan~Ma,
and Hong Wang
\thanks{Guanghua Yang and Zheng Shi are with the School of Intelligent Systems Science and Engineering, Jinan University, Zhuhai 519070, China (e-mails: ghyang@jnu.edu.cn, shizheng0124@gmail.com).}  
\thanks{Huan Zhang and Shaodan Ma are  with the Department of Electrical and Computer Engineering, University of Macau, Macao S.A.R., China (e-mails: cquptzh@gmail.com, shaodanma@um.edu.mo).}
\thanks{Hong Wang is with the School of Communication and Information Engineering, Nanjing University of Posts and Telecommunications, Nanjing 210003, China, and also with the National Mobile Communications Research Laboratory, Southeast University, Nanjing 210096, China (e-mail: wanghong@njupt.edu.cn).}
}

\maketitle
\begin{abstract}
The outage performance of multiple-input multiple-output (MIMO) technique has received intense attention in order to ensure the reliability requirement for mission-critical machine-type communication (cMTC) applications. 
In this paper, the outage probability is asymptotically studied for MIMO channels to thoroughly investigate the transmission reliability. 
To fully capture the spatial correlation effects, the MIMO fading channel matrix is modelled according to three types of Kronecker correlation structure, i.e., independent, semi-correlated and full-correlated Rayleigh MIMO channels. The outage probabilities under all three Kronecker models are expressed as representations of the weighted sum of the generalized Fox's H functions. 
The simple analytical results empower the asymptotic outage analyses at high signal-to-noise ratio (SNR), which are conducted not only to reveal helpful insights into understanding the behavior of fading effects, but also to offer useful design guideline for MIMO configurations. Particularly, the asymptotic outage probability is proved to be a monotonically increasing and convex function of the transmission rate. 
In the absence of the channel state information (CSI), the transmitter tends to equally allocate the total transmit power among its antennas to enhance the system reliability especially in high SNR regime. In the end, the analytical results are validated through extensive numerical experiments.

\end{abstract}
\IEEEpeerreviewmaketitle
\section{Introduction}



\IEEEPARstart{5}{G} systems are anticipated to not only support extraordinarily high data rate and capacity, but also provide ultra-reliability, scalability and low latency for emerging communication paradigms, e.g., mission-critical machine-type communications (cMTC) \cite{Tullberg2016METIS}. The multiple-input multiple-output (MIMO) is an indispensable enabler that fulfills these requirements for 5G. Specifically, the MIMO technique capitalizes on the spatial dimension to explore its potentials of boosting the spectral efficiency and reliability \cite{telatar1999capacity,foschini1998limits}. 
Most of the existing works concentrate on studying the information-theoretical capacity of MIMO systems for the purpose of the spectral efficiency enhancement. To name a few, the water-filling power allocation was shown to achieve the capacity if the channel state information (CSI) is perfectly known at transceiver \cite{telatar1999capacity}. In addition, the author in \cite{telatar1999capacity} also pointed out that the ergodic capacity is achievable for independent Rayleigh MIMO channels by using Gaussian random codes together with equal power transmission in the absence of CSI at the transmitter. The ergodic capacity was further examined for multiple-antenna fading channels in consideration of the spatial fading correlation and rank deficiency of the channel in \cite{shin2003capacity}. It revealed that the spatial correlation impairs the diversity gain, while the double scattering and keyhole effects deteriorate the spatial multiplexing gains. Moreover, in \cite{hanlen2012capacity}, high- and low-power asymptotics of the ergodic capacity were got by assuming spatial correlations at both the transmit and receive sides, and asymptotic channel capacity was also given for a large number of transmitters in the presence of only receive side correlation. 

Apart from the spectral efficiency, the reliability of communications has became of ever-increasing importance in the Internet-of-Things (IoT) applications (e.g., automated transportation, industrial control and augmented/virtual reality), ultra-reliable low-latency communications (URLLC) and tactile internet \cite{avranas2018energy,bennis2018ultrareliable}. For instance, the URLLC is envisioned to support a reliability higher than 99.999\%, while the tactile Internet targets at $10^{-5}$ or even $10^{-7}$ of packet error probability. Since the outage probability is frequently used to characterize the reception reliability, the outage probability of MIMO systems also has attracted considerable attention in the literature \cite{telatar1999capacity,foschini1998limits,hochwald2004multi,
wang2004outage}. In \cite{telatar1999capacity}, the outage probability was obtained in closed-form for two special cases of MIMO systems, including single-input multiple-output (SIMO) and multiple-input single-output (MISO). In \cite{foschini1998limits}, by assuming that the multiple antennas are placed with sufficiently large spacing apart from each other to decorrelate path losses, the outage probability of the MIMO systems is approximated by using the bounds of the mutual information. The outage probability of the independent and identically distributed (i.i.d.) MIMO fading channels was obtained by approximating the mutual information as a Gaussian random variable for a large number of antennas in \cite{hochwald2004multi}. By assuming i.i.d. fading channels in \cite{wang2004outage}, the outage expression was derived in a closed integral form by means of Laplace transform, which can be evaluated numerically. %

However, the prior works in \cite{telatar1999capacity,foschini1998limits,hochwald2004multi,wang2004outage} did not take into account the correlation between antenna elements, which exists in realistic propagation environments because of mutual antenna coupling and close spacing between adjacent elements \cite{shin2008mimo}. The spatial correlation would remarkably affect the reliability of MIMO systems. In \cite{smith2003exact}, the exact outage probability was derived for MIMO systems with the number of antennas less than or equal to three. Conversely, the relationship between outage probability and outage capacity under low outage capacity was studied for large size MIMO systems in \cite{tong2008asymptotic}. In \cite{chiani2003capacity}, the outage probability was derived by using the method of characteristic function to account for the correlation at the receiver side, and fast Fourier transform was employed to enable the calculation of the outage probability. The numerical results revealed that the spatial correlation has a detrimental impact on the outage performance. Similarly to \cite{chiani2003capacity}, the method of characteristic function was adopted to obtain the mean and variance of the capacity over semi-correlated (i.e., either transmit- or receive-correlated) MIMO flat-fading channels in \cite{smith2003capacity}. The outage probability was then approximated by a Gaussian distribution. As observed from \cite{chiani2003capacity,smith2003capacity}, the characteristic function or equivalently moment-generating function (MGF) of the capacity plays a vital role in deriving the outage probability of MIMO systems. In \cite{khan2005capacity}, the character expansion method was initially introduced to give a closed-form expression for the MGF of the capacity under full-correlated Rayleigh MIMO channels if the numbers of transmit and receive antennas are identical. The same method was further extended to derive the MGF of the capacity for the case with arbitrary numbers of transmit and receive antennas in \cite{simon2006capacity}. The outage probability could then be evaluated by performing numerical inversion of Laplace transform. Moreover, the character expansion method used in \cite{khan2005capacity,simon2006capacity} was revisited in \cite{ghaderipoor2012application}. Afterwards, a unified analytical framework was founded to study the MGF of the capacity for more general MIMO fading channels. With the aid of MGF in \cite{ghaderipoor2012application}, the outage probability could be approximated with arbitrary small numerical error by favor of the Abate-Whitt method in \cite{chiang2016finite}. The proposed approximation warranted the further investigation of the diversity-multiplexing tradeoff (DMT).
Unfortunately, the outage probability of MIMO systems was obtained in the literature by relying upon either approximations or numerical inversions even under spatially independent fading channels, which hardly offer further helpful insights about the system parameters. This is due to the fact that the relevant theoretical analyses are obliged to solve intractable integrals over complex matrices, and consequently produces lots of complicated results, e.g., the involvement of hypergeometric functions of matrix arguments. The complex representations of the outage probability impede the asymptotic analysis, which is commonly conducted to gain physical insights into the quantitative effects of the spatial correlation, transmission rate, power allocation strategy, etc. To the authors' best knowledge, the asymptotic behavior of the outage probability has never been reported so far.

To address the above issues, Mellin transform is applied in the paper to derive exact and tractable representations for the outage probabilities of MIMO systems in three different Kronecker correlation channel models, i.e., independent, semi-correlated and full-correlated Rayleigh MIMO channels. The solutions are expressed in terms of the generalized Fox's H function which has been implemented in popular mathematical softwares. Upon the exact expressions, the asymptotic analysis of the outage probability in the high signal-to-noise ratio (SNR) regime is derived to enable further investigation of diversity order, modulation and coding gain, spatial correlation and power allocation strategy. The asymptotic expressions of the outage probabilities under different Kronecker correlation models follow the same structure. The unified expression demonstrates that full diversity can be achieved regardless of the presence of spatial correlation, whereas the spatial correlation does negatively influence the outage performance. Particularly, the impact of spatial correlation is quantified by carrying out the asymptotic analysis, and the qualitative relationship between the spatial correlation and the outage probability is established by virtue of the concept of majorization in \cite{shin2008mimo}. Moreover, the transmission rate affects the outage performance via the term of modulation and coding gain, and the asymptotic outage probability is proved to be an increasing and convex function of the transmission rate. 
It is also found that without CSI at the transmitter, equal power allocation is optimal to minimize the outage probability under high SNR, no matter whether the MIMO channels are correlated or not. Finally, the numerical analysis verifies the analytical results.

The remaining of the paper is outlined as follows. In Section \ref{sec:sys_mod}, the MIMO system model is presented and the outage probability is formulated. By using Mellin transform, we then derive the exact and asymptotic expressions for the outage probabilities under three different Kronecker correlation models in Sections \ref{sec:ana}. The asymptotic results are thoroughly investigated to reveal more physical insights in Section \ref{sec:asym_res}. In Section \ref{sec:num}, the numerical analysis is conducted for verification purposes. Lastly, Section \ref{sec:con} summarizes our work with concluding remarks.


\emph{Notations:} We shall use the following notations throughout the paper. Bold uppercase and lowercase letters are used to denote matrices and vectors, respectively. ${\bf A}^{\mathrm{T}}$, ${\bf A}^{\mathrm{H}}$, ${\bf A}^{-1}$ and ${\bf A}^{1/2}$ denote the transpose, conjugate transpose, matrix inverse and Hermitian square root of matrix ${\bf A}$, respectively. $\mathrm{vec}$, $\rm{tr}$, $\mathrm{det}$ and $\mathrm{diag}$ are the operators of vectorization, trace, determinant and diagonalization respectively. $\Delta \left( {\bf A} \right)$ refers to the Vandermonde determinant of the eigenvalues of matrix ${\bf A}$. $\mathbf{0}_n$ and $\mathbf{I}_n$ stand for $1 \times n$ all-zero vector and $n \times n$ identity matrix, respectively. $\mathbb C$ and $\mathbb C^{m\times n}$ denote the sets of complex numbers and $m\times n$-dimensional complex matrices, respectively. The symbol ${\rm i}=\sqrt{-1}$ is the imaginary unit. $\Re{(s)}$ denotes the real part of the complex number $s$. $o(\cdot)$ denotes little-O notation. $(\cdot)_n$ represents Pochhammer symbol. $|S|$ refers to the cardinality of set $S$. Any other notations will be defined in the place where they occur.

\section{System Model}\label{sec:sys_mod}
By considering a point-to-point MIMO system with $N_t$ transmit and ${N_r}$ receive antennas, the received signal vector $\mathbf{y}\in\mathbb{C}^{{N_r}\times 1}$ is written as
\begin{equation}\label{eq:pp_sig}
\mathbf{y}=\sqrt{\frac{P}{{N_t}}}\mathbf{H}\mathbf{x}+\mathbf{n},
\end{equation}
where $\mathbf{H}\in \mathbb{C}^{{N_r}\times {N_t}}$ is the matrix of the channel coefficients, $\mathbf{x}\in \mathbb{C}^{{N_t}\times 1}$ is the vector of transmitted signals, $\mathbf{n}\in \mathbb{C}^{{N_r}\times 1}$ is the complex-valued additive white Gaussian noise vector with zero mean and covariance matrix $\sigma^{2}\mathbf{I}_{N_r}$, and $P$ is the average total transmitted power. By using a Gaussian codebook, each entry of $\mathbf{x}$ is drawn randomly and independently from standard complex normal distribution. Moreover, in order to account for the effect of the antenna correlation, the channel matrix $\mathbf{H}$ is modeled herein according to the Kronecker correlation structure, which separates the spatial correlation into two independent constituent components, i.e., transmit and receive correlations \cite{martin2004asymptotic}. The Kronecker model turns out to be valid irrespective of antenna configurations and intra-array spacings if the transmitter and receiver have independent angular power profiles\cite{clerckx2013mimo}. 
Following the Kronecker model, the channel matrix reads as\cite{larsson2008space,paulraj2003introduction}
\begin{equation}\label{kron_mod}
\mathbf{H}={\mathbf{R}_r}^{{1}/{2}}{\bf H}_w{\mathbf{R}_{t}}^{{1}/{2}},
\end{equation}
where $\mathbf{H}_w\in \mathbb{C}^{{N_r}\times {N_t}}$ is a random matrix whose entries are independent and identically distributed (i.i.d.), complex circularly symmetric Gaussian random variables, i.e., $\mathrm{vec}\left(\mathbf{H}_w\right) \sim \mathcal{CN}\left(\mathbf 0_{{N_t}{N_r}},\mathbf I_{N_t}\otimes \mathbf I_{N_r}\right)$, $\mathbf{R}_t$ and $\mathbf{R}_{r}$ are respectively termed as the transmit and receive correlation matrices, and both of them are positive semi-definite Hermitian matrices. In addition, \eqref{kron_mod} implies that the vectorized channel matrix $\mathrm{vec}\left(\mathbf{H}\right) = {\left({\mathbf R_{t}}^{\rm T}\otimes \mathbf R_{r}\right)}^{1/2}\mathrm{vec}\left(\mathbf{H}_w\right)$ still obeys a complex multivariate Gaussian distribution with mean vector $\mathbf 0_{{N_t}{N_r}}$ and covariance matrix ${\mathbf R_{t}}^{\rm T}\otimes \mathbf R_{r}$. The validity of the Kronecker model has been corroborated through realistic electromagnetic field measurements\cite{kermoal2002stochastic,yu2001second}. For the sake of simplicity, we assume that the correlation matrices follow the constraints as ${\rm tr}{\left(\mathbf R_{t}\right)}=N_t$ and ${\rm tr}{\left(\mathbf R_{r}\right)}=N_r$. Considering whether $\mathbf{R}_t$ and/or $\mathbf{R}_{r}$ is identity matrix, the spatially correlated channel model of MIMO system in \eqref{kron_mod} can further be divided into three cases, i.e., 1) Independent MIMO channels: both $\mathbf{R}_t=\mathbf I_{N_t}$ and $\mathbf{R}_{r}=\mathbf I_{N_r}$; 2) Semi-correlated MIMO channels: either $\mathbf{R}_t=\mathbf I_{N_t}$ or $\mathbf{R}_{r}=\mathbf I_{N_r}$; 3) Full-correlated MIMO channels: neither $\mathbf{R}_t=\mathbf I_{N_t}$ nor $\mathbf{R}_{r}=\mathbf I_{N_r}$.


By assuming perfect knowledge of the CSI at the receiver, the mutual information capacity of MIMO system can be expressed as \cite{simon2006capacity}
\begin{equation}\label{eqn:mul_cap}
\mathcal{I}({\bf x};{\bf y}|{\bf H})=\mathrm{log}_{2}\mathrm{det}\left(\mathbf{I}_{{N_r}} + {\rho}\mathbf{H}\mathbf{H}^\mathrm{H}\right),
\end{equation}
where $\rho ={P}/{(\sigma^{2}N_t)}$ stands for the average transmit SNR per antenna. By implementing random coding with long codewords at the transmitter and typical set decoding at the receiver, the error probability of decoding a packet can be approximated by using the outage probability, which is one of the most concerned performance metrics, especially in the absence of the CSI at the transmitter. From the perspective of information theory, the outage event of the MIMO system described by \eqref{eq:pp_sig} will take place if the mutual information capacity is less than the transmission rate, i.e., $\mathcal{I}({\bf x};{\bf y}|{\bf H}) < R$, where $R$ is the transmission rate. Accordingly, the corresponding outage probability is explicitly obtained on the basis of \eqref{eqn:mul_cap} as
\begin{equation}\label{eqn:out_pro}
{p_{out}}= \Pr \left( {{{\log }_2}\det \left( {{{\bf{I}}_{N_r}} + {\rho}{\bf{H}}{{\bf{H}}^{\rm{H}}}} \right) < R} \right).
\end{equation}
Since ${{\bf{H}}}{\bf{H}}^{\rm{H}}$ has $N_r$ eigenvalues, the unordered eigenvalues of ${{\bf{H}}}{\bf{H}}^{\rm{H}}$ are denoted by ${\bs{\lambda }} = (\lambda_1,\cdots,\lambda_{N_r})$\footnote{It is worth noting that ${{\bf{H}}}{\bf{H}}^{\rm{H}}$ is a singular matrix and has at least $(N_r-N_t)$ zero eigenvalues if ${N_t}<{N_r}$.}. The outage probability in (\ref{eqn:out_pro}) can be rewritten as
\begin{equation}\label{eqn:out_def}
{p_{out}} = \Pr \left( { \underbrace {\prod\limits_{i = 1}^{{N_r}} {\left( {1 + {\rho}{\lambda _i}} \right)} }_{ \triangleq  G} < {2^R}} \right) = F_G(2^R),
\end{equation}
where $F_G(x)$ denotes the cumulative distribution function (CDF) of $G$. From \eqref{eqn:out_def}, it boils down to determining the distribution of the product of multiple shifted eigenvalues $\lambda_1,\cdots,\lambda_{N_r}$. However, the eigenvalues are correlated no matter whether the correlations among antennas present or not. The occurrence of the correlation among eigenvalues will yield the involvement of a multi-fold integral in deriving the expression of $F_G(2^R)$, which challenges the subsequent outage analysis considerably. Since there is no versatile expression for the joint PDF ${f_{\bs{\lambda }}}\left( {{\lambda _1}, \cdots ,{\lambda _{N_r}}} \right)$ under the aforementioned three different correlation models, their outage analyses should be undertaken individually to comprehensively understand the outage behaviour of spatial correlation across antennas at both the transmitter and receiver.

\section{Analysis of Outage Probability}\label{sec:ana}
The outage probability is the fundamental performance metric to characterize the error performance of decodings. However, the correlation among eigenvalues complicates the exact analysis of the outage probability for MIMO fading channels especially in the presence of the spatial correlation. Nonetheless, the product form of $G$ motivates us to apply Mellin transform to obtain the distribution of $G$ \cite{shi2017asymptotic}. Specifically, the Mellin transform of the probability density function (PDF) of $G$, $\left\{ {\mathcal M f_G} \right\}\left( s \right)$, is given by
\begin{align}\label{eqn:mellin_G}
\left\{ {\mathcal M f_G} \right\}\left( s \right) &= {\mathbb E}\left( G ^{s-1} \right) = \int\nolimits_0^\infty  {{x^{s - 1}}{f_G}\left( x \right)dx}\notag\\
& =\int\nolimits_0^\infty  { \cdots \int\nolimits_0^\infty  {\prod\limits_{i = 1}^{N_r} {{{\left( {1 + {\rho}{\lambda _i}} \right)}^{s - 1}}{f_{\bs{\lambda }}}\left( {{\lambda _1},\cdots,{\lambda _{N_r}}}\right)d\lambda_1  \cdots d{\lambda _{N_r}}} } } \triangleq \varphi(s),
\end{align}
where ${f_{\bs{\lambda }}}\left( {{\lambda _1}, \cdots ,{\lambda _{N_r}}} \right)$ is defined as the joint PDF of ${\bs{\lambda }}$. By utilizing the inverse Mellin transform together with its associated property of integration  \cite[eq.(8.3.15)]{debnath2010integral}, the CDF of $G$ can be obtained as
\begin{equation}\label{eqn:cdf_g_inverse}
{F_G}\left( x \right) = \left\{{\mathcal M} ^{ - 1}{\left[ - \frac{1}{s}\varphi \left( {s + 1} \right)\right]}\right\}\left( x \right)
 = \frac{{  1}}{{2\pi {\rm{i}}}}\int\nolimits_{c - {\rm i}\infty }^{c + {\rm i}\infty } {\frac{{{x^{ - s}}}}{-s}\varphi \left( {s + 1} \right)ds},
\end{equation}
where $c \in (-\infty,0)$, because the Mellin transform of ${F_G}\left( x \right)$ exists for any complex number $s$ in the fundamental strip $-\infty <\Re{(s)} < 0$ by noticing ${F_G}\left( x \right)=0$ for $x < 1$ and $\lim\nolimits_{x\to \infty}{F_G}\left( x \right)=1$ \cite[p400]{szpankowski2010average}.

Due to the possible existence of transmit and receive antenna correlations, the performance of the correlated MIMO channels are thoroughly investigated by considering the following three scenarios, i.e., spatial independence at both the transmit and receive sides, spatial correlation at either transmitter or receiver side, spatial correlation at both transmit and receive sides. As mentioned above, these scenarios refer to independent, semi-correlated and full-correlated MIMO channels, respectively. We hereafter derive the exact outage probabilities for the three different channel models separately inasmuch as their corresponding joint PDFs of the eigenvalues cannot be generalized in a unified fashion. Moreover, the analytical results also lay a basis for the asymptotic analysis of the outage probability in the high SNR regime.
\subsection{Independent Rayleigh MIMO Channels}\label{sec:ind}
If the links between transmit and receive antennas experience independent fading channels, i.e., $\mathbf{R}_t=\mathbf I_{N_t}$ and $\mathbf{R}_{r}=\mathbf I_{N_r}$, the Kronecker channel model in \eqref{kron_mod} collapses to $\mathbf{H}={\bf H}_w$. Accordingly, all the entries of $\mathbf{H}$ are i.i.d. complex Gaussian random variables with zero mean and unit variance. Thus, $\mathbf{H}\mathbf{H}^\mathrm{H}$ complies with the complex Wishart distribution as $\mathbf{H}\mathbf{H}^\mathrm{H} \sim \mathcal W_{N_r}(N_t,{\bf I}_{N_r})$\cite{couillet2011random}. We stipulate herein that $N_t \ge N_r$ without loss of generality, and the similar results can be obtained for the case of $N_t < N_r$ by following the same procedure.
\subsubsection{Exact Outage Probability}
The Mellin transform is applied to derive the exact expression for the outage probability if $\mathbf{H}\mathbf{H}^\mathrm{H} \sim \mathcal W_{N_r}(N_t,{\bf I}_{N_r})$ and $N_t \ge N_r$. To begin with, the joint PDF of unordered strictly positive eigenvalues of the Wishart matrix $\mathbf{H}\mathbf{H}^\mathrm{H}$, ${f_{\bs{\lambda }}}\left( {{\lambda _1}, \cdots ,{\lambda _{N_r}}} \right)$, is given by \cite[Theorem 2.17]{antonia2004random}, \cite[Corollary 3.2.19]{muirhead2008aspects} as 
\begin{equation}\label{eqn:joint_pdf_lambdav}
{f_{\bs{\lambda }}}\left( {{\lambda _1}, \cdots ,{\lambda _{N_r}}} \right) = \frac{1}{{N_r}!} {e^{ - \sum\limits_{i = 1}^{N_r} {{\lambda _i}} }}\prod\limits_{i = 1}^{N_r} {\frac{{{\lambda _i}^{{N_t} - {N_r}}}}{{\left( {{N_r} - i} \right)!\left( {{N_t} - i} \right)!}}} \Delta {\left( \bs \Lambda  \right)^2},
\end{equation}
where $\bs \Lambda = {\rm diag}(\lambda_1,\cdots,\lambda_{N_r})$.

By combining (\ref{eqn:mellin_G}) with (\ref{eqn:joint_pdf_lambdav}), it can be proved in Appendix \ref{app:mellin_g_m} that the Mellin transform of the PDF of $G$, $\varphi_{\rm ind}(s)$, is given by
\begin{align}\label{eqn:mellin_g_m}
\varphi_{\rm ind}(s) 
& = \frac{1}{{N_r}!} \sum\limits_{{{\bs\sigma _1},{\bs\sigma _2} \in {S_{N_r}}}}\frac{{{\mathop{\rm sgn}} \left( {{\bs\sigma _1}} \right){\mathop{\rm sgn}} \left( {{\bs\sigma _2}} \right)}} {{{\rho} }^{{N_t}{N_r}}}\prod\limits_{i = 1}^{N_r} {\frac{{\Gamma \left( {\tau+{\sum\nolimits_{l = 1}^2 {{\sigma _{l,i}}} } } \right)}}{{\left( {{N_r} - i} \right)!\left( {{N_t} - i} \right)!}}}  \notag\\
&\quad \times \Psi \left( {{\tau+{\sum\limits_{l = 1}^2 {{\sigma _{l,i}}} }  },s + \tau+ {{\sum\limits_{l = 1}^2 {{\sigma _{l,i}}} }};{\rho^{-1}}} \right),
\end{align}
where $\tau =|{N_t}-{N_r}|-1$ \footnote{The absolute value is adopted for the purpose of further extensions by incorporating the cases of ${N_t} < {N_r}$ as well.}, $\Psi \left( \cdot,\cdot;\cdot \right)$ denotes Tricomi's confluent hypergeometric function \cite[eq. (9.211.4)]{gradshteyn1965table}, $S_{N_r}$ denotes the set of permutations of $\left\{1,2,\cdots,{N_r}\right\}$, $\bs \sigma_l \triangleq (\sigma_{l,1},\cdots,\sigma_{l,{N_r}})$ for $ l\in\left(1,2\right)$ and ${\rm sgn}(\bs \sigma_l)$ denotes the signature of the permutation $\bs \sigma_l$, and ${\rm sgn}(\bs \sigma_l)$ is $1$ whenever the minimum number of transpositions necessary to reorder $\bs \sigma_l$ as $(1,2,\cdots,{N_r})$ is even, and $-1$ otherwise.

By using the following useful lemma, \eqref{eqn:mellin_g_m} can be further simplified to a single-fold summation.

\begin{lemma}\label{the:leb_for}
If $\eta ({\bs\sigma _1},{\bs\sigma _2})$ is a function of $\bs\sigma _1$ and ${\bs\sigma _2}$ irrespective of the ordering of the elements in the permutations of the set of two-tuples $\{(\sigma _{1,l},\sigma _{2,l}):l\in[1,N_r]\}$, the summation of ${{\mathop{\rm sgn}} \left( {{\bs\sigma _1}} \right){\mathop{\rm sgn}} \left( {{\bs\sigma _2}} \right)}\eta({\bs\sigma _1},{\bs\sigma _2})$ over all permutations of $\bs\sigma _1$ and $\bs\sigma _2$ degenerates to
\begin{align}\label{eqn:leb_for_gen}
\sum\limits_{{{\bs\sigma _1},{\bs\sigma _2} \in {S_{N_r}}}} {{\mathop{\rm sgn}} \left( {{\bs\sigma _1}} \right){\mathop{\rm sgn}} \left( {{\bs\sigma _2}} \right)}\eta({\bs\sigma _1},{\bs\sigma _2}) &= N_r!\sum\limits_{{{\bs\sigma } \in {S_{N_r}}}} {{\mathop{\rm sgn}} \left( {{\bs\sigma }} \right)}\eta({\bs\sigma },\bar{\bs\sigma})\notag\\
&= N_r!\sum\limits_{{{\bs\sigma } \in {S_{N_r}}}} {{\mathop{\rm sgn}} \left( {{\bs\sigma }} \right)}\eta(\bar{\bs\sigma},{\bs\sigma }),
\end{align}
where ${\bs\sigma}=(\sigma_1,\cdots,\sigma_{N_r})$ and $\bar{\bs\sigma}=(1,\cdots,N_r)$.
\end{lemma}
\begin{proof}
Please refer to Appendix \ref{app:proof_leb_for}.
\end{proof}

Hence, \eqref{eqn:mellin_g_m} can be further rewritten by favor of Lemma \ref{the:leb_for} as
\begin{align}\label{eqn:mellin_g_m_rew}
\varphi_{\rm ind}(s) & = \sum\limits_{{{\bs\sigma } \in {S_{N_r}}}}\frac{{{\mathop{\rm sgn}} \left( {{\bs\sigma }} \right)}} {{{\rho} }^{{N_t}{N_r}}}\prod\limits_{i = 1}^{N_r} {\frac{{\Gamma \left( {\tau+i+ {{\sigma _{i}}}  } \right)}}{{\left( {{N_r} - i} \right)!\left( {{N_t} - i} \right)!}}}   \Psi \left( {{\tau+i+ {{\sigma _{i}}}  },s + \tau+ i+ {{\sigma _{i}}};{\rho^{-1}}} \right).
\end{align}

By substituting \eqref{eqn:mellin_g_m_rew} into \eqref{eqn:cdf_g_inverse}, the CDF of $G$ can then be obtained as shown in the following theorem.
\begin{theorem} \label{the:cdf_g_ind}
The CDF of $G$ can be represented in terms of the generalized Fox's H function as
\begin{align}\label{eqn:cdf_ML}
{F_G^{(1)}}\left( x \right)  &= \sum\limits_{{{\bs\sigma } \in {S_{N_r}}}}{{{\mathop{\rm sgn}} \left( {{\bs\sigma }} \right)}} \prod\limits_{i = 1}^{N_r} {\frac{{\Gamma \left( {\tau+i+ {{\sigma _{i}}}  } \right)}}{{\left( {{N_r} - i} \right)!\left( {{N_t} - i} \right)!}}} \notag\\
&\quad \times\underbrace {Y_{1,{N_r} + 1}^{{N_r},1}\left[ {\left. {\begin{array}{*{20}{c}}
{\left( {1,1,0,1} \right)}\\
{\left( {1,1,\rho^{-1},\tau + i+ {{\sigma _{i}}}} \right)_{i=1,\cdots,{N_r}},\left( {0,1,0,1} \right)}
\end{array}} \right|\frac{x}{{\rho}^{N_r}}} \right]}_{\mathcal Y_{ {{\bs\sigma}}}^{(1)}(x)},
\end{align}
where the generalized Fox's H function is defined by using the integral of Mellin-Branes type as \cite{chelli2013performance,yilmaz2010outage}
\begin{multline}\label{eqn:fox_H_mellin}
Y_{p,q}^{m,n}\left[ {\left. {\begin{array}{*{20}{c}}
{\left( {{a_1},{\alpha _1},{A_1},{\varphi _1}} \right), \cdots ,\left( {{a_p},{\alpha _p},{A_p},{\varphi _p}} \right)}\\
{\left( {{b_1},{\beta _1},{B_1},{\phi _1}} \right), \cdots ,\left( {{b_q},{\beta _q},{B_q},{\phi _q}} \right)}
\end{array}} \right|x} \right] \\
= \frac{1}{{2\pi {\rm{i}}}}\int_{\cal L} {\frac{{\prod\nolimits_{j = 1}^m {\Xi \left( {{b_j},{\beta _j},{B_j},{\phi _j}} \right)} \prod\nolimits_{i = 1}^n {\Xi \left( {1 - {a_i}, - {\alpha _i},{A_i},{\varphi _i}} \right)} }}{{\prod\nolimits_{i = n + 1}^p {\Xi \left( {{a_i},{\alpha _i},{A_i},{\varphi _i}} \right)} \prod\nolimits_{j = m + 1}^q {\Xi \left( {1 - {b_j}, - {\beta _j},{B_j},{\phi _j}} \right)} }}{x^{ - s}}ds},
\end{multline}
where an efficient MATHEMATICA implementation of \eqref{eqn:fox_H_mellin} has been provided in \cite{chelli2013performance} and $\Xi \left( {a,\alpha ,A,\varphi } \right) = {A^{\varphi  + a + \alpha s - 1}}\Psi \left( {\varphi ,\varphi  + a + \alpha s;A} \right)$.
\end{theorem}
\begin{proof}
The proof is given in Appendix \ref{app:cdf_g_ind}.
\end{proof}
As a consequence, substituting \eqref{eqn:cdf_ML} into \eqref{eqn:out_def} yields a closed-form expression of outage probability for independent Rayleigh MIMO channels as
\begin{align}\label{out_iid}
p_{out}^{\mathrm{ind}}  = {F_G^{(1)}}\left( 2^R \right) =  \sum\limits_{{{\bs\sigma } \in {S_{N_r}}}}{{{\mathop{\rm sgn}} \left( {{\bs\sigma }} \right)}} \prod\limits_{i = 1}^{N_r} {\frac{{\Gamma \left( {\tau+i+ {{\sigma _{i}}}  } \right)}}{{\left( {{N_r} - i} \right)!\left( {{N_t} - i} \right)!}}{\mathcal Y_{ {{\bs\sigma}}}^{(1)}(2^R)}}.
\end{align}
Although the outage probability can be expressed in a compact form in \eqref{out_iid}, the involved generalized Fox's H function is too complex to extract insightful results. In order to obtain tractable results and gain more insights, we have to recourse to the asymptotic analysis of the outage probability at high SNR.
\subsubsection{Asymptotic Outage Probability}
In order to derive the asymptotic outage probability in the high SNR region, i.e., $\rho \to \infty$, the following lemma is introduced first.
\begin{lemma}\label{the:asy_out_ind}
As $\rho \to \infty$, the generalized Fox's H function ${\mathcal Y_{ {{\bs\sigma}}}^{(1)}(x)}$ follows the asymptotic expansion as
\begin{equation}\label{eqn:h_fin_asy}
{\mathcal Y_{ {{\bs\sigma}}}^{(1)}(x)}={{\rho} ^{-{N_t}{N_r}}}\underbrace {G_{{N_r} + 1,{N_r} + 1}^{0,{N_r} + 1}\left( {\left. {\begin{array}{*{20}{c}}
{1,{\left(1 + \tau + i+ {{\sigma _{i}}}\right)}_{i=1,\cdots,{N_r}}}\\
{\underbrace {1, \cdots, 1}_{N_r},0}
\end{array}} \right|x} \right)}_{{\mathfrak g_{{\bs{\sigma}}}}\left( x \right)} +o\left( {{\rho ^{ - {N_t}{N_r}}}} \right),
\end{equation}
where $G^{m,n}_{p,q}(\cdot)$ denotes the Meijer G-function\cite{gradshteyn1965table}, and ${\mathfrak g_{{\bs{\sigma}}}}\left( x \right)$ can be easily evaluated by using \cite[eq.(5.2.21)]{bateman1954tables} because ${\mathfrak g_{{\bs{\sigma}}}}\left( x \right)$ can be expressed in terms of inverse Laplace transform as
\begin{equation}\label{eqn:g2R_laplace}
{\mathfrak g_{{\bs{\sigma}}}}\left( x \right) = \frac{1}{{2\pi {\rm{i}}}}\int\nolimits_{-c - {\rm i}\infty }^{-c + {\rm i}\infty } {\frac{{e^{s\ln x}}}{{s\prod\nolimits_{i = 1}^{N_r} {\prod\nolimits_{t = 1}^{\tau + i+  {\sigma _i}} {\left( {s - t} \right)} } }} ds}.
\end{equation}
\end{lemma}
\begin{proof}
Please refer to Appendix \ref{app:proof_phi}.
\end{proof}

By combining \eqref{out_iid} and Lemma \ref{the:asy_out_ind}, the asymptotic outage probability under the independent Rayleigh MIMO channels is given by the following theorem.
\begin{theorem}
Under high SNR, the outage probability is asymptotically equal to
\begin{equation}\label{eqn:F_g_fin}
p_{out}^{\mathrm{ind}} =
\frac{{{\rho^{-N_t N_r}}}}{{\prod\nolimits_{i = 1}^{N_r} {\left( {{N_r} - i} \right)!\left( {{N_t} - i} \right)!} }}\sum\limits_{{{\bs\sigma } \in {S_{N_r}}}}{{{\mathop{\rm sgn}} \left( {{\bs\sigma }} \right)}} \prod\limits_{i = 1}^{N_r} {{{\Gamma \left( {\tau+i+ {{\sigma _{i}}}  } \right)}}{\mathfrak g_{{\bs{\sigma}}}}\left( 2^R \right)}+o\left( {{\rho ^{ - {N_t}{N_r}}}} \right).
\end{equation}
\end{theorem}

Similar to \eqref{out_iid} and \eqref{eqn:F_g_fin} for $N_t \ge N_r$, the exact and asymptotic results of the outage probability can be obtained for the case of $N_t < N_r$, because the outage probability given by \eqref{eqn:out_pro} can be rewritten by using the property \cite[Exercise 7.25, p167]{abadir2005matrix} as
\begin{equation}\label{eqn:out_prob_inter_ch}
{p_{out}}= \Pr \left( {{{\log }_2}\det \left( {{{\bf{I}}_{N_t}} + {\rho}{{\bf{H}}^{\rm{H}}}{\bf{H}}} \right) < R} \right),
\end{equation}
where $\mathbf{H}^\mathrm{H}\mathbf{H}$ is still a complex Wishart matrix, i.e., $\mathbf{H}^\mathrm{H}\mathbf{H} \sim \mathcal W_{N_t}(N_r,{\bf I}_{N_t})$. Besides, the joint PDF of the unordered eigenvalues of $\mathbf{H}^\mathrm{H}\mathbf{H}$ is similar to \eqref{eqn:joint_pdf_lambdav} by straightforward interchanging $N_t$ and $N_r$. Hence, the same approach can be used to derive closed-form expressions for the corresponding exact and asymptotic outage probabilities. The details as well as the results are omitted here to avoid redundancy. Moreover, we leave the discussion with regard to the useful insights of the asymptotic results to the next section.

%

\subsection{Semi-Correlated Rayleigh MIMO Channels}
The spatial correlation at only one side (either the transmitter or the receiver) is commonly characterized by semi-correlated Rayleigh MIMO channels\cite{kang2003impact}. Specifically, the Kronecker models in \eqref{kron_mod} for the transmit and the receive correlations respectively reduce to $\mathbf{H}={\bf H}_w{{\bf R}_t}^{1/2}$ and $\mathbf{H}={{\bf R}_r}^{1/2}{\bf H}_w$. Nevertheless, thanks to the property \cite[Exercise 7.25, p167]{abadir2005matrix} as \eqref{eqn:out_prob_inter_ch}, the outage probabilities for these two models can be tackled in exactly the same manner. Without loss of generality, we take the receive-correlated Kronecker model, i.e., $\mathbf{H}={{\bf R}_r}^{1/2}{\bf H}_w$, as an example. This implies that the channel matrix $\bf H$ has $N_t$ i.i.d. columns, each with mean zero and covariance ${\bf R}_r$. Hence, $\mathbf{H}\mathbf{H}^\mathrm{H}$ is by definition a Wishart matrix, i.e., $\mathbf{H}\mathbf{H}^\mathrm{H} \sim \mathcal W_{N_r}(N_t,{\bf R}_{r})$\cite{couillet2011random}.

\subsubsection{Exact Outage Probability}
To proceed, the joint distribution of unorder positive eigenvalues of ${\bf{H}}{\bf{H}}^{\rm{H}}$ should be determined first. However, the eigenvalues of the Wishart matrix ${\bf{H}}{\bf{H}}^{\rm{H}}$ follow different joint distributions for $N_t \ge N_r$ and $N_t < N_r$, because the Wishart matrices are nonsingular and singular for the two cases, respectively.
Specifically, if $N_t \ge N_r$, the joint distribution of the unordered strictly positive eigenvalues of ${\bf{H}}{\bf{H}}^{\rm{H}}$ equals \cite{james1964distributions,couillet2011random}
\begin{equation}\label{eqn:joint_pdf_eig_wishart}
{f_{\bs \lambda} }\left( {{\lambda _1}, \cdots ,{\lambda _{N_r}}} \right) = \frac{{\left( { - 1} \right)^{\frac{{N_r}\left( {{N_r} - 1} \right)}{2}}}{\det \left( {{{\left\{ {{e^{ - \frac{{{\lambda _i}}}{{{r_j}}}}}} \right\}}_{1 \le i,j \le {N_r}}}} \right)}{{\Delta \left( {\bf{\Lambda }} \right)}}}{{{N_r}!}{\det \left( {{{\mathbf{R}_r}^{N_t}}} \right)}{{\Delta \left( {{{{\mathbf{R}_r}^{-1}}}} \right)}}}\prod\limits_{j = 1}^{N_r} {\frac{{{\lambda _j}^{{N_t} - {N_r}}}}{{\left( {{N_t} - j} \right)!}}},\,N_t \ge N_r,
\end{equation}
where ${{r_1}} > \cdots > {{r_{N_r}}}>0$ denote the eigenvalues of ${\bf R}_r$. Whereas if $N_t < N_r$, ${\bf{H}}{\bf{H}}^{\rm{H}}$ is not a full rank matrix and has at least $N_r-N_t$ zero eigenvalues. The rest $N_t$ strictly positive eigenvalues are defined as $\tilde{\bs\lambda} = (\lambda_1,\cdots,\lambda_{N_t})$, and the joint distribution of $\tilde{\bs\lambda}$ is given by \cite{GAO2000155,antonia2004random}
\begin{equation}\label{eqn:joint_pdf_eig_semi_small}
{f_{\tilde {\bs\lambda} }}\left( {{\lambda _1}, \cdots ,{\lambda _{{N_t}}}} \right) = \frac{{{{\left( { - 1} \right)}^{{N_t}\left( {{N_r} - {N_t}} \right)}}\Delta \left( {{\bf{\tilde \Lambda }}} \right)}}{{\prod\nolimits_{j = 1}^{{N_t}} {j!} \Delta \left( {{{\bf{R}}_r}} \right)}}\det \left( \begin{array}{l}
{\left\{ {{r_j}^{{N_r} - {N_t} - 1}{e^{ - \frac{{{\lambda _i}}}{{{r_j}}}}}} \right\}_{\scriptstyle1 \le i \le {N_t}\hfill\atop
\scriptstyle1 \le j \le {N_r}\hfill}}\\
{\left\{ {{r_j}^{i - {N_t} - 1}} \right\}_{\scriptstyle{N_t} + 1 \le i \le {N_r}\hfill\atop
\scriptstyle1 \le j \le {N_r}\hfill}}
\end{array} \right),{N_t} < {N_r},
\end{equation}
where $\tilde {\bs \Lambda} = {\rm diag}(\lambda_1,\cdots,\lambda_{N_t})$. Due to the different forms of the joint eigenvalue distributions for these two cases, their outage probabilities should be derived separately. Following the same steps as \eqref{eqn:mellin_g_m}-\eqref{eqn:mellin_g_m_rew}, the Mellin transform of $G$ under semi-correlated Rayleigh MIMO channels, $\varphi_{\rm semi}(s)$, can be obtained as
\begin{subequations}\label{eqn:mellin_t_corr_gpdf_sumre}
\label{eq:main}
\begin{align}\label{eqn:mellin_t_corr_gpdf_sumre1}
 \varphi_{\rm semi}^{\ge}(s) &= \frac{{{{\left( { - 1} \right)}^{\frac{1}{2}{N_r}\left( {{N_r} - 1} \right)}}}{ \rho ^{-\frac{1}{2}{{{N_r}\left( {{N_r} + 1} \right)}} - \left( {N_t - {N_r}} \right){N_r}}}}{{\det \left( {{{{\bf{R}}_r}^{N_t}}} \right)\Delta \left( {{{{\bf{R}}_r}^{ - 1}}} \right)\prod\nolimits_{j = 1}^{N_r} {\left( {N_t - j} \right)!} }}\sum\limits_{{\bs\sigma} \in {S_{N_r}}} { {{\rm{sgn}}\left( {{\bs\sigma }} \right)} } \notag\\
&\quad\times\prod\limits_{i = 1}^{N_r} {\Gamma \left( {i + \tau  + 1} \right)\Psi \left( {i + \tau  + 1,s + i + \tau  + 1;\frac{1}{{\rho {r_{{\sigma _{i}}}}}}} \right)},\,N_t \ge N_r,
\end{align}
\begin{align}\label{eq:semi_me_small_conf_fin}
\varphi _{{\rm{semi}}}^ < (s) &= \frac{{{{\left( { - 1} \right)}^{{N_t}\left( {{N_r} - {N_t}} \right)}}{\rho ^{ - \frac{1}{2}{N_t}\left( {{N_t} + 1} \right)}}}}{{\Delta \left( {{{\bf{R}}_r}} \right)}}\sum\limits_{{\bs{\sigma }} \in {S_{{N_r}}}}{\mathop{\rm sgn}} \left( {\bs{\sigma }} \right)\notag\\
&\quad \times  {\prod\limits_{i = 1}^{{N_t}} {{r_{{\sigma _i}}}^{{N_r} - {N_t} - 1}} \prod\limits_{i = {N_t} + 1}^{{N_r}} {{r_{{\sigma _i}}}^{i - {N_t} - 1}} } \prod\limits_{i = 1}^{{N_t}} {\Psi \left( {i,s + i,\frac{1}{{\rho {r_{{\sigma _i}}}}}} \right)},\,N_t < N_r,
\end{align}
\end{subequations}
where the proofs are detailed in Appendix \ref{app:mellin_semi}.

\begin{theorem}\label{the:fg_semi}
The CDFs of $G$ for the cases of $N_t \ge N_r$ and $N_t < N_r$ are respectively given by
\begin{subequations}\label{eqn:mellin_t_corr_gpdf_sumsmall}
\begin{align}\label{eqn:mellin_t_corr_gpdf_sum_conf_recon}
 &F_G^{(2\ge)}(x) = \frac{{{{\left( { - 1} \right)}^{\frac{1}{2}{N_r}\left( {{N_r} - 1} \right)}}}\prod\nolimits_{i = 1}^{N_r} {\Gamma \left( {i + \tau  + 1} \right)}}{{\det \left( {{{{\bf{R}}_r}^{N_t}}} \right)\Delta \left( {{{{\bf{R}}_r}^{ - 1}}} \right)\prod\nolimits_{j = 1}^{N_r} {\left( {N_t - j} \right)!} }}\sum\limits_{{\bs\sigma} \in {S_{N_r}}} { {{\rm{sgn}}\left( {{\bs\sigma }} \right)} }\prod\limits_{i = 1}^{N_r} {{r_{{\sigma _{i}}}}^{i + \tau + 1}}{{\mathcal Y _{\bs\sigma ,N_r,\tau }^{(2)}}\left( x \right)},
 \end{align}
 \begin{align}\label{eqn:F_G_small_conf_con}
  &{F_G^{(2<)}}\left( x \right) = \frac{{{{\left( { - 1} \right)}^{{N_t}\left( {{N_r} - {N_t}} \right)}}}}{{\Delta \left( {{{\bf{R}}_r}} \right)}}\sum\limits_{{\bs{\sigma }} \in {S_{{N_r}}}} {{\mathop{\rm sgn}} \left( {\bs{\sigma }} \right){\prod\limits_{i = 1}^{{N_t}} {{r_{{\sigma _i}}}^{{N_r} + i - {N_t} - 1}} \prod\limits_{i = {N_t} + 1}^{{N_r}} {{r_{{\sigma _i}}}^{i - {N_t} - 1}} } } {{\mathcal Y _{\bs\sigma,N_t,-1}^{(2)}}\left( x  \right)},
\end{align}
\end{subequations}
where ${{\mathcal Y _{\bs\sigma,N,\upsilon}^{(2)}}\left( x  \right)}$ is defined as
\begin{multline}\label{eqn:def_mathcalY_2}
{{\mathcal Y _{\bs\sigma,N,\upsilon}^{(2)}}\left( x  \right)} = \\{Y_{1,{N}+ 1}^{N,1}\left[ {\left. {\begin{array}{*{20}{c}}
{\left( {1,1,0,1} \right)}\\
{\left( {1,1,{{\left(\rho {r_{{\sigma _{i}}}}\right)}^{-1}},i+\upsilon+1 } \right)_{i=1,\cdots,N},\left( {0,1,0,1} \right)}
\end{array}} \right|{\frac{x}{{\left(\prod\nolimits_{i = 1}^{{N}} {{r_{{\sigma _i}}}} \right)}{\rho ^{{N}}}}}} \right]}
,\,N\le N_r.
\end{multline}
\end{theorem}
\begin{proof}
Please refer to Appendix \ref{app:fg_semi}.
\end{proof}
%

According to \eqref{eqn:out_def}, the outage probabilities of the semi-correlated Rayleigh MIMO channels for $N_t\ge N_r$ and $N_t < N_r$ are respectively given by
\begin{equation}\label{out_SC}
 p_{out}^{\mathrm{semi}}=
  \left\{ {\begin{array}{*{20}{c}}
{F_G^{(2\ge)}(2^R),}&{N_t\ge N_r,}\\
{F_G^{(2<)}(2^R),}&{N_t < N_r.}
\end{array}} \right.
\end{equation}

\subsubsection{Asymptotic Outage Probability}
Next, in order to provide further insightful results for \eqref{eqn:mellin_t_corr_gpdf_sumsmall}, the asymptotic analysis is conducted for the outage probability of semi-correlated Rayleigh MIMO channels at high SNR. At first, the following lemma regarding the asymptotic expression of the ${\cal Y}_{{\bs{\sigma }},N}^{(2)}\left( x \right)$ as $\rho\to \infty$ is presented to facilitate the next asymptotic outage analysis.
%
%

\label{lem:zera_asy}
\begin{lemma}\label{lem:zera_asy}
At high transmit SNR, i.e., $\rho\to \infty$, ${\cal Y}_{{\bs{\sigma }},N,\upsilon }^{(2)}\left( x \right)$ has the asymptotic expansion as
\begin{multline}\label{eqn:zeta_asy}
{\cal Y}_{{\bs{\sigma }},N,\upsilon }^{(2)}\left( x \right) = {\rho^{-\frac{1}{2}{{N\left( {N + 1} \right)}} - \left( {\upsilon+1} \right)N}}\prod\limits_{i = 1}^N{r_{{\sigma _{i}}}}^{ - i - \upsilon -1 }\frac{1}{{2\pi {\rm{i}}}}\int\nolimits_{ - c - {\rm i}\infty }^{ - c + {\rm i}\infty }\frac{{\Gamma \left( s \right)}}{{\Gamma \left( {1 + s} \right)}} \\
 \times{\prod\limits_{i = 1}^N {\frac{{\Gamma \left( {s - i - \upsilon  - 1} \right)}}{{\Gamma \left( s \right)}}\sum\limits_{{n_i} = 0}^\infty  {\frac{{{{\left( {i + \upsilon  + 1} \right)}_{{n_i}}}}}{{{{\left( { - s +i + \upsilon  + 2} \right)}_{{n_i}}}}}\frac{{{{\left( {{{\rho {r_{{\sigma _{i}}}}}}} \right)}^{{-n_i}}}}}{{{n_i}!}}} } {x^s}ds}
+ o\left( {{\rho ^{ - N_tN_r}}} \right),
\end{multline}
where $c$ is set as $c < \left( {N/2 + \upsilon  + 1} \right)\left( {N - 1} \right) - N_tN_r$.
\end{lemma}
\begin{proof}
Please refer to Appendix \ref{app:zeta_asy}.
\end{proof}

By using Lemma \ref{lem:zera_asy}, the asymptotic expression of the outage probability under semi-correlated Rayleigh MIMO channels is given by the following theorem.
\begin{theorem}\label{the:pout_asy_semi}
As the transmit SNR $\rho$ increases to $\infty$, the asymptotic outage probabilities of semi-correlated Rayleigh MIMO channels for $N_t\ge N_r$ and $N_t < N_r$ are generalized as
\begin{align}\label{eqn:out_single_side_corr_asy}
{p_{out}^{{\rm semi}}}&= \frac{{{{{{\rho }}}^{-{N_t}{N_r}}}}}{{\det \left( {{{{\mathbf{R}}_r}}} \right)^{N_t}\prod\nolimits_{i = 1}^{{\mathcal N}} {\left( {{N_t} - i} \right)!\left( {{N_r} - i} \right)!} }}\notag\\
&\quad \times\sum\limits_{{\bs{\sigma}} \in {S _{\mathcal N}}} {\rm{sgn}}\left( {\bs \sigma} \right) \prod\limits_{i = 1}^{{\mathcal N}} {\Gamma \left( {\tau + {i}  + \sigma_i} \right)} {\mathfrak g_{{\bs \sigma }}}\left(2^R \right) + o\left( {{\rho ^{ - {N_t}{N_r}}}} \right),
\end{align}
where ${\mathcal N=\min\{N_t,N_r\}}$ and $\tau =|{N_t}-{N_r}|-1$.
\end{theorem}
\begin{proof}
Please see the proof of \eqref{eqn:out_single_side_corr_asy} in Appendix \ref{app:proof_asy_out_sin_corr}.
\end{proof}
Moreover, the similar exact and asymptotic results can be derived for the outage probability under the transmit-correlated Kronecker model, whose results are omitted to save space. Additionally, as opposed to \eqref{eqn:F_g_fin}, it is found from \eqref{eqn:out_single_side_corr_asy} that the impact of the spatial correlation arises and is quantified by $\det \left( {{{{\mathbf{R}}_r}^{N_t}}} \right)$. We defer the further discussions regarding \eqref{eqn:out_single_side_corr_asy} to the next section.
\subsection{Full-Correlated Rayleigh MIMO Channels}\label{sec:full}


If both the transmit and receive correlations occur, full-correlated MIMO channels are modeled as \eqref{kron_mod}, i.e., $\mathbf{H}={\mathbf{R}_r}^{{1}/{2}}{\bf H}_w{\mathbf{R}_{t}}^{{1}/{2}}$. Unfortunately, $\mathbf{H}\mathbf{H}^\mathrm{H}$ is no longer a Wishart matrix under the circumstance when both the transmit and receive correlations arise. Moreover, with the property \cite[Exercise 7.25, p167]{abadir2005matrix}, the outage probabilities for $N_t \ge N_r$ and $N_t < N_r$ can be derived in the same fashion. Hence, we assume $N_t \ge N_r$ in this subsection unless otherwise specified. 

\subsubsection{Exact Outage Probability}
By favor of the character expansions, the joint distribution of the $N_r$ unordered strictly positive eigenvalues of $\mathbf{H}\mathbf{H}^\mathrm{H}$ under the full-correlated Rayleigh MIMO channels is obtained by Ghaderipoor \emph{et al.} in \cite{ghaderipoor2012application} as
\begin{equation}\label{eqn:eig_pdf_full}
{f_{\bs{\lambda }}}\left( {{\lambda _1}, \cdots ,{\lambda _{N_r}}} \right) = \sum\limits_{{{\bf{k}}_{N_r}}} {\frac{{{{\left( { - 1} \right)}^{\frac{{{N_r}\left( {{N_r} - 1} \right)}}{2}}}\mathcal A}}{{{N_r}!\Delta \left( {\bf{K}} \right)}}\Delta \left( {\bs{\lambda }} \right) \det \left( {{{\left\{ {{\lambda _i}^{{k_j} + {N_t} - {N_r}}} \right\}}_{ {i,j} }}} \right)},
\end{equation}
where $\mathcal A $ is independent of $\bs \lambda$ and is explicitly given by
\begin{align}\label{eqn:def_A_cal}
\mathcal A &= \frac{{{\prod\nolimits_{i = 1}^{N_r} {{a_i}^{N_t}} \prod\nolimits_{j = 1}^{N_t} {{b_j}^{N_r}} }}}{{{\Delta \left( {\bf{A}} \right)\Delta \left( {\bf{B}} \right)}}{\prod\nolimits_{j = 1}^{N_r} {\left( {{k_j} + {N_t} - {N_r}} \right)!} }}\notag\\
&\quad \times {\det {\left( \left\{{{\left( { - {a_i}} \right)}^{{k_j}}}\right\}_{i,j} \right)} \det \left( {{{\left\{ {{b_i}^{{k_j} + N - M}} \right\}}_{{i,1 \le j \le M}}},{{\left\{ {{b_i}^{N - j}} \right\}}_{ {i,M + 1 \le j \le N}}}} \right)},
\end{align}
${\bf a} =(a_1,\cdots,a_{N_r})$ and ${\bf b}  =(b_1,\cdots,b_{N_t})$ represent the eigenvalues of ${{\bf R}_r}^{-1}$ and ${{\bf R}_t}^{-1}$, respectively, ${\bf k}_{N_r}=(k_1,\cdots,k_{N_r})$ stands for all irreducible representation of the general linear group $\mathrm{GL}({N_r},\mathbb C)$ and $k_1\ge\cdots\ge k_{N_r}$ are integers, ${\bf A}$, $\bf B$ and $\bf K$ are the diagonalizations of vectors $\bf a$, $\bf b$ and ${\bf k}_{N_r}$, respectively.

It is proved in Appendix \ref{app:full_mellin} that the Mellin transform of $f_G(x)$ under full-correlated Rayleigh MIMO channels is expressed as
\begin{align}\label{eqn:mellin_fullc_tricomi}
\varphi_{\rm full}\left( s \right) &= \frac{{{{\left( { - 1} \right)}^{{N_r}\left( {{N_t} - {N_r}} \right)}}{{{{\rho }} }^{-\frac{1}{2}{{{N_r}\left( {{N_r} + 1} \right)}}}}\prod\nolimits_{i = 1}^{N_r} {{a_i}^{N_r}} \prod\nolimits_{j = 1}^{N_t} {{b_j}^{N_r}} }}{{\Delta \left( {\bf{A}} \right)\Delta \left( {\bf{B}} \right)\prod\nolimits_{i = 1}^{N_r} {{{\left( {s + i - 2} \right)}^{i - 1}}} }}\notag\\
&\quad \times \det \left( {\begin{array}{*{20}{c}}
{{{\left\{ {\Psi \left( {1,s + {N_r};\frac{{{a_i}{b_j}}}{\rho }} \right)} \right\}}_{1 \le i \le {N_r},j}}}\\
{{{\left\{ {{b_j}^{{N_t} - i}} \right\}}_{{N_r} + 1 \le i \le {N_t},j}}}
\end{array}} \right).
\end{align}

Accordingly, the CDF of $G$ can be obtained by using inverse Mellin transform, as shown in the following theorem.
\begin{theorem}\label{the:full_cdf}
The PDF of $G$ under full-correlated Kronecker channel model is given by
\begin{align}\label{eqn:F_G_CDF_fullcorr}
&{F_G}\left( x \right) = \frac{{{{\left( { - 1} \right)}^{{N_r}\left( {{N_t} - {N_r}} \right) + \frac{1}{2}{{{N_r}\left( {{N_r} - 1} \right)}}}}}{{{{\rho }} }^{ \frac{1}{2}{{{N_r}\left( {{N_r} - 1} \right)}}}}}{{\Delta \left( {\bf{A}} \right)\Delta \left( {\bf{B}} \right)}}\sum\limits_{{\bs{\sigma }} \in {S_{N_t}}} {{\mathop{\rm sgn}} \left( {\bs{\sigma }} \right)\prod\limits_{i = {N_r} + 1}^{N_t} {{b_{{\sigma _i}}}^{{N_t} + {N_r} - i}} } \notag\\
&\underbrace {Y_{{N_r},2{N_r}}^{{N_r},{N_r}}\left[ {\left. {\begin{array}{*{20}{c}}
{\left( {1,1,0,1} \right),{{\left( {{N_r} ,1,0,1} \right)}_{j = 1, \cdots ,{N_r}-1}}}\\
{{{\left( {{N_r},1,\frac{{{a_{{i}}}{b_{\sigma_i}}}}{\rho },1} \right)}_{i = 1, \cdots ,{N_r}}},\left( {0,1,0,1} \right),{{\left( {j,1,0,1} \right)}_{j = 1, \cdots ,{N_r}-1}}}
\end{array}} \right|x {\prod\limits_{i = 1}^{N_r} {\frac{{{a_i}{b_{{\sigma _i}}}}}{\rho }} }} \right]}_{\mathcal Y_{ {{\bs\sigma}}}^{(3)}(x)}.
\end{align}
\end{theorem}
\begin{proof}
  Please see Appendix \ref{app:full_cdf}.
\end{proof}

According to \eqref{eqn:out_def}, the outage probability under fully correlated MIMO channels can be obtained by
\begin{equation}\label{eqn:out_excfull}
p_{out}^{\rm full} = \frac{{{{\left( { - 1} \right)}^{{N_r}\left( {{N_t} - {N_r}} \right) + \frac{1}{2}{{{N_r}\left( {{N_r} - 1} \right)}}}}}{{{{\rho }} }^{ \frac{1}{2}{{{N_r}\left( {{N_r} - 1} \right)}}}}}{{\Delta \left( {\bf{A}} \right)\Delta \left( {\bf{B}} \right)}}\sum\limits_{{\bs{\sigma }} \in {S_{N_t}}} {{\mathop{\rm sgn}} \left( {\bs{\sigma }} \right)\prod\limits_{i = {N_r} + 1}^{N_t} {{b_{{\sigma _i}}}^{{N_t} + {N_r} - i}} {\mathcal Y_{ {{\bs\sigma}}}^{(3)}(2^R)}}.
\end{equation}

\subsubsection{Asymptotic Outage Probability}
In order to derive the asymptotic expression at high SNR for the outage probability in \eqref{eqn:out_excfull}, the following lemma associated with the asymptotic expression of  ${\mathcal Y_{ {{\bs\sigma}}}^{(3)}(x)}$ is developed first.
\begin{lemma}\label{the:asy_y_3}
As $\rho \to \infty$, ${\mathcal Y_{ {{\bs\sigma}}}^{(3)}(x)}$ is asymptotic to
\begin{align}\label{eqn:mellin_cdf_g_fullcsers_exp}
{\mathcal Y_{ {{\bs\sigma}}}^{(3)}(x)} = &{{{{\rho }}}^{-{N_r}^2}}\prod\limits_{i = 1}^{N_r} {{a_i}^{N_r}}{{b_{\sigma_i}}^{N_r}}\notag\\
&\times\frac{1}{{2\pi {\rm{i}}}}\int\nolimits_{-c - {\rm i}\infty }^{-c + {\rm i}\infty } {\frac{{\Gamma \left( s \right)}}{{\Gamma \left( {1 + s} \right)}} \prod\limits_{i = 1}^{N_r} {\frac{{\Gamma \left( {s - {N_r}} \right)}}{{\Gamma \left( {s - i + 1} \right)}}}  {\sum\limits_{{n_{{i}}} = 0}^\infty  {\frac{{{{\left( {\frac{{{a_{{ i}}}{b_{\sigma_i}}}}{\rho }} \right)}^{{n_{{i}}}}}}}{{{{\left( { 1 + {N_r}  - s} \right)}_{{n_{{i}}}}}}}} {x^s}ds} }
\notag\\
& + o\left( {{\rho ^{ - {N_t}{N_r}-\frac{1}{2}N_r(N_r+1)}}} \right).
\end{align}
\end{lemma}
\begin{proof}
  Please see Appendix \ref{app:asy_y_3}.
\end{proof}

By using Lemma \ref{the:asy_y_3}, the asymptotic expression of the outage probability under full-correlated Rayleigh MIMO channels is given by the following theorem.
\begin{theorem}\label{the:asy_out_full}
At high transmit SNR, the outage probability under full-correlated Rayleigh MIMO channels asymptotically equals
\begin{align}\label{eqn:F_G_mellin_trans_detzerofina}
p_{out}^{\rm full} = \frac{{{{{{\rho }} }^{-{N_r}{N_t}}}}}{{\det \left( {{{{\bf{R}}_r}}} \right)^{N_t}\det \left( {{{{\bf{R}}_t}}} \right)^{N_r}}}g_{\bf 0}(R)+ o\left( {{\rho ^{ - {N_t}{N_r}}}} \right),
\end{align}
where ${\bf 0}$ is a ${1 \times N_r}$ vector with all its elements equal to zero, $g_{\bf n}(R)$ is defined as
\begin{equation}\label{eqn:def_g_R_full}
g_{\bf n}(R) = G_{{N_r} + 1,{N_r} + 1}^{0,{N_r} + 1}\left( {\left. {\begin{array}{*{20}{c}}
{1,{N_t} + 1 + n_1, \cdots ,{N_t} + {N_r} + n_{N_r}}\\
{0,1, \cdots ,{N_r}}
\end{array}} \right|2^R} \right),
\end{equation}
and ${\bf n} = (n_1,\cdots,n_{N_r})$.
\end{theorem}
\begin{proof}
  Please see Appendix \ref{app:asy_out_full}.
\end{proof}
It is worth mentioning that the asymptotic analysis of the outage probability for $N_t < N_r$ can be carried out in an analogous way owing to the property \cite[Exercise 7.25, p167]{abadir2005matrix}. Similar to \eqref{eqn:out_excfull} and \eqref{eqn:F_G_mellin_trans_detzerofina}, the exact and asymptotic outage expressions corresponding to $N_t < N_r$ can be obtained by directly interchanging ${\bf R}_t$ and $N_t$ with ${\bf R}_r$ and $N_r$, respectively.


\section{Discussions of Asymptotic Results}\label{sec:asym_res}
Although the exact outage probabilities corresponding to the three different spatial correlation models take different forms, their asymptotic expressions are virtually consistent with each other by comparing
\eqref{eqn:F_g_fin}, \eqref{eqn:out_single_side_corr_asy} and \eqref{eqn:F_G_mellin_trans_detzerofina}. The asymptotic outage probabilities commonly exhibit the same basic mathematical structure as \cite[eq.(3.158)]{tse2005fundamentals}, \cite{jankiraman2004space}
\begin{equation}\label{eqn:gen_out_asy}
p_{out}  =  \mathcal S ({\bf R}_t,{\bf R}_r) {\left( {{\mathcal C(R)} } \rho \right)^{ - d}} + o\left( {{\rho ^{ - d}}} \right),
\end{equation}
where $\mathcal S({\bf R}_t,{\bf R}_r)$ quantifies the impact of spatial correlation at transmit and receive sides, $\mathcal C(R)$ is the modulation and coding gain, and $d$ stands for the diversity order. Moreover, the asymptotic results in (\ref{eqn:F_G_mellin_trans_detzerofina}) for full-correlated Rayleigh fading channels in fact encompass the results for independent and semi-correlated ones as special cases (${\bf R}_t = {\bf I}_{N_t}$ and/or ${\bf R}_r = {\bf I}_{N_r}$), which will be rigorously proved later on. Without loss of generality, we assume $N_t \ge N_r$ in the sequel, and the similar results also apply to the case of $N_t < N_r$. Therefore, by identifying \eqref{eqn:F_G_mellin_trans_detzerofina} with \eqref{eqn:gen_out_asy}, $\mathcal S({\bf R}_t,{\bf R}_r)$, $\mathcal C(R)$ and $d$ under full-correlated Rayleigh MIMO channels are explicitly given by
\begin{equation}\label{eqn:spa_imp_exp}
\mathcal S({\bf R}_t,{\bf R}_r) = \frac{1}{{\det \left( {{{{\bf{R}}_r}}} \right)^{N_t}\det \left( {{{{\bf{R}}_t}}} \right)^{N_r}}},
\end{equation}
\begin{equation}\label{eqn:cm_imp_exp}
\mathcal C(R) = {\left({g_{\bf{0}}}(R)\right)}^{-\frac{1}{N_tN_r}},
\end{equation}
and $d=N_tN_r$, respectively. Besides, even though \eqref{eqn:F_g_fin}, \eqref{eqn:out_single_side_corr_asy} and \eqref{eqn:F_G_mellin_trans_detzerofina} follow the same asymptotic form, this unified feature does not carry over to the exact outage probabilities. This is due to the fact that the exact analyses are carried out on the basis of the joint distribution of the unordered eigenvalues precisely, while the joint PDFs under the complex correlation models (e.g., \eqref{eqn:joint_pdf_eig_wishart}, \eqref{eqn:joint_pdf_eig_semi_small} and \eqref{eqn:eig_pdf_full}) are inapplicable to those under the relatively simple cases (e.g., \eqref{eqn:joint_pdf_lambdav}) because of the premise of the distinct eigenvalues of the correlation matrix (${\bf R}_t$ and/or ${\bf R}_r$). This further demonstrates the necessity of splitting the study of the outage probability for MIMO channels into three different scenarios.

From the asymptotic structure of the outage probability in \eqref{eqn:gen_out_asy}, the number of antennas, the transmission rate, the fading correlation affect the outage performance of the MIMO systems through diversity order $d$, the modulation and coding gain $\mathcal C(R)$, the impact factor of the spatial correlation $\mathcal S({\bf R}_t,{\bf R}_r)$, respectively. Moreover, although the uniform power allocation at transmitter is assumed at the very beginning of Section \ref{sec:sys_mod}, the similar results can be readily extended to any power allocation strategies. To comprehensively understand the asymptotic behavior of the outage probability, these impact factors are discussed individually. 




%

\subsection{Diversity Order}
The terminology of the diversity order can be used to measure the degree of freedom of communication systems, which is defined as the ratio of the outage probability to the transmit SNR on a log-log scale as
\begin{equation}\label{eqn:d_def}
 d = \mathop {\lim }\limits_{\rho  \to \infty } \frac{{\log {p_{out}}}}{{\log \rho }}.
\end{equation}
Hence, the diversity order indicates the decaying speed of the outage probability with respect to the transmit SNR. As disclosed by \eqref{eqn:F_g_fin}, \eqref{eqn:out_single_side_corr_asy} and \eqref{eqn:F_G_mellin_trans_detzerofina},
full diversity can be achieved by MIMO systems regardless of the presence of the spatial correlation, i.e., $d=N_tN_r$. Although the spatial correlation does not impair the diversity order, it does severely deteriorate the outage performance through the component $\mathcal S({\bf R}_t,{\bf R}_r)$, which will be analyzed later.




\subsection{Modulation and Coding Gain}\label{sec:mc_gain}
The modulation and coding gain $\mathcal C(R)$ quantifies the amount of the SNR reduction required to reach the same outage probability when employing a certain modulation and coding scheme (MCS). In other words, $\mathcal C(R)$ characterizes how much gain can be benefited from the adopted MCS. Accordingly, the increase of $\mathcal C(R)$ is in favor of the improvement of the outage performance. It is worthwhile to note that MCS determines the transmission rate $R$. From \eqref{eqn:cm_imp_exp}, in order to reflect the behaviour of $\mathcal C(R)$, it suffices to investigate the property of the function ${g_{\bf{0}}}(R)$. 
Although both the asymptotic expressions of the outage probabilities in (\ref{eqn:F_g_fin}) and (\ref{eqn:out_single_side_corr_asy}) are expressed in terms of summations of ${\mathfrak g_{{\bs{\sigma}}}}(2^R)$ over all the permutations ${\bs{\lambda }}$ that appear to be different from (\ref{eqn:F_G_mellin_trans_detzerofina}), the following remark confirms their consistency.
\begin{remark}\label{the:rem_g}
The integral representation of ${g_{\bf{0}}}(R)$ and the relationship between ${g_{\bf{0}}}(R)$ and ${\mathfrak g_{{\bs{\sigma}}}}(2^R)$ are shown as
\begin{align}\label{eqn:g_meijer_relation1}
{g_{\bf{0}}}(R) &= \frac{{\int\nolimits_{\prod\nolimits_{j = 1}^{N_r} {\left( {1 + {\lambda _j}} \right)}  \le 2^R} {{{\left( {\Delta \left( {\bs{\lambda }} \right)} \right)}^2}\prod\nolimits_{j = 1}^{N_r} {{\lambda _j}^{{N_t} - {N_r}}} d{\lambda _1} \cdots d{\lambda _{N_r}}} }}{{{N_r}!\prod\nolimits_{j = 1}^{N_r} {\left( {{N_t} - j} \right)!\left( {{N_r} - j} \right)!} }} \\
&= \frac{\sum\nolimits_{{{\bs\sigma } \in {S_{N_r}}}}{{{\mathop{\rm sgn}} \left( {{\bs\sigma }} \right)}} \prod\nolimits_{i = 1}^{N_r} {{{\Gamma \left( {\tau+i+ {{\sigma _{i}}}  } \right)}}{\mathfrak g_{{\bs{\sigma}}}}\left( 2^R \right)}}{{\prod\nolimits_{i = 1}^{N_r} {\left( {{N_r} - i} \right)!\left( {{N_t} - i} \right)!} }}
\label{eqn:g_meijer_relation2}.
\end{align}
\end{remark}
 \begin{proof}
Please see Appendix \ref{app:rem_g}.
\end{proof}

It should be mentioned herein that ${\mathfrak g_{{\bs{\sigma}}}}(2^R)$ is not only a monotonically increasing but also convex function with respect to the transmission rate $R$ by using \cite[Lemma 4]{shi2017asymptotic}. Fortunately, the same property also carries over to the function ${g_{\bf{0}}}(R)$, which is given by the following theorem.
\begin{theorem}\label{the:conv_g0}
${g_{\bf{0}}}(R)$ is a monotonically increasing and convex function of the transmission rate $R$.
\end{theorem}
\begin{proof}
Please see Appendix \ref{app:conv_g0}.
\end{proof}
Evidently from Theorem \ref{the:conv_g0}, the transmission rate is an increasing and convex function of the asymptotic outage probability. Without dispute, the monotonicity and convexity of ${g_{\bf{0}}}(R)$ can greatly facilitate the optimal rate selection of MIMO systems if the asymptotic results are used.

\subsection{Spatial Correlation}\label{sec:scorr}
It is found from \eqref{eqn:spa_imp_exp} that the impact factor of spatial antenna correlation, $\mathcal S({\bf R}_t,{\bf R}_r)$, depends on the determinants of the transmit and receive correlation matrices, i.e., ${\det \left( {{{\bf{R}}_r}} \right)}$ and $\det \left( {{{\bf{R}}_t}} \right)$. 
Although the effect of the spatial correlation can be quantified by $\mathcal S({\bf R}_t,{\bf R}_r)$, it is also imperative to draw a qualitative conclusion about the outage behaviour of the spatial correlation. 
To characterize the spatial correlation, the majorization theory is usually adopted as a powerful mathematical tool to establish a tractable framework \cite{shin2008mimo,feng2018impact}. The majorization-based correlation model is defined as follows.
\begin{definition}\label{Definition1}
For two $N \times N$ semidefinite positive matrices $\mathbf{R}_{1}$ and $\mathbf{R}_{2}$, ${\bf r}_1=(r_{1,1},\cdots,r_{1,{N}})$ and ${\bf r}_2=(r_{2,1},\cdots,r_{2,{N}})$ are defined as the vectors of the eigenvalues of $\mathbf{R}_{1}$ and $\mathbf{R}_{2}$, respectively, where the eigenvalues are arranged in descending order as ${r_{i,1}} \ge \cdots  \ge {r_{i,{N}}},i \in \left\{ {1,2} \right\}$. We denote $\mathbf{R}_{1} \preceq \mathbf{R}_{2}$ and say the matrix $\mathbf{R}_{1} $ is majorized by the matrix $ \mathbf{R}_{2}$ if
\begin{equation}
\sum\limits_{j = 1}^k {{r_{1,j}}}  \le \sum\limits_{j = 1}^k {{r_{2,j}}},\,\left( {k = 1,2, \cdots ,{N} - 1} \right)  \quad{\rm and}\quad \sum\limits_{j = 1}^{N_t} {{r_{1,j}}}  = \sum\limits_{j = 1}^{N_t} {{r_{2,j}}} .
\end{equation}
We also say the matrix $\mathbf{R}_{1} $ is more correlated than the matrix $ \mathbf{R}_{2}$.
\end{definition}
It is easily found by definition that ${\rm diag}(1,1,\cdots,1) \preceq \mathbf{R}_i \preceq {\rm diag}({N},0,\cdots,0)$ if ${\rm tr}(\mathbf{R}_i) = N$, where ${\rm diag}({N_t},0,\cdots,0)$ and ${\rm diag}(1,1,\cdots,1)$ correspond to completely correlated and independent cases, respectively. Notice that $\det(\mathbf{R}_{i})=\prod\nolimits_{j=1}^{N}r_{i,j}$, the property of the majorization in \cite[F.1.a]{marshall1979inequalities} proves that the determinant of the correlation matrix is a Schur-concave function, where $\det(\mathbf{R}_{1}) \ge \det(\mathbf{R}_{2})$ if $\mathbf{R}_{1} \preceq \mathbf{R}_{2}$, the interested reader is referred to \cite{marshall1979inequalities} for further details regarding the schur monotonicity. By recalling $\mathcal S({\bf R}_t,{\bf R}_r)$ is the composition of the determinants and using the fact associated with the composition involving Schur-concave functions \cite{marshall1979inequalities}, we arrive at
\begin{equation}\label{eqn:}
\mathcal S({\bf R}_{t_1},{\bf R}_{r_1}) \le \mathcal S({\bf R}_{t_2},{\bf R}_{r_2}),
\end{equation}
whenever ${\bf R}_{t_1} \preceq {\bf R}_{t_2}$  and $ {\bf R}_{r_1} \preceq {\bf R}_{r_2}$. As a consequence, it is concluded that the presence of the spatial correlation adversely impacts the outage performance.

\subsection{Power Allocation Strategy}\label{sec:power_all_imp}
Although the exact and asymptotic outage analyses are based on the assumption of the equal powers allocated to transmit antennas, the similar analytical results can be readily extended to any power allocation strategies. For the sake of extension, we denote by ${\bf R}_{\bf x} = {\mathbb E}(\bf x{\bf x}^{\rm H})$ the covariance matrix of input signals and ${\rm tr}({\bf R}_{\bf x})\le N_t$. Clearly, ${\bf R}_{\bf x}$ is a positive semidefinite matrix. In the circumstance, the mutual information capacity is given by
\begin{equation}\label{eqn:mul_caprew}
\mathcal{I}({\bf x};{\bf y}|{\bf H})=\mathrm{log}_{2}\mathrm{det}\left(\mathbf{I}_{{N_r}} + {\rho}\mathbf{H}{\bf R}_{\bf x}\mathbf{H}^\mathrm{H}\right),
\end{equation}
Actually, \eqref{eqn:mul_caprew} can be easily converted into the simple case of \eqref{eqn:mul_cap} equivalently. To this end, denote by ${{\bf R}_{\bf x}}^{1/2}$ the square root of ${{\bf R}_{\bf x}}$. By using the Kronecker model $\mathbf{H}={\mathbf{R}_r}^{{1}/{2}}{\bf H}_w{\mathbf{R}_{t}}^{{1}/{2}}$, \eqref{eqn:mul_caprew} can be rewritten as
$\mathcal{I}({\bf x};{\bf y}|{\bf H})=\mathrm{log}_{2}\mathrm{det}\left(\mathbf{I}_{{N_r}} + {\rho}\tilde {\mathbf{H}}\tilde{\mathbf{H}}^\mathrm{H}\right)$, where $\tilde {\mathbf{H}} = {\mathbf{R}_r}^{{1}/{2}}{\bf H}_w{\tilde {\mathbf{R}}_{t}}^{{1}/{2}}$ and ${\tilde {\mathbf{R}}_{t}}={ {\mathbf{R}}_{t}}^{{1}/{2}}{\bf R}_{\bf x}{ {\mathbf{R}}_{t}}^{{1}/{2}}$. Thus the exact and asymptotic outage probabilities even for non-independent inputs, i.e., ${\bf R}_{\bf x} \neq {\bf I}_{N_t}$, can also be obtained by following the same methodology as developed in Section \ref{sec:ana}. In contrast to \eqref{eqn:gen_out_asy}, the power allocation strategy obviously will influence the asymptotic outage probability that can be unified as
\begin{equation}\label{eqn:gen_out_asyre}
p_{out}  = \mathcal P ({\bf R}_{\bf x}) \mathcal S ({\bf R}_t,{\bf R}_r) {\left( {{\mathcal C(R)} } \rho \right)^{ - d}} + o\left( {{\rho ^{ - d}}} \right),
\end{equation}
where $\mathcal P ({\bf R}_{\bf x})$ quantifies the impact of the power allocation strategy. According to the above introduced equivalent conversion of $\mathcal{I}({\bf x};{\bf y}|{\bf H})$, it is not hard to prove that $\mathcal P ({\bf R}_{\bf x})$ is given by
\begin{equation}\label{eqn:P_def_nonin}
\mathcal P ({\bf R}_{\bf x}) = \frac{1}{{\det \left( {{{{\bf{R}}_{\bf x}}}} \right)^{N_r}}}.
\end{equation}
By using the arithmetic-geometric mean inequality \cite[Exercise 12.11]{abadir2005matrix}, $\mathcal P ({\bf R}_{\bf x})$ is found to be lower bounded as
\begin{equation}\label{eqn:P_bound}
  \mathcal P ({\bf R}_{\bf x}) \ge {{\left(\frac{1}{N_t}{\rm tr} \left( {{{{\bf{R}}_{\bf x}}}} \right)\right)}^{-N_r}} \ge 1,
\end{equation}
where the equality holds if and only if all the eigenvalues of ${\bf{R}}_{\bf x}$ are the same, that is, the total transmit power are evenly assigned to the transmit antennas. In fact, the concept of majorization adopted in Section \ref{sec:scorr} also applies to study the effect of power allocation, because the forms of the quantified impacts of the spatial correlation and power allocation look alike by comparing \eqref{eqn:P_def_nonin} to \eqref{eqn:spa_imp_exp}. Thereby, we conclude that $\mathcal P ({\bf R}_{{\bf x}_1}) \le \mathcal P ({\bf R}_{{\bf x}_2})$ if ${\bf R}_{{\bf x}_1} \preceq {\bf R}_{{\bf x}_2}$. Furthermore, unlike the water-filling algorithm, the equal power allocation can provide the best performance of diminishing the outage probability under high SNR. It is not beyond our expectation, because the condition of high SNR endows each antenna with the same capability of attaining the maximum potential array gain, and the absence of perfect CSI at the transmitter implies that every link between the transmit and receive antennas has equal chance for the performance enhancement. This result is also consistent with the  power allocation assumption in \cite{shin2003capacity} from the perspective of the ergodic capacity. 
\section{Numerical Results}\label{sec:num}
In this section, numerical results are presented for verifications and discussions. For notational convenience, we define $\mathbf t$ and $\mathbf r$ as the row vectors of the eigenvalues of the transmit and receive correlation matrices, i.e., $\mathbf R_t$ and $\mathbf R_r$. 

\subsection{Verifications}
Figs. \ref{fig:out_Nt3} and \ref{fig:out_Nt2} depict the outage probabilities versus the transmit SNR for three different correlation scenarios under different numbers of transmit and receive antennas, wherein $\mathbf t$ and $\mathbf r$ are set to all-one vectors unless otherwise specified. 
 The labels `Sim.', `Exa.' and `Asy.' in the figures indicate the simulated, exact and asymptotic outage probabilities, respectively. As can be observed from the two figures, the exact and simulation results are in perfect agreement, which confirms the correctness of the exact analysis. Besides, it can be seen from Figs. \ref{fig:out_Nt3} and \ref{fig:out_Nt2} that the asymptotic results coincide well with the exact and simulation ones at high SNR, which validates the asymptotic results as well. Moreover, it is clear from the two figures that the spatial correlation does not affect the diversity order, which is identical to the slope of the outage probability curve. However, by comparing the outage probabilities under three different correlation scenarios, the negative impact of spatial correlation can be demonstrated. As uncovered by the asymptotic analysis in Section \ref{sec:asym_res}, the outage performance degradation caused by the antenna correlation is quantified by the spatial correlation impact factor, i.e., $\mathcal S({\bf R}_t,{\bf R}_r)$. Hence, the gap between any two outage probability curves shown in Figs. \ref{fig:out_Nt3} and \ref{fig:out_Nt2} is determined by the difference between the values of their associated correlation impact factors. Moreover, by comparing between Figs. \ref{fig:out_Nt3} and \ref{fig:out_Nt2}, it is found that the outage probability curves are exactly the same if the antenna and correlation settings at the transmitter and receiver are interchanged, such as the outage probabilities over independent and full-correlated channel models. 




\begin{figure}[!t]
\centering
\includegraphics[width=3.5in]{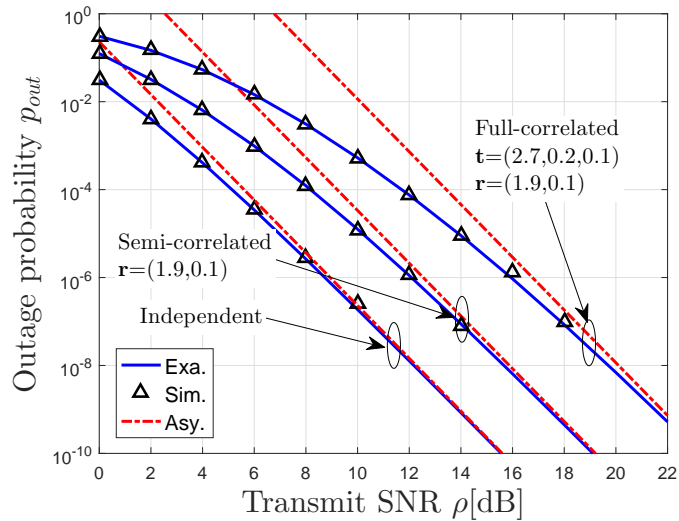}
\caption{Outage probability versus the transmit SNR $\rho$ with ${N_t}=3$, ${N_r}=2$ and $R=$2bps/Hz.}\label{fig:out_Nt3}
\end{figure}
\begin{figure}[!t]
\centering
\includegraphics[width=3.5in]{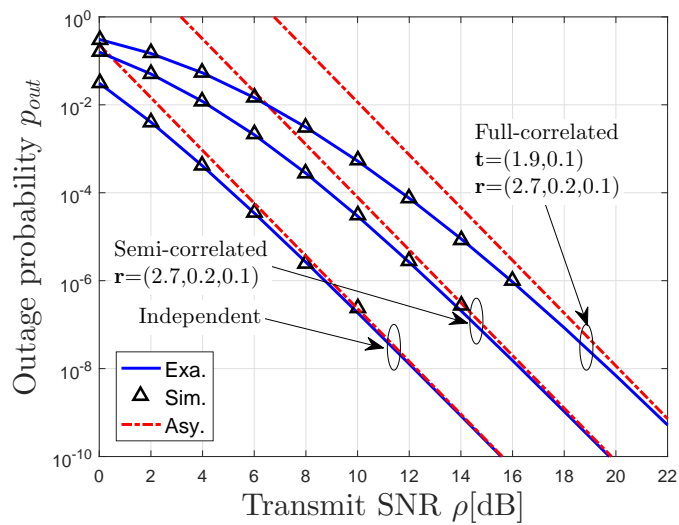}
\caption{Outage probability versus the transmit SNR $\rho$ with ${N_t}=2$, ${N_r}=3$ and $R=$2bps/Hz.}\label{fig:out_Nt2}
\end{figure}

\subsection{Coding and Modulation Gain}
Fig. \ref{fig:code_gain} illustrates the impacts of the transmission rate $R$ on the coding and modulation gain ${\mathcal C(R)}$ for different numbers of transmit and receive antennas. It is shown in Fig. \ref{fig:code_gain} that ${\mathcal C(R)}$ increases with the number of antennas, which further justifies the benefit of using MIMO. Additionally, it can be observed from Fig. \ref{fig:code_gain} that the increase of the transmission rate impairs the coding and modulation gain ${\mathcal C(R)}$, which consequently leads to the deterioration of the outage performance. This is consistent with the asymptotic analysis in Section \ref{sec:mc_gain}. Aside from degrading the outage performance, the increase of the transmission rate causes the enhancement of the system throughput. The two opposite effects force us to properly select the transmission rate in practice. Fortunately, the optimal rate selection can be eased by using the asymptotic outage probability thanks to its increasing monotonicity and convexity with respect to the transmission rate.

\begin{figure}[!t]
\begin{center}
\includegraphics[width=3.5in]{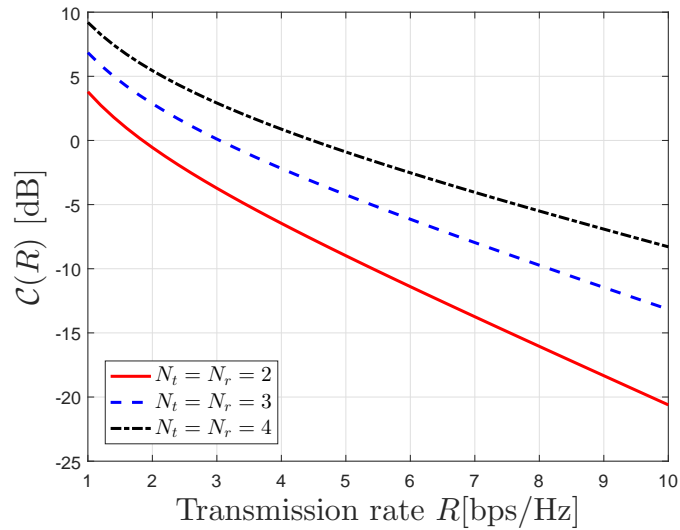}
\caption{Coding and modulation gain ${\mathcal C(R)}$ versus the transmission rate $R$.}\label{fig:code_gain}
\end{center}
\end{figure}
\subsection{Impact of Spatial Correlation}
Without loss of generality, the impact of transmit antenna correlation is investigated in Fig. \ref{fig:corr_imp}, where the outage probability is plotted against the transmit SNR under three different transmit correlation matrices, i.e., ${\bf{R}}_{t_1}$, ${\bf{R}}_{t_2}$ and ${\bf{R}}_{t_3}$. For notational simplicity, the vectors of the eigenvalues of ${\bf{R}}_{t_1}$, ${\bf{R}}_{t_2}$ and ${\bf{R}}_{t_3}$ are denoted by ${\mathbf t}_1$, ${\mathbf t}_2$ and ${\mathbf t}_3$, respectively, and they are set as ${\mathbf t}_1 = (1,1,1)$, ${\mathbf t}_2=(2.3,0.5,0.2)$ and ${\mathbf t}_3 = (2.7,0.2,0.1)$. According to the concept of majorization, the relationship of the transmit correlation matrices follows as ${{\bf{R}}_{t_3}}\succeq{{\bf{R}}_{t_2}}\succeq{{\bf{R}}_{t_1}}$, and $\mathbf{R}_{t_3}$ are the most correlated correlation matrix among them. It is readily observed in Fig. \ref{fig:corr_imp} that the spatial correlation negatively influences the outage performance, and the outage probability curve associated with $\mathbf{R}_{t_3}$ displays the worst performance. The numerical result corroborates the validity of the analytical results in Section \ref{sec:scorr}.


\begin{figure}[!t]
\begin{center}
\includegraphics[width=3.5in]{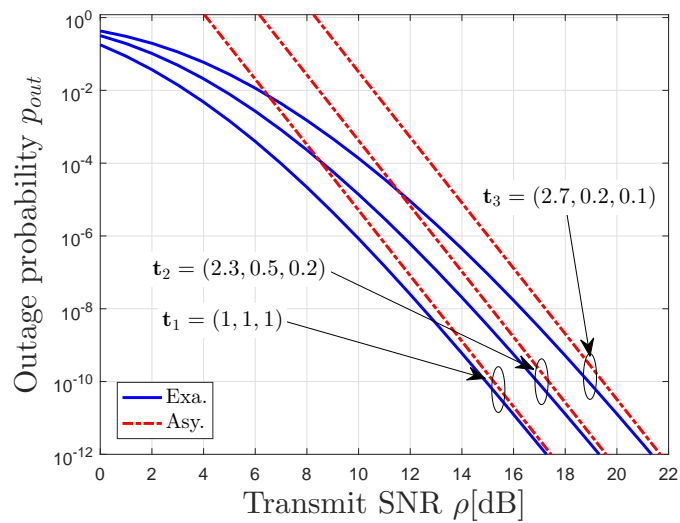}
\caption{Outage probability versus the transmit SNR $\rho$ with ${N_t}={N_r}=3$ and ${\mathbf r}=(2.7,0.2,0.1)$.}\label{fig:corr_imp}
\end{center}
\end{figure}
\subsection{Impact of Power Allocation}
In Fig. \ref{fig:power_str}, the effect of power allocation strategy is examined by considering three different diagonal input covariance matrices, i.e., ${{\bf R}_{{\bf x}_1}}={\rm diag}(1,1,1)$, ${{\bf R}_{{\bf x}_2}}={\rm diag}(2.6,0.2,0.2)$ and ${{\bf R}_{{\bf x}_3}}={\rm diag}(2.9,0.07,0.03)$. It is worth mentioning that the diagonal entries of input covariance matrix indicate the amount of power allocated to the transmit antennas. Hence, the case of ${{\bf R}_{{\bf x}_1}}$ corresponds to the equal power allocation. By using the definition of the majorization, we have ${{\bf R}_{{\bf x}_3}} \succeq {{\bf R}_{{\bf x}_2}} \succeq {{\bf R}_{{\bf x}_1}}$. As expected, the equal power allocation performs the best among the three power allocation strategies in terms of the outage probability. Thus, the analytical results in Section \ref{sec:power_all_imp} are verified.


\begin{figure}[!t]
\begin{center}
\includegraphics[width=3.5in]{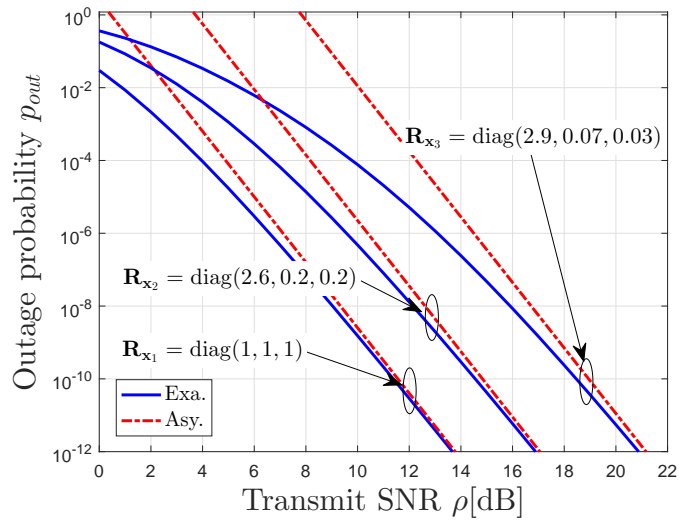}
\caption{Outage probability versus the transmit SNR $\rho$ with $R=3$bps/Hz, $N_t=N_r=3$, ${\mathbf t}=(1.3,1,0.7)$ and ${\mathbf r}=(1.5,1,0.5)$.}\label{fig:power_str}
\end{center}
\end{figure}
\section{Conclusions}\label{sec:con}
This paper has derived novel representations for the outage probabilities of the MIMO systems by invoking Mellin transform, where Kronecker model has been employed to accommodate three correlated fading channels, including independent, semi-correlated and full-correlated Rayleigh MIMO channels. The analytical results have been expressed in terms of the generalized Fox's H function. The compact and simple expressions not only have enabled the accurate evaluation of the outage probability, but also have facilitated the asymptotic analysis under high SNR to gain a profound understanding of fading effects and MIMO configurations, which has never been performed in the literature. On one hand, the unified asymptotic results have revealed meaningful insights into the effects of the spatial correlation, the number of antennas, transmission rate and powers. For instance, the spatial correlation degenerates the outage performance, while full diversity can be achieved no matter whether the spatial correlation occurs or not. On the other hand, the asymptotic results have paved the way for the simplification of practical system designs. For example, the increasing monotonicity and convexity of the asymptotic outage probability will facilitate the proper selection of target transmission rate. Moreover, unlike the water-filling fashion of power allocation, the minimization of the outage probability forces the transmitter to equally allocate the total power among its antennas especially for high SNR if only the statistical knowledge of the CSI is available at the transmitter.
\appendices
\section{Proof of \eqref{eqn:mellin_g_m}}\label{app:mellin_g_m}
By plugging (\ref{eqn:joint_pdf_lambdav}) into (\ref{eqn:mellin_G}) and using the Leibniz formula for the determinant expansion \cite{horn2012matrix}, $\varphi_{\rm ind}(s)$ can be expressed as
\begin{align}
\varphi_{\rm ind}(s)&= \frac{1}{{N_r}!} \int\nolimits_0^\infty  { \cdots \int\nolimits_0^\infty  {\prod\limits_{i = 1}^{N_r} {{{\left( {1 +  {\rho}{\lambda _i}} \right)}^{s - 1}}{e^{ -\sum\limits_{i = 1}^{N_r} {{\lambda _i}} }}\prod\limits_{i = 1}^{N_r} {\frac{{{\lambda _i}^{{N_t} - {N_r}}}}{{\left( {{N_r} - i} \right)!\left( {{N_t} - i} \right)!}}}  } } } \notag\\
 &\quad \times\sum\limits_{{\bs\sigma _1} \in {S_{N_r}}} {{\mathop{\rm sgn}} \left( {{\bs\sigma _1}} \right)\prod\limits_{i = 1}^{N_r} {{\lambda _i}^{{\sigma _{1,i}} - 1}} } \sum\limits_{{\bs \sigma _2} \in {S_{N_r}}} {{\mathop{\rm sgn}} \left( {{\bs\sigma _2}} \right)\prod\limits_{i = 1}^{N_r} {{\lambda _i}^{{\sigma _{2,i}} - 1}} } d{\lambda _1} \cdots d{\lambda _{N_r}}.
\label{eqn:out_prob_mell_re}
\end{align}
By switching the order of summation and integration, we get
\begin{align}\label{eqn:mellin_g_m1}
\varphi_{\rm ind}(s) &=\frac{1}{{N_r}!}  \sum\limits_{{{\bs\sigma _1},{\bs\sigma _2} \in {S_{N_r}}}} {{\mathop{\rm sgn}} \left( {{\bs\sigma _1}} \right){\mathop{\rm sgn}} \left( {{\bs\sigma _2}} \right)} \prod\limits_{i = 1}^{N_r} {\frac{{\int\nolimits_0^\infty  {{{\left( {1 +  {\rho}{\lambda _i}} \right)}^{s - 1}}{\lambda _i}^{  {N_t} - {N_r} - 2+{\sum\nolimits_{l = 1}^2 {{\sigma _{l,i}}} }}{e^{ - {\lambda _i}}}d{\lambda _i}} }}{{\left( {{N_r} - i} \right)!\left( {{N_t} - i} \right)!}}}. 
\end{align}
By comparing between the integration in \eqref{eqn:mellin_g_m1} and the Tricomi's confluent hypergeometric function \cite[eq. (9.211.4)]{gradshteyn1965table}, \eqref{eqn:mellin_g_m1} can be finally represented by \eqref{eqn:mellin_g_m}.

\section{Proof of Lemma \ref{the:leb_for}}\label{app:proof_leb_for}
Similar to the proofs of Leibniz formulae \cite[eq.(2)]{simon2006capacity} and \cite[eqs.(64-65)]{ghaderipoor2012application}, denote by $S_{N_r} \times S_{N_r}$ the Cartesian product of $S_{N_r}$. Hence, $(\bs \sigma_1, \bs \sigma_2)\in S_{N_r} \times S_{N_r}$. We further establish the one-to-one mapping $\vartheta(\bs \sigma_1, \bs \sigma_2)$ as a vector of ordered pairs $(\sigma _{1,l},\sigma _{2,l})$ for $l\in [1,N_r]$, i.e., $\vartheta(\bs \sigma_1, \bs \sigma_2) \triangleq \left((\sigma _{1,l},\sigma _{2,l}):l\in[1,N_r]\right)$. We thus reach the relation ${\mathop{\rm sgn}} \left( {{\bs\sigma _1}} \right){\mathop{\rm sgn}} \left( {{\bs\sigma _2}} \right) = {\mathop{\rm sgn}} \left( {{\bs\sigma }} \right){\mathop{\rm sgn}} \left( {\bar{\bs\sigma}} \right)$ after a certain number of transpositions to achieve $\vartheta(\bs \sigma, \bar{\bs\sigma})$ starting from $\vartheta(\bs \sigma_1, \bs \sigma_2)$, where ${\mathop{\rm sgn}} \left( {\bar{\bs\sigma}} \right)=1$. Accordingly, the left hand side of \eqref{eqn:leb_for_gen} can be easily obtained as
\begin{align}\label{eqn:sing_g_twosum_com}
\sum\limits_{{{\bs\sigma _1},{\bs\sigma _2} \in {S_{N_r}}}} {{\mathop{\rm sgn}} \left( {{\bs\sigma _1}} \right){\mathop{\rm sgn}} \left( {{\bs\sigma _2}} \right)}\eta({\bs\sigma _1},{\bs\sigma _2})
 &= \sum\limits_{{{\bs\sigma _1},{\bs\sigma _2} \in {S_{N_r}}}} { {{\rm{sgn}}\left( {{{\bs\sigma }}} \right){\mathop{\rm sgn}} \left( {\bar{\bs\sigma}} \right)\eta({\bs\sigma},\bar{\bs\sigma})} } \notag\\
 &= N_r!\sum\limits_{{\bs\sigma} \in {S_{N_r}}} {{\rm{sgn}}\left( {{{\bs\sigma }}} \right)\eta({\bs\sigma},\bar{\bs\sigma})},
\end{align}
where the last step holds according to the definition of $\eta ({\bs\sigma _1},{\bs\sigma _2})$, the one-to-one mapping and the cardinality of $S_{N_r}$, i.e., $|S_{N_r}|=N_r!$. The proof for the first equality is thus completed, and the second equality can be proved in the same manner.

\section{Proof of Theorem \ref{the:cdf_g_ind}}\label{app:cdf_g_ind}
Putting (\ref{eqn:mellin_g_m_rew}) into (\ref{eqn:cdf_g_inverse}) and using the identity ${\Gamma(-s)}/{\Gamma(1-s)}=-1/s$ leads to
\begin{align}
{F_G^{(1)}}\left( x \right) &= \sum\limits_{{{\bs\sigma } \in {S_{N_r}}}}\frac{{{\mathop{\rm sgn}} \left( {{\bs\sigma }} \right)}} {{{\rho} }^{{N_t}{N_r}}}\prod\limits_{i = 1}^{N_r} {\frac{{\Gamma \left( {\tau+i+ {{\sigma _{i}}}  } \right)}}{{\left( {{N_r} - i} \right)!\left( {{N_t} - i} \right)!}}} \notag\\
 &\quad \times \frac{1}{{2\pi {\rm{i}}}}\int\nolimits_{c - {\rm i}\infty }^{c + {\rm i}\infty } {\frac{{\Gamma \left( { - s} \right)}}{{\Gamma \left( {1 - s} \right)}}\prod\limits_{i = 1}^{N_r} {\Psi \left( {{\tau+i+ {{\sigma _{i}}}  },s + \tau+ i+ {{\sigma _{i}}}+1;{\rho^{-1}}} \right)} {x^{ - s}}ds}.
\label{eqn:cdf_g_mellin_in}
\end{align}

By using the definition of $\Xi \left( {a,\alpha ,A,\varphi } \right)$, ${F_G^{(1)}}\left( x \right)$ can be rewritten as
\begin{align}\label{eqn:cdf111}
{F_G^{(1)}}\left( x \right)
 &= \sum\limits_{{{\bs\sigma } \in {S_{N_r}}}}{{{\mathop{\rm sgn}} \left( {{\bs\sigma }} \right)}} \prod\limits_{i = 1}^{N_r} {\frac{{\Gamma \left( {\tau+i+ {{\sigma _{i}}}  } \right)}}{{\left( {{N_r} - i} \right)!\left( {{N_t} - i} \right)!}}} \notag\\
&\quad \times \frac{1}{{2\pi {\rm{i}}}}\int\nolimits_{c - {\rm i}\infty }^{c + {\rm i}\infty } {\frac{{\Xi \left( {0, - 1,0,1} \right)\prod\nolimits_{i = 1}^{N_r} {\Xi \left( {1,1,\rho^{-1},\tau+ i+ {{\sigma _{i}}}} \right)} }}{{\Xi \left( {1, - 1,0,1} \right)}}{{\left( \frac{x}{{\rho}^{N_r}} \right)}^{ - s}}ds},
\end{align}
where the last step holds by using \cite [eq. (55)]{shi2017asymptotic}, i.e., $\Xi \left( {a, - 1,0,1} \right) = \Gamma \left( {a - s} \right),\, a - s > 0$.
By identifying the integration in \eqref{eqn:cdf111} with the generalized Fox's H function \cite{chelli2013performance,yilmaz2010outage}, \eqref{eqn:cdf111} can finally be obtained as (\ref{eqn:cdf_ML}).
\section{Proof of Lemma \ref{the:asy_out_ind}}\label{app:proof_phi}
By using Property 2 in \cite{shi2017asymptotic}, the generalized Fox's H function ${\mathcal Y_{ {{\bs\sigma}}}^{(1)}(x)}$ can be rewritten as
\begin{equation}\label{eqn:h_sigma_pro2}
{\mathcal Y_{ {{\bs\sigma}}}^{(1)}(x)} =
 Y_{{N_r} + 1,1}^{1,{N_r}}\left[ {\left. {\begin{array}{*{20}{c}}
{\left( {0,1,\rho^{-1},\tau+i+ {{\sigma _{i}}}  } \right)_{i=1,\cdots,{N_r}}, \left( {1,1,0,1} \right)}\\
{\left( {0,1,0,1} \right)}
\end{array}} \right|\frac{{\rho}^{N_r}}{x}} \right].
\end{equation}
\eqref{eqn:h_sigma_pro2} can thus be expressed in terms of the integral representation of Mellin-Branes type as
\begin{multline}\label{eqn:h_simg_pro_rew}
{\mathcal Y_{ {{\bs\sigma}}}^{(1)}(x)} =\rho^{-{N_t}{N_r}} \frac{1}{{2\pi {\rm{i}}}}\int\nolimits_{-c - {\rm i}\infty }^{-c + {\rm i}\infty } {\frac{{\Gamma \left( s \right)}}{{\Gamma \left( {1 + s} \right)}}\prod\limits_{i = 1}^{N_r} {\Psi \left( {\tau+ i+{{\sigma _{i}}},\tau - s+ 1 +i+\sigma _{i} ;\rho^{-1}} \right)} {x^s}ds}.
\end{multline}
By using the relation \cite[eq.(9.210.2)]{gradshteyn1965table}, (\ref{eqn:h_simg_pro_rew}) can be decomposed as
\begin{multline}\label{eqn:h_simg_pro_rewexp}
{\mathcal Y_{ {{\bs\sigma}}}^{(1)}(x)}
 =\rho^{-{N_t}{N_r}} \\
 \times\frac{1}{{2\pi {\rm{i}}}}\int\nolimits_{-c - {\rm i}\infty }^{-c + {\rm i}\infty } {\frac{1}{s}}\prod\limits_{i = 1}^{N_r} {\left( \begin{array}{l}
\frac{{\Gamma \left( {s - \tau - i-{{\sigma _{i}}}   } \right)}}{{\Gamma \left( s \right)}}{}_1{F_1}\left( {\tau+i+{{\sigma _{i}}}   ,\tau - s + 1 + i+{{\sigma _{i}}}  ;\rho^{-1}} \right)+\\
\frac{{\Gamma \left( {\tau  - s +i+{{\sigma _{i}}}} \right)}}{{\Gamma \left( { \tau +i+{{\sigma _{i}}}} \right)}}{{\rho }^{\tau -s +i+{{\sigma _{i}}} }}{}_1{F_1} \left( {s,1 + s - \tau - i-{{\sigma _{i}}}  ;\rho^{-1}} \right)
\end{array} \right)}{x^s}ds.
\end{multline}
where ${}_1{F_1} \left( {\alpha,\gamma;z} \right)$ is the confluent hypergeometric function. To proceed, the confluent hypergeometric function in \eqref{eqn:h_simg_pro_rewexp} can be expanded in terms of the representation of an infinite series as ${}_1{F_1}\left( {\alpha ,\beta ;x} \right) = \sum\nolimits_{n = 0}^\infty  {{{{\left( \alpha  \right)}_n}}}{{{x^n}}}/{{{{\left( \beta  \right)}_n}}}/{{n!}}$  \cite[eq.(9.210.1)]{gradshteyn1965table}, where $(\cdot)_n$ is Pochhammer symbol.
Notice that $c$ could be any real number within $(-\infty,0)$, we set $c < \min\left\{ - i-{{\sigma _{i}}}  - \tau:i\in [1,N_r]  \right\}$, higher order terms relative to ${\rho ^{ - {N_t}{N_r}}}$ can be ignored as $\rho \to \infty$. Hence, ${\mathcal Y_{ {{\bs\sigma}}}^{(1)}(x)}$ is asymptotic to
\begin{equation}\label{eqn:h_simg_pro_rewexpasy}
{\mathcal Y_{ {{\bs\sigma}}}^{(1)}(x)}  = {\rho^{-{N_t}{N_r}}}\frac{1}{{2\pi {\rm{i}}}}\int\nolimits_{-c - {\rm i}\infty }^{-c + {\rm i}\infty } {\frac{{\Gamma \left( s \right)}}{{\Gamma \left( {1 + s} \right)}}\prod\limits_{i = 1}^{N_r} {\frac{{\Gamma \left( {s - \tau -i-{{\sigma _{i}}}} \right)}}{{\Gamma \left( s \right)}}} {x^s}ds}  + o\left( {{\rho ^{ - {N_t}{N_r}}}} \right),
\end{equation}
which can be consequently rewritten as \eqref{eqn:h_fin_asy} by recognizing the contour integral in \eqref{eqn:h_simg_pro_rewexpasy} as Meijer G-function \cite[eq.(9.301)]{gradshteyn1965table}. Since $\sigma_1,\cdots,\sigma_{N_r}$ are integer, the integral in \eqref{eqn:h_simg_pro_rewexpasy} can also be simplified as \eqref{eqn:g2R_laplace}, which is readily calculated by using Residue theorem as \cite[eq.(5.2.21)]{bateman1954tables}.

\section{Proof of \eqref{eqn:mellin_t_corr_gpdf_sumre}}\label{app:mellin_semi}
Due to the different forms of joint eigenvalue distributions for $N_t \ge N_r$ and $N_t < N_r$, we are obliged to derive their outage probabilities separately. In analogous to \eqref{eqn:mellin_g_m}-\eqref{eqn:mellin_g_m_rew}, the similar approach can be adopted to derive the Mellin transform of the PDF of $G$ under semi-correlated Rayleigh MIMO channels, $\varphi_{\rm semi}(s)$.
\subsection{$\varphi_{\rm semi}(s)$ for $N_t \ge N_r$}
By substituting \eqref{eqn:joint_pdf_eig_wishart} into \eqref{eqn:mellin_G} and using the determinant expansion, $\varphi_{\rm semi}(s)$ can be expressed after some basic rearrangements as
\begin{align}\label{eqn:mellin_t_corr_R_out}
&\varphi_{\rm semi}^{\ge}(s) 
 =\frac{{\left( { - 1} \right)^{\frac{1}{2}{N_r}\left( {{N_r} - 1} \right)}}}{{{N_r}!}{\det \left( {{{{\bf{R}}_r}^{N_t}}} \right)\Delta \left( {{{{\bf{R}}_r}^{ - 1}}} \right)}}\prod\limits_{j = 1}^{N_r} {\frac{1}{{\left( {{N_t} - j} \right)!}}} \notag\\
&\quad \times \int\nolimits_0^\infty  { \cdots \int\nolimits_0^\infty  {\prod\limits_{i = 1}^{N_r} {{{\left( {1 + {\rho }{\lambda _i}} \right)}^{s - 1}}\det \left( {{{\left\{ {{e^{ - \frac{{{\lambda _i}}}{{{r_j}}}}}} \right\}}_{1 \le i,j \le {N_r}}}} \right)\Delta \left( {\bf{\Lambda }} \right)\prod\limits_{j = 1}^{N_r} {{\lambda _j}^{{N_t} - {N_r}}} d{\lambda _1} \cdots d{\lambda _{N_r}}} } }\notag\\
& = \frac{{\left( { - 1} \right)^{\frac{1}{2}{N_r}\left( {{N_r} - 1} \right)}}}{{{N_r}!}{\det \left( {{{{\bf{R}}_r}^{N_t}}} \right)\Delta \left( {{{{\bf{R}}_r}^{ - 1}}} \right)}}\prod\limits_{j = 1}^{N_r} {\frac{1}{{\left( {N_t - j} \right)!}}} \notag\\
&\quad\times \int\nolimits_0^\infty  { \cdots \int\nolimits_0^\infty  } \sum\limits_{{\bs\sigma _1},{\bs \sigma _2} \in {S_{N_r}}} {{\rm{sgn}}\left( {{\bs\sigma _1}} \right){\rm{sgn}}\left( {{\bs\sigma _2}} \right)\prod\limits_{i = 1}^{N_r} {{{{\left( {1 + {\rho }{\lambda _i}} \right)}^{s - 1}}}{{\lambda _i}^{{\sigma _{2,i}} + \tau}}{e^{ - \frac{{{\lambda _i}}}{{{r_{{\sigma _{1,i}}}}}}}}} } d{\lambda _1} \cdots d{\lambda _{N_r}}.
\end{align}
Interchanging the order of summations and integrations together with \cite[eq.(9.211.4)]{gradshteyn1965table} then gives
\begin{align}\label{eqn:mellin_t_corr_gpdf_sum}
\varphi_{\rm semi}^{\ge}(s) &= \frac{{{{\left( { - 1} \right)}^{\frac{1}{2}{N_r}\left( {{N_r} - 1} \right)}}}}{{{N_r}!\det \left( {{{{\bf{R}}_r}^{N_t}}} \right)\Delta \left( {{{{\bf{R}}_r}^{ - 1}}} \right)\prod\nolimits_{j = 1}^{N_r} {\left( {N_t - j} \right)!} }}\sum\limits_{{\bs\sigma _1},{\bs\sigma _2} \in {S_{N_r}}} { {{\rm{sgn}}\left( {{\bs\sigma _1}} \right){\rm{sgn}}\left( {{\bs \sigma _2}} \right)} } \notag\\
&\quad\times\prod\limits_{i = 1}^{N_r} {\int\nolimits_0^\infty  {{{\left( {1 + {\rho }{\lambda _i}} \right)}^{s - 1}}{\lambda _i}^{{\sigma _{2,i}} + \tau }{e^{ - \frac{{{\lambda _i}}}{{{r_{{\sigma _{1,i}}}}}}}}d{\lambda _i}} }\notag\\
&= \frac{{{{\left( { - 1} \right)}^{\frac{1}{2}{N_r}\left( {{N_r} - 1} \right)}}}{ \rho ^{-\frac{{{N_r}\left( {{N_r} + 1} \right)}}{2} - \left( {N_t - {N_r}} \right){N_r}}}}{{{N_r}!\det \left( {{{{\bf{R}}_r}^{N_t}}} \right)\Delta \left( {{{{\bf{R}}_r}^{ - 1}}} \right)\prod\nolimits_{j = 1}^{N_r} {\left( {N_t - j} \right)!} }}\sum\limits_{{\bs\sigma _1},{\bs\sigma _2} \in {S_{N_r}}} { {{\rm{sgn}}\left( {{\bs\sigma _1}} \right){\rm{sgn}}\left( {{\bs \sigma _2}} \right)} } \notag\\
&\quad\times\prod\limits_{i = 1}^{N_r} {\Gamma \left( {{\sigma _{2,i}} + \tau  + 1} \right)\Psi \left( {{\sigma _{2,i}} + \tau  + 1,s + {\sigma _{2,i}} + \tau  + 1;\frac{1}{{\rho {r_{{\sigma _{1,i}}}}}}} \right)}.
\end{align}

By virtue of Lemma \ref{the:leb_for}, \eqref{eqn:mellin_t_corr_gpdf_sum} can be finally simplified as \eqref{eqn:mellin_t_corr_gpdf_sumre1}.


\subsection{$\varphi_{\rm semi}(s)$ for $N_t < N_r$}
By noticing the zero eigenvalues of ${\bf{H}}{\bf{H}}^{\rm{H}}$ if $N_t < N_r$ does not change the value of the mutual information capacity, the random variable $G$ can be rewritten as $G = \prod\nolimits_{i = 1}^{{N_t}} {\left( {1 + {\rho}{\lambda _i}} \right)} $ by assuming $\lambda _{N_t+1}=\cdots=\lambda _{N_r}=0$. Hence, the Mellin transform of $G$ for $N_t < N_r$, $\varphi_{\rm semi}^{<}(s)$, can be expressed as
\begin{equation}\label{eqn:mellin_semi_small_def}
\varphi _{{\rm{semi}}}^ < (s) = \int_0^\infty  { \cdots \int_0^\infty  {\prod\limits_{i = 1}^{{N_t}} {{{\left( {1 + \rho {\lambda _i}} \right)}^{s - 1}}{f_{\tilde {\bs\lambda} }}\left( {{\lambda _1}, \cdots ,{\lambda _{{N_t}}}} \right)d{\lambda _1} \cdots d{\lambda _{{N_t}}}} } }.
\end{equation}
By plugging \eqref{eqn:joint_pdf_eig_semi_small} into \eqref{eqn:mellin_semi_small_def}, then combining the determinant expansion, and after some algebraic manipulations, we arrive at
\begin{align}\label{eqn:varphi_semi_small_rew}
&\varphi _{{\rm{semi}}}^ < (s) = \frac{{{{\left( { - 1} \right)}^{{N_t}\left( {{N_r} - {N_t}} \right)}}}}{{\prod\nolimits_{j = 1}^{{N_t}} {j!} \Delta \left( {{{\bf{R}}_r}} \right)}}\notag\\
&\quad \times \int_0^\infty  { \cdots \int_0^\infty  {\prod\limits_{i = 1}^{{N_t}} {{{\left( {1 + \rho {\lambda _i}} \right)}^{s - 1}}\Delta \left( {{\bf{\tilde \Lambda }}} \right)\det \left( \begin{array}{l}
{\left\{ {{r_j}^{{N_r} - {N_t} - 1}{e^{ - \frac{{{\lambda _i}}}{{{r_j}}}}}} \right\}_{\scriptstyle1 \le i \le {N_t}\hfill\atop
\scriptstyle1 \le j \le {N_r}\hfill}}\\
{\left\{ {{r_j}^{i - {N_t} - 1}} \right\}_{\scriptstyle{N_t} + 1 \le i \le {N_r}\hfill\atop
\scriptstyle1 \le j \le {N_r}\hfill}}
\end{array} \right)d{\lambda _1} \cdots d{\lambda _{{N_t}}}} } } \notag\\
&= \frac{{{{\left( { - 1} \right)}^{{N_t}\left( {{N_r} - {N_t}} \right)}}}}{{\prod\nolimits_{j = 1}^{{N_t}} {j!} \Delta \left( {{{\bf{R}}_r}} \right)}}\sum\limits_{\scriptstyle{{\bs{\sigma }}_1} \in {S_{{N_t}}}\hfill\atop
\scriptstyle{{\bs{\sigma }}_2} \in {S_{{N_r}}}\hfill} {{\mathop{\rm sgn}} \left( {{{\bs{\sigma }}_1}} \right){\mathop{\rm sgn}} \left( {{{\bs{\sigma }}_2}} \right)\prod\limits_{i = 1}^{{N_t}} {{r_{{\sigma _{2,i}}}}^{{N_r} - {N_t} - 1}} \prod\limits_{i = {N_t} + 1}^{{N_r}} {{r_{{\sigma _{2,i}}}}^{i - {N_t} - 1}}  } \notag\\
&\quad \times \prod\limits_{i = 1}^{{N_t}} {{{\int_0^\infty  {{{\left( {1 + \rho {\lambda _i}} \right)}^{s - 1}}{\lambda _i}^{{\sigma _{1,i}} - 1}e} }^{ - \frac{{{\lambda _i}}}{{{r_{{\sigma _{2,i}}}}}}}}d{\lambda _i}}.
\end{align}
By using \cite[eq.(9.211.4)]{gradshteyn1965table} along with some rearrangements, it follows that
\begin{align}\label{eq:semi_me_small_conf}
\varphi _{{\rm{semi}}}^ < (s) &= \frac{{{{\left( { - 1} \right)}^{{N_t}\left( {{N_r} - {N_t}} \right)}}{\rho ^{ - \frac{1}{2}{N_t}\left( {{N_t} + 1} \right)}}}}{{{N_t}!\Delta \left( {{{\bf{R}}_r}} \right)}}\sum\limits_{{{\bs{\sigma }}_2} \in {S_{{N_r}}}} {{\mathop{\rm sgn}} \left( {{{\bs{\sigma }}_2}} \right)  \prod\limits_{i = {N_t} + 1}^{{N_r}} {{r_{{\sigma _{2,i}}}}^{i - {N_t} - 1}} } \notag\\
&\quad \times \sum\limits_{{{\bs{\sigma }}_1} \in {S_{{N_t}}}} {{\mathop{\rm sgn}} \left( {{{\bs{\sigma }}_1}} \right)\prod\limits_{i = 1}^{{N_t}} {{{r_{{\sigma _{2,i}}}}^{{N_r} - {N_t} - 1}}\Psi \left( {{\sigma _{1,i}},s + {\sigma _{1,i}},\frac{1}{{\rho {r_{{\sigma _{2,i}}}}}}} \right)} }.
\end{align}

\begin{lemma}\label{the:leb_for1}
If $\eta ({\bs\sigma _1},{\bs\sigma _2})$ is a function of ${{\bs{\sigma }}_1} \in {S_{{N_t}}}$ and ${\bs\sigma _2}\in {S_{{N_r}}}$ irrespective of the elements of the ordering of the permutations of the set of two-tuples $\{(\sigma _{1,l},\sigma _{2,l}):l\in[1,N_t]\}$, and $N_t < N_r$, the summation of ${{\mathop{\rm sgn}} \left( {{\bs\sigma _1}} \right){\mathop{\rm sgn}} \left( {{\bs\sigma _2}} \right)}\eta({\bs\sigma _1},{\bs\sigma _2})$ over all permutations of $\bs\sigma _1$ and $\bs\sigma _2$ reduces to
\begin{align}\label{eqn:leb_for_gen1}
\sum\limits_{{\scriptstyle{{\bs{\sigma }}_1} \in {S_{{N_t}}}\hfill\atop
\scriptstyle{{\bs{\sigma }}_2} \in {S_{{N_r}}}\hfill}} {{\mathop{\rm sgn}} \left( {{\bs\sigma _1}} \right){\mathop{\rm sgn}} \left( {{\bs\sigma _2}} \right)}\eta({\bs\sigma _1},{\bs\sigma _2})
 = N_t!\sum\limits_{{\bs\sigma} \in {S_{N_r}}} {{\rm{sgn}}\left( {{{\bs\sigma }}} \right)\eta(\bar{\bs\sigma}_1,{\bs\sigma})},
\end{align}
where $\bar{\bs\sigma}_1=(1,\cdots,N_t)$.
\end{lemma}
\begin{proof}
Similar to the proof of Lemma \ref{the:leb_for}, we establish the one-to-one mapping $\vartheta(\bs \sigma_1, \tilde{\bs \sigma}_2)$ as a vector of ordered pairs $(\sigma _{1,l},\sigma _{2,l})$ for $l\in [1,N_t]$, i.e., $\vartheta(\bs \sigma_1,\tilde{\bs \sigma}_2) \triangleq \left((\sigma _{1,l},\sigma _{2,l}):l\in[1,N_t]\right)$, where $\tilde{\bs \sigma}_2$ is constructed by slicing the first $N_t$ elements from $\bs \sigma_2$, i.e., ${\bs\sigma}_2=(\tilde{\bs\sigma}_2,\sigma _{2,N_t+1},\cdots,\sigma _{2,N_r})$. After a number of transpositions to achieve $\vartheta(\bar{\bs\sigma}_1, \tilde{\bs\sigma})$ starting from $\vartheta(\bs \sigma_1, \tilde{\bs \sigma}_2)$, we conclude that ${\mathop{\rm sgn}} \left( {{\bs\sigma _1}} \right){\mathop{\rm sgn}} \left( {{\bs\sigma _2}} \right) = {\mathop{\rm sgn}} \left( {\bar{\bs\sigma}_1} \right){\mathop{\rm sgn}} \left( {{\bs\sigma}} \right)$, where ${\bs\sigma}=(\tilde{\bs\sigma},\sigma _{2,N_t+1},\cdots,\sigma _{2,N_r})$ and ${\mathop{\rm sgn}} \left( {\bar{\bs\sigma}_1} \right)=1$. Thus, the left hand side of \eqref{eqn:leb_for_gen1} can be written as
\begin{align}\label{eqn:sing_g_twosum_com1}
\sum\limits_{{\scriptstyle{{\bs{\sigma }}_1} \in {S_{{N_t}}}\hfill\atop
\scriptstyle{{\bs{\sigma }}_2} \in {S_{{N_r}}}\hfill}} {{\mathop{\rm sgn}} \left( {{\bs\sigma _1}} \right){\mathop{\rm sgn}} \left( {{\bs\sigma _2}} \right)}\eta({\bs\sigma _1},{\bs\sigma _2})
 &= \sum\limits_{{\scriptstyle{{\bs{\sigma }}_1} \in {S_{{N_t}}}\hfill\atop
\scriptstyle{{\bs{\sigma }}_2} \in {S_{{N_r}}}\hfill}} {{\mathop{\rm sgn}} \left( \bar{\bs\sigma}_1 \right){\mathop{\rm sgn}} \left( {{\bs\sigma}} \right)}\eta(\bar{\bs\sigma}_1,{\bs\sigma}) \notag\\
 &= N_t!\sum\limits_{{\bs\sigma} \in {S_{N_r}}} {{\rm{sgn}}\left( {{{\bs\sigma }}} \right)\eta(\bar{\bs\sigma}_1,{\bs\sigma})},
\end{align}
where the last step holds similar to \eqref{eqn:sing_g_twosum_com}. The proof is thus accomplished.
\end{proof}
By using Lemma \ref{the:leb_for1}, \eqref{eq:semi_me_small_conf} can be finally written as \eqref{eq:semi_me_small_conf_fin}.

\section{Proof of Theorem \ref{the:fg_semi}}\label{app:fg_semi}
On the basis of $\varphi _{{\rm{semi}}} (s)$, the CDF of $G$, $ {F_G^{(2)}}\left( x \right)$, can be derived by using the inverse Mellin transform similarly to the proof of Theorem \ref{the:cdf_g_ind}. Due to the different forms of the Mellin transforms $\varphi _{{\rm{semi}}}^ \ge (s)$ and $\varphi _{{\rm{semi}}}^ < (s)$, the derivation of the CDF ${F_G^{(2)}}\left( x \right)$ should be split into two cases as follows.
\subsection{$ F_G^{(2)}(x)$ for $N_t \ge N_r$}
By putting \eqref{eqn:mellin_t_corr_gpdf_sumre1} into \eqref{eqn:cdf_g_inverse}, ${F_G^{(2\ge)}}\left( x \right)$ can be obtained by using inverse Mellin transform as
\begin{align}\label{eqn:mellin_t_corr_gpdf_sum_conf}
{F_G^{(2\ge)}}\left( x \right) &= \frac{{{{\left( { - 1} \right)}^{\frac{1}{2}{N_r}\left( {{N_r} - 1} \right)}}}{ \rho ^{-\frac{{{N_r}\left( {{N_r} + 1} \right)}}{2} - \left( {N_t - {N_r}} \right){N_r}}}\prod\nolimits_{i = 1}^{N_r} {\Gamma \left( {i + \tau  + 1} \right)}}{{\det \left( {{{{\bf{R}}_r}^{N_t}}} \right)\Delta \left( {{{{\bf{R}}_r}^{ - 1}}} \right)\prod\nolimits_{j = 1}^{N_r} {\left( {N_t - j} \right)!} }}\sum\limits_{{\bs\sigma} \in {S_{N_r}}} { {{\rm{sgn}}\left( {{\bs\sigma }} \right)} } \notag\\
 &\quad \times \frac{1}{{2\pi {\rm{i}}}}\int\nolimits_{c - {\rm i}\infty }^{c + {\rm i}\infty } {\frac{{\Gamma \left( { - s} \right)}}{{\Gamma \left( {1 - s} \right)}}\prod\limits_{i = 1}^{N_r} {\Psi \left( {i + \tau  + 1,s + i + \tau  + 2;\left({{\rho {r_{{\sigma _{i}}}}}}\right)^{-1}} \right)} {x^{ - s}}ds}.
 \end{align}
With the definition of $\Xi \left( {a,\alpha ,A,\varphi } \right)$, ${F_G^{(2\ge)}}\left( x \right)$ can be rewritten as
\begin{align}\label{eqn:mellin_t_corr_gpdf_sum_conf_re}
 &{F_G^{(2\ge)}}\left( x \right) = \frac{{{{\left( { - 1} \right)}^{\frac{1}{2}{N_r}\left( {{N_r} - 1} \right)}}}\prod\nolimits_{i = 1}^{N_r} {\Gamma \left( {i + \tau  + 1} \right)}}{{\det \left( {{{{\bf{R}}_r}^{N_t}}} \right)\Delta \left( {{{{\bf{R}}_r}^{ - 1}}} \right)\prod\nolimits_{j = 1}^{N_r} {\left( {N_t - j} \right)!} }}\sum\limits_{{\bs\sigma} \in {S_{N_r}}} { {{\rm{sgn}}\left( {{\bs\sigma }} \right)} }\prod\limits_{i = 1}^{N_r} {{r_{{\sigma _{i}}}}^{i + \tau+1 }}\notag\\
 &\quad \times\frac{1}{{2\pi {\rm{i}}}}\int\nolimits_{c - {\rm i}\infty }^{c + {\rm i}\infty } {\frac{{\Xi \left( {0, - 1,0,1} \right)\prod\nolimits_{i = 1}^{N_r} {\Xi \left( {1,1,\left({{\rho {r_{{\sigma _{i}}}}}}\right)^{-1},{i} + \tau  + 1} \right)} }}{{\Xi \left( {1, - 1,0,1} \right)}}{{\left(\frac{x}{\det \left( {{{\bf{R}}_r}} \right){{{{\rho }} }^{N_r}}} \right)}^{ - s}}ds},
 \end{align}
where $\det \left( {{{\bf{R}}_r}} \right) = {\prod\nolimits_{i = 1}^{N_r} {{r_{{\sigma _i}}}} }$.
Hence, ${F_G^{(2\ge)}}\left( x \right)$ is finally expressed in terms of the representation of the generalized Fox's H function as \eqref{eqn:mellin_t_corr_gpdf_sum_conf_recon}.
\subsection{$ F_G^{(2)}(x)$ for $N_t < N_r$}
Substituting \eqref{eq:semi_me_small_conf_fin} into \eqref{eqn:cdf_g_inverse} produces the expression of $F_G^ < \left( x \right)$ as
\begin{align}\label{eqn:F_G_small_inv}
{F_G^{(2<)}}\left( x \right) &= \frac{{{{\left( { - 1} \right)}^{{N_t}\left( {{N_r} - {N_t}} \right)}}{\rho ^{ - \frac{1}{2}{N_t}\left( {{N_t} + 1} \right)}}}}{{\Delta \left( {{{\bf{R}}_r}} \right)}}\sum\limits_{{\bs{\sigma }} \in {S_{{N_r}}}} {{\mathop{\rm sgn}} \left( {\bs{\sigma }} \right)\prod\limits_{i = 1}^{{N_t}} {{r_{{\sigma _i}}}^{{N_r} - {N_t} - 1}} \prod\limits_{i = {N_t} + 1}^{{N_r}} {{r_{{\sigma _i}}}^{i - {N_t} - 1}} } \notag\\
&\quad \times \frac{1}{{2\pi {\rm{i}}}}\int_{c - {\rm i}\infty }^{c + {\rm i}\infty } {\frac{{\Gamma \left( { - s} \right)}}{{\Gamma \left( {1 - s} \right)}}\prod\limits_{i = 1}^{{N_t}} {\Psi \left( {i,s + i + 1,\left({{\rho {r_{{\sigma _{i}}}}}}\right)^{-1}} \right)} {x^{ - s}}ds}.
\end{align}
By using the definition of $\Xi \left( {a,\alpha ,A,\varphi } \right)$, ${F_G^{(2<)}}\left( x \right)$ can be obtained after some rearrangements as
\begin{align}\label{eqn:F_G_small_conf}
  &{F_G^{(2<)}}\left( x \right) = \frac{{{{\left( { - 1} \right)}^{{N_t}\left( {{N_r} - {N_t}} \right)}}}}{{\Delta \left( {{{\bf{R}}_r}} \right)}}\sum\limits_{{\bs{\sigma }} \in {S_{{N_r}}}} {{\mathop{\rm sgn}} \left( {\bs{\sigma }} \right){\prod\limits_{i = 1}^{{N_t}} {{r_{{\sigma _i}}}^{{N_r} + i - {N_t} - 1}} \prod\limits_{i = {N_t} + 1}^{{N_r}} {{r_{{\sigma _i}}}^{i - {N_t} - 1}} } } \notag\\
  &\times \frac{1}{{2\pi {\rm{i}}}}\int_{c - {\rm i}\infty }^{c + {\rm i}\infty } {\frac{{\Xi \left( {0, - 1,0,1} \right)\prod\nolimits_{i = 1}^{{N_t}} {\Xi \left( {1,1,{{\left( {\rho {r_{{\sigma _i}}}} \right)}^{ - 1}},i} \right)} }}{{\Xi \left( {1, - 1,0,1} \right)}}{{\left( {\frac{x}{{\left(\prod\nolimits_{i = 1}^{{N_t}} {{r_{{\sigma _i}}}} \right)}{\rho ^{{N_t}}}}} \right)}^{ - s}}ds}.
\end{align}
Accordingly, ${F_G^{(2<)}}\left( x \right)$ is consequently obtained by recognizing the contour integral as the generalized Fox's H function as \eqref{eqn:F_G_small_conf_con}.

\section{Proof of Lemma 2}\label{app:zeta_asy}
Similar to (\ref{eqn:h_sigma_pro2}), ${\cal Y}_{{\bs{\sigma }},N,\upsilon }^{(2)}\left( x \right)$ can be rewritten by means of \cite[Property 2]{shi2017asymptotic} as
\begin{align}\label{eqn:zeta_rewritten_pro2}
{\cal Y}_{{\bs{\sigma }},N,\upsilon }^{(2)}\left( x \right)
&=Y_{N + 1,1}^{1,N}\left[ {\left. {\begin{array}{*{20}{c}}
{{{\left( {0,1,{{\left( {\rho {r_{{\sigma _i}}}} \right)}^{ - 1}},i + \upsilon  + 1} \right)}_{i = 1, \cdots ,N}},\left( {1,1,0,1} \right)}\\
{\left( {0,1,0,1} \right)}
\end{array}} \right|\frac{{\left( {\prod\nolimits_{i = 1}^N {{r_{{\sigma _i}}}} } \right){\rho ^N}}}{x}} \right]\notag\\
&= {\rho^{-\frac{1}{2}{{N\left( {N + 1} \right)}} - \left( {\upsilon+1} \right)N}}\prod\limits_{i = 1}^N{r_{{\sigma _{i}}}}^{ - i - \upsilon -1 }\notag \\
&\times\frac{1}{{2\pi {\rm{i}}}}\int\nolimits_{ - c - {\rm i}\infty }^{ - c + {\rm i}\infty } {\frac{{\Gamma \left( s \right)}}{{\Gamma \left( {1 + s} \right)}}\prod\limits_{i = 1}^N {\Psi \left( {i + \upsilon  + 1, - s + i + \upsilon  + 2;\left({{\rho {r_{{\sigma _{i}}}}}}\right)^{-1}} \right)} {x^s}ds}.
\end{align}
With \cite[eq.(9.210.2)]{gradshteyn1965table}, (\ref{eqn:zeta_rewritten_pro2}) is obtained as
\begin{align}\label{eqn:zeta_grad_rew}
&{\cal Y}_{{\bs{\sigma }},N,\upsilon }^{(2)}\left( x \right) = {\rho^{-\frac{1}{2}{{N\left( {N + 1} \right)}} - \left( {\upsilon+1} \right)N}}\prod\limits_{i = 1}^N{r_{{\sigma _{i}}}}^{ - i - \upsilon -1 }\times\notag\\
&\frac{1}{{2\pi {\rm{i}}}}\int\nolimits_{ - c - {\rm i}\infty }^{ - c + {\rm i}\infty } {\frac{{\Gamma \left( s \right)}}{{\Gamma \left( {1 + s} \right)}}\prod\limits_{i = 1}^N {\left( {\begin{array}{*{20}{l}}
{\frac{{\Gamma \left( {s - {i} - \upsilon  - 1} \right)}}{{\Gamma \left( s \right)}}{}_1{F_1}\left( {{i} + \upsilon  + 1, - s + {i} + \upsilon  + 2;\frac{1}{{\rho {r_{{\sigma _{i}}}}}}} \right) + }\\
{\frac{{\Gamma \left( {{i} + \upsilon  + 1 - s} \right)}}{{\Gamma \left( {{i} + \upsilon  + 1} \right)}}{{\left( \frac{1}{{\rho {r_{{\sigma _{i}}}}}} \right)}^{s - {i} - \upsilon  - 1}}{}_1{F_1}\left( {s,s - {i} - \upsilon ;\frac{1}{{\rho {r_{{\sigma _{i}}}}}}} \right)}
\end{array}} \right)} {x^s}ds}.
\end{align}
By recalling $c<0$, we set the value of $c$ as
\begin{align}\label{eqn:c_semi}
c &< \sup\left\{c:{\frac{1}{2}{{N\left( {N + 1} \right)}} + \left( {\upsilon  + 1} \right)N - c - i - \upsilon  - 1 > N_tN_r,\,1 \le i \le N} \right\}\notag\\
&= \left( {\frac{1}{2}N + \upsilon  + 1} \right)\left( {N - 1} \right) - N_tN_r.
\end{align}
On the basis \eqref{eqn:c_semi}, as $\rho\to \infty$, \eqref{eqn:zeta_grad_rew} can be simplified by ignoring the higher order terms $o\left( {{\rho ^{ - N_tN_r}}} \right)$ as
\begin{align}\label{eqn:zeta_grad_rew_cset}
&{\cal Y}_{{\bs{\sigma }},N,\upsilon }^{(2)}\left( x \right) = {\rho^{-\frac{1}{2}{{N\left( {N + 1} \right)}} - \left( {\upsilon+1} \right)N}}\prod\limits_{i = 1}^N{r_{{\sigma _{i}}}}^{ - i - \upsilon -1 }\frac{1}{{2\pi {\rm{i}}}}\int\nolimits_{ - c - {\rm i}\infty }^{ - c + {\rm i}\infty } \frac{{\Gamma \left( s \right)}}{{\Gamma \left( {1 + s} \right)}}\notag\\
&\quad \times{\prod\limits_{i = 1}^N {\frac{{\Gamma \left( {s - {i} - \upsilon  - 1} \right)}}{{\Gamma \left( s \right)}}{}_1{F_1}\left( {{i} + \upsilon  + 1, - s + {i} + \upsilon  + 2;\frac{1}{{\rho {r_{{\sigma _{i}}}}}}} \right)} {x^s}ds}+ o\left( {{\rho ^{ - N_tN_r}}} \right).
\end{align}
By putting the infinite series expansion of ${}_1{F_1}\left( {\alpha ,\beta ;x} \right)$\cite[eq.(9.210.1)]{gradshteyn1965table} into \eqref{eqn:zeta_grad_rew_cset} finally leads to (\ref{eqn:zeta_asy}). 

\section{Proof of Theorem \ref{the:pout_asy_semi}}\label{app:proof_asy_out_sin_corr}
\subsection{Asymptotic expression of ${p_{out}^{\rm semi}}$ for $N_t \ge N_r$}
Substituting (\ref{eqn:zeta_asy}) into \eqref{eqn:mellin_t_corr_gpdf_sum_conf_recon} along with \eqref{out_SC}, we have
\begin{align}\label{eqn:out_corr_sin_asyexp_sub}
&{p_{out}^{\rm semi}} = \frac{{{{\left( { - 1} \right)}^{\frac{1}{2}{N_r}\left( {{N_r} - 1} \right)}}{{{{\rho }}}^{-\frac{{{N_r}\left( {{N_r} + 1} \right)}}{2} - \left( {{N_t} - {N_r}} \right){N_r}}}}\prod\nolimits_{i = 1}^{N_r} {\Gamma \left( {{ i} + \tau  + 1} \right)}}{{\det \left( {{{{\bf{R}}_r}^{N_t}}} \right)\Delta \left( {{{{\bf{R}}_r}^{ - 1}}} \right)\prod\nolimits_{j = 1}^{N_r} {\left( {{N_t} - j} \right)!} }}\sum\limits_{{\bs\sigma} \in {S_{N_r}}} { {{\rm{sgn}}\left( {{\bs\sigma}} \right)} } \notag\\
&\quad\times \frac{1}{{2\pi {\rm{i}}}}\int\nolimits_{ - c - {\rm i}\infty }^{ - c + {\rm i}\infty } {\frac{{\Gamma \left( s \right)}}{{\Gamma \left( {1 + s} \right)}}\prod\limits_{i = 1}^{N_r} {\frac{{\Gamma \left( {s - {i} - \tau  - 1} \right)}}{{\Gamma \left( s \right)}}\sum\limits_{{n_i} = 0}^\infty  {\frac{{{{\left( {{i} + \tau  + 1} \right)}_{{n_i}}}}}{{{{\left( { - s + {i} + \tau  + 2} \right)}_{{n_i}}}}}\frac{{{{\left( {{{\rho {r_{{\sigma _{i}}}}}}} \right)}^{{-n_i}}}}}{{{n_i}!}}} } {2^{Rs}}ds}  \notag\\
&\quad+ o\left( {{\rho ^{ - {N_t}{N_r}}}} \right).
\end{align}
By interchanging the order of summations with multiplications, then with integrations, it follows that
\begin{align}\label{eqn:out_corr_sin_asyexp_exchange}
{p_{out}^{\rm semi}} &= \frac{{{{\left( { - 1} \right)}^{\frac{1}{2}{N_r}\left( {{N_r} - 1} \right)}}{{{{\rho }}}^{-\frac{{{N_r}\left( {{N_r} + 1} \right)}}{2} - \left( {{N_t} - {N_r}} \right){N_r}}}}}{{\det \left( {{{{\bf{R}}_r}^{N_t}}} \right)\Delta \left( {{{{\bf{R}}_r}^{ - 1}}} \right)\prod\nolimits_{j = 1}^{N_r} {\left( {{N_t} - j} \right)!} }}\sum\limits_{{\bs\sigma} \in {S_{N_r}}} { {{\rm{sgn}}\left( {{\bs\sigma}} \right)} } \notag\\
&\quad \times\sum\limits_{{n_1}, \cdots ,{n_{N_r}} = 0}^\infty  {\prod\limits_{i = 1}^{N_r} {\Gamma \left( {{i} + \tau  + 1 + {n_i}} \right)\frac{{{{\left( {{{\rho {r_{{\sigma _{i}}}}}}} \right)}^{{-n_i}}}}}{{{n_i}!}}} }\notag\\
&\quad \times\frac{1}{{2\pi {\rm{i}}}}\int\nolimits_{ - c - {\rm i}\infty }^{ - c + {\rm i}\infty } {\frac{{\Gamma \left( s \right)}}{{\Gamma \left( {1 + s} \right)}}\prod\limits_{i = 1}^{N_r} {\frac{{\Gamma \left( {s - {i} - \tau  - 1} \right)}}{{\Gamma \left( s \right){{\left( { - s + {i} + \tau  + 2} \right)}_{{n_i}}}}}} {2^{Rs}}ds}  + o\left( {{\rho ^{ - {N_t}{N_r}}}} \right).
\end{align}
By using ${\left( { - s + {i} + \tau  + 2} \right)_{{n_i}}} = {\left( { - 1} \right)^{{n_i}}}\left( {s - {i} - \tau  - 1 - {n_i}} \right) \cdots \left( {s - {i} - \tau  - 2} \right)$ and the recursion relationship of Gamma function as $\Gamma(s+1) = s \Gamma(s)$, (\ref{eqn:out_corr_sin_asyexp_exchange}) can be further expressed as
\begin{align}\label{eqn:out_corr_sin_asyexp_recur}
{p_{out}^{\rm semi}} &= \frac{{{{\left( { - 1} \right)}^{\frac{1}{2}{N_r}\left( {{N_r} - 1} \right)}}{{{{\rho }}}^{-\frac{{{N_r}\left( {{N_r} + 1} \right)}}{2} - \left( {{N_t} - {N_r}} \right){N_r}}}}}{{\det \left( {{{{\bf{R}}_r}^{N_t}}} \right)\Delta \left( {{{{\bf{R}}_r}^{ - 1}}} \right)\prod\nolimits_{j = 1}^{N_r} {\left( {{N_t} - j} \right)!} }}\sum\limits_{{\bs\sigma} \in {S_{N_r}}} { {{\rm{sgn}}\left( {{\bs\sigma}} \right)} } \notag\\
&\quad \times\sum\limits_{{n_1}, \cdots ,{n_{N_r}} = 0}^\infty  {\prod\limits_{i = 1}^{N_r} {\Gamma \left( {{i} + \tau  + 1 + {n_i}} \right)\frac{{{{\left( -{{{\rho {r_{{\sigma _{i}}}}}}} \right)}^{{-n_i}}}}}{{{n_i}!}}} }\notag\\
&\quad\times\frac{1}{{2\pi {\rm{i}}}}\int\nolimits_{ - c - {\rm i}\infty }^{ - c + {\rm i}\infty } {\frac{{\Gamma \left( s \right)}}{{\Gamma \left( {1 + s} \right)}}\prod\limits_{i = 1}^{N_r} {\frac{{\Gamma \left( {s - \tau- {i} - {n_i}- 1 } \right)}}{{\Gamma \left( s \right)}}} {2^{Rs}}ds}  + o\left( {{\rho ^{ - {N_t}{N_r}}}} \right).
\end{align}
By interchanging the order of summations further together with the definition of ${\mathfrak g_{{\bs \sigma }}}\left(x \right)$, we get
\begin{align}\label{eqn:out_corr_sin_asyexp_gdef}
&{p_{out}^{\rm semi}} = \frac{{{{\left( { - 1} \right)}^{\frac{1}{2}{N_r}\left( {{N_r} - 1} \right)}}{{{{\rho }}}^{-\frac{{{N_r}\left( {{N_r} + 1} \right)}}{2} - \left( {{N_t} - {N_r}} \right){N_r}}}}}{{\det \left( {{{{\bf{R}}_r}^{N_t}}} \right)\Delta \left( {{{{\bf{R}}_r}^{ - 1}}} \right)\prod\nolimits_{j = 1}^{N_r} {\left( {{N_t} - j} \right)!} }}\sum\limits_{{n_1}, \cdots ,{n_{N_r}} = 0}^\infty  {\prod\limits_{i = 1}^{N_r} {\Gamma \left( {{i} + \tau  + 1 + {n_i}} \right)\frac{{{{\left( -{{{\rho }}} \right)}^{{-n_i}}}}}{{{n_i}!}}} }\notag\\
&\times
{\mathfrak g_{{\bf n + 1 }}}\left(2^R \right)\sum\limits_{{\bs\sigma} \in {S_{N_r}}} {{\rm{sgn}}\left( {{\bs\sigma}} \right)\prod\limits_{i = 1}^{N_r} {{{{{{{r_{{\sigma _{i}}}}}}} }^{{-n_i}}}} }+ o\left( {{\rho ^{ - {N_t}{N_r}}}} \right),
\end{align}
where ${\bf 1}$ denotes a $1 \times {N_r}$ all-one vector and ${\bf n} = ({n_1,\cdots,n_{N_r}})$. Since the following identity holds
\begin{equation}\label{eqn:vand_R_corr}
\sum\limits_{{\bs \sigma} \in {S_{N_r}}} {{\rm{sgn}}\left( {{\bs\sigma}} \right)\prod\limits_{i = 1}^{N_r} {{{ {{{{r_{{\sigma _{i}}}}}}}}^{{-n_i}}}} }  = \det \left( {{{\left\{ {{{ {{{{r_{{j}}}}}}}^{{-n_i}}}} \right\}}_{1 \le i,j \le {N_r}}}} \right),
\end{equation}
\eqref{eqn:vand_R_corr} is clearly equal to zero if there exist $k$ and $l$ such that $n_k=n_l$. Hence, the index vector ${\bf n}$ for the dominant terms in (\ref{eqn:out_corr_sin_asyexp_gdef}) belongs to the set of the permutations of $\{0,1,\cdots,{N_r}-1\}$, i.e., ${\Omega _{N_r}}$, so as to ensure the term of \eqref{eqn:vand_R_corr} non-zero. Moreover, from \eqref{eqn:out_corr_sin_asyexp_gdef}, the terms with order larger than $N_t N_r$ vanish comparing to the dominant terms with order smaller than or equal to $N_t N_r$ as $\rho \to \infty$. Hence, the summation term is either zero or the order of $\rho$ larger than $N_t N_r$ if ${\bf{n}} \notin  {\Omega _{N_r}}$. More specifically, by substituting \eqref{eqn:vand_R_corr} into \eqref{eqn:out_corr_sin_asyexp_gdef} and merely keeping the dominant-terms, it follows that
\begin{align}\label{eqn:out_corr_sin_asyexp_gpermu}
{p_{out}^{\rm semi}} &= \frac{{{{\left( { - 1} \right)}^{\frac{1}{2}{N_r}\left( {{N_r} - 1} \right)}}{{{{\rho }}}^{-\frac{{{N_r}\left( {{N_r} + 1} \right)}}{2} - \left( {{N_t} - {N_r}} \right){N_r}}}}}{{\det \left( {{{{\bf{R}}_r}^{N_t}}} \right)\Delta \left( {{{{\bf{R}}_r}^{ - 1}}} \right)\prod\nolimits_{j = 1}^{N_r} {\left( {{N_t} - j} \right)!} }}\sum\limits_{{\bf{n}} \in {\Omega _{N_r}}}  {\frac{{{{\left( { - {\rho }} \right)}^{-\sum\nolimits_{i = 1}^{N_r} {{n_i}} }}}\prod\nolimits_{i = 1}^{N_r} {\Gamma \left( {{i} + \tau  + 1 + {n_i}} \right)}}{{\prod\nolimits_{i = 1}^{N_r} {{n_i}!} }} }\notag\\
&\quad \times \det \left( {{{\left\{ {{{ {{{{r_{{j}}}}}}}^{{-n_i}}}} \right\}}_{1 \le i,j \le {N_r}}}} \right) {\mathfrak g_{{\bf n + 1 }}}\left(2^R \right) + o\left( {{\rho ^{ - {N_t}{N_r}}}} \right).
\end{align}
If ${\bf{n}} \in  {\Omega _{N_r}}$, we have $\det \left( {{{\left\{ {{{ {{{{r_{{j}}}}}}}^{{-n_i}}}} \right\}}_{1 \le i,j \le {N_r}}}} \right) = {\rm{sgn}}\left( {\bf{n}} \right)\Delta \left( {{{{\bf{R}}_r}^{ - 1}}} \right)$ follows from Vandermonde determinant and ${\sum\nolimits_{i = 1}^{N_r} {{n_i}} }={{N_r}({N_r}-1)}/{2}$. Therefore, (\ref{eqn:out_corr_sin_asyexp_gpermu}) can be further derived as
\begin{align}\label{eqn:out_corr_sin_asyexp_vand}
{p_{out}^{\rm semi}}&= \frac{{{{{{\rho }}}^{-{N_t}{N_r}}}}}{{\det \left( {{{{\bf{R}}_r}^{N_t}}} \right)\prod\nolimits_{i = 1}^{N_r} {\left( {{N_t} - i} \right)!\left( {{N_r} - i} \right)!} }}\notag\\
&\quad \times\sum\limits_{{\bf{n}} \in {\Omega _{N_r}}} {\rm{sgn}}\left( {\bf{n}} \right) \prod\limits_{i = 1}^{N_r} {\Gamma \left( {{i} + \tau  + 1 + {n_i}} \right)} {\mathfrak g_{{\bf n + 1 }}}\left(2^R \right) + o\left( {{\rho ^{ - {N_t}{N_r}}}} \right).
\end{align}
By redefining ${{\bs \sigma }} \triangleq {\bf{n}} + {\bf{1}}$, we get ${{\bs \sigma} \in {S_{N_r}}}$. As a consequence, (\ref{eqn:out_single_side_corr_asy}) follows for $N_t\ge N_r$.

\subsection{Asymptotic expression of ${p_{out}^{\rm semi}}$ for $N_t < N_r$}
By putting (\ref{eqn:zeta_asy}) into \eqref{eqn:F_G_small_conf_con} together with \eqref{out_SC}, after some algebraic manipulations, the outage probability for $N_t < N_r$, i.e., $p_{out}^{{\rm{semi}}}$, can be simplified as
\begin{align}\label{eqn:psemi_small_asy}
p_{out}^{{\rm{semi}}} &= \frac{{{{\left( { - 1} \right)}^{{N_t}\left( {{N_r} - {N_t}} \right)}}{\rho ^{ - \frac{1}{2}{N_t}\left( {{N_t} + 1} \right)}}}}{{\Delta \left( {{{\bf{R}}_r}} \right)\prod\nolimits_{i = 1}^{{N_t}} { \left( N_t -i \right)!} }}\notag\\
&\times \sum\limits_{{\bs{\sigma }} \in {S_{{N_r}}}} {{\mathop{\rm sgn}} \left( {\bs{\sigma }} \right)\prod\limits_{i = {N_t} + 1}^{{N_r}} {{r_{{\sigma _i}}}^{i - {N_t} - 1}} \sum\limits_{{n_1}, \cdots ,{n_{{N_t}}} = 0}^\infty  {\prod\limits_{i = 1}^{{N_t}} {\frac{{{r_{{\sigma _i}}}^{{N_r} - {N_t} - {n_i} - 1}{{\left( { - \rho } \right)}^{ - {n_i}}}\Gamma \left( {i + {n_i}} \right)}}{{{n_i}!}}} } }\notag\\
&\times \frac{1}{{2\pi {\rm{i}}}}\int_{ - c - {\rm{i}}\infty }^{ - c + {\rm{i}}\infty } {\frac{{\Gamma \left( s \right)}}{{\Gamma \left( {1 + s} \right)}}\prod\limits_{i = 1}^{{N_t}} {\frac{{\Gamma \left( {s - i - {n_i}} \right)}}{{\Gamma \left( s \right)}}} {2^{Rs}}ds}  + o\left( {{\rho ^{ - {N_t}{N_r}}}} \right).
\end{align}
By expressing the integral in \eqref{eqn:psemi_small_asy} in terms of ${\mathfrak g_{{\bs \sigma }}}\left(x \right)$ and swapping the order of the summations, we get
\begin{align}\label{eqn:psemi_small_asy_sw}
&p_{out}^{{\rm{semi}}} = \frac{{{{\left( { - 1} \right)}^{{N_t}\left( {{N_r} - {N_t}} \right)}}{\rho ^{ - \frac{1}{2}{N_t}\left( {{N_t} + 1} \right)}}}}{{\Delta \left( {{{\bf{R}}_r}} \right)\prod\nolimits_{i = 1}^{{N_t}} { \left( N_t-i \right)!} }}\sum\limits_{{n_1}, \cdots ,{n_{{N_t}}} = 0}^\infty  {\frac{{{{\left( { - \rho } \right)}^{ - \sum\nolimits_{i = 1}^{{N_t}} {{n_i}} }}\prod\nolimits_{i = 1}^{{N_t}} {\Gamma \left( {i + {n_i}} \right)} }}{{\prod\nolimits_{i = 1}^{{N_t}} {{n_i}!} }}} {\mathfrak g_{{\bf{n}} - \tau {\bf{1}}}}\left( {{2^R}} \right)\notag\\
&\quad \times\sum\limits_{{\bs{\sigma }} \in {S_{{N_r}}}} {{\mathop{\rm sgn}} \left( {\bs{\sigma }} \right)\prod\limits_{i = 1}^{{N_t}} {{r_{{\sigma _i}}}^{{N_r} - {N_t} - {n_i} - 1}} \prod\limits_{i = {N_t} + 1}^{{N_r}} {{r_{{\sigma _i}}}^{i - {N_t} - 1}} }  + o\left( {{\rho ^{ - {N_t}{N_r}}}} \right)\notag\\
& = \frac{{{{\left( { - 1} \right)}^{{N_t}\left( {{N_r} - {N_t}} \right)}}{\rho ^{ - \frac{1}{2}{N_t}\left( {{N_t} + 1} \right)}}}}{{\Delta \left( {{{\bf{R}}_r}} \right)\prod\nolimits_{i = 1}^{{N_t}} { \left( N_t-i \right)!} }}\sum\limits_{{n_1}, \cdots ,{n_{{N_t}}} = 0}^\infty  {\frac{{{{\left( { - \rho } \right)}^{ - \sum\nolimits_{i = 1}^{{N_t}} {{n_i}} }}\prod\nolimits_{i = 1}^{{N_t}} {\Gamma \left( {i + {n_i}} \right)} }}{{\prod\nolimits_{i = 1}^{{N_t}} {{n_i}!} }}} \notag\\
&\quad \times \det \left( {\begin{array}{*{20}{l}}
{{{\left\{ {{r_j}^{{N_r} - {N_t} - {n_i} - 1}} \right\}}_{\scriptstyle1 \le i \le {N_t},\hfill\atop
\scriptstyle1 \le j \le {N_r}\hfill}}}\\
{{{\left\{ {{r_j}^{i - {N_t} - 1}} \right\}}_{\scriptstyle{N_t} + 1 \le i \le {N_r},\hfill\atop
\scriptstyle1 \le j \le {N_r}\hfill}}}
\end{array}} \right){\mathfrak g_{{\bf{n}} - \tau {\bf{1}}}}\left( {{2^R}} \right) + o\left( {{\rho ^{ - {N_t}{N_r}}}} \right),
\end{align}
where ${\bf{n}} = \left( {{n_1}, \cdots ,{n_{{N_t}}}} \right)$ and the last step holds by using the Laplace expansion of the determinant. Similar to \eqref{eqn:out_corr_sin_asyexp_gpermu}, the index vector ${\bf n}$ for the dominant terms in \eqref{eqn:psemi_small_asy_sw} belongs to the set of the permutations of $\{{{N_r} - {N_t}, \cdots ,{N_r} - 1}\}$, i.e., ${\Theta  _{N_t}}$, so as to ensure the determinant non-zero. Considering that the summation term is either zero or the order of $\rho$ larger than $N_t N_r$ if ${\bf{n}} \notin  {\Theta _{N_t}}$. Hence, if ${\bf{n}} \in  {\Theta _{N_t}}$, we have
\begin{align}\label{eqn:out_semi_small_iden}
\det \left( {\begin{array}{*{20}{l}}
{{{\left\{ {{r_j}^{{N_r} - {N_t} - {n_i} - 1}} \right\}}_{\scriptstyle1 \le i \le {N_t},\hfill\atop
\scriptstyle1 \le j \le {N_r}\hfill}}}\\
{{{\left\{ {{r_j}^{i - {N_t} - 1}} \right\}}_{\scriptstyle{N_t} + 1 \le i \le {N_r},\hfill\atop
\scriptstyle1 \le j \le {N_r}\hfill}}}
\end{array}} \right) = {\left( { - 1} \right)^{{N_t}\left( {{N_r} - {N_t}} \right)}}\det \left( {\begin{array}{*{20}{l}}
{{{\left\{ {{r_j}^{{N_r} - {n_{{N_t} - i}} - {N_t} - 1}} \right\}}_{\scriptstyle1 \le i \le {N_t},\hfill\atop
\scriptstyle1 \le j \le {N_r}\hfill}}}\\
{{{\left\{ {{r_j}^{i - {N_t} - 1}} \right\}}_{\scriptstyle{N_t} + 1 \le i \le {N_r},\hfill\atop
\scriptstyle1 \le j \le {N_r}\hfill}}}
\end{array}} \right)\notag\\
 = {\left( { - 1} \right)^{{N_t}\left( {{N_r} - {N_t}} \right)}}{\mathop{\rm sgn}} \left( {\bf{n}} \right)\det \left( {{{\left\{ {{r_j}^{i - {N_t} - 1}} \right\}}_{\scriptstyle1 \le i \le {N_r}\hfill\atop
\scriptstyle1 \le j \le {N_r}\hfill}}} \right) = {\left( { - 1} \right)^{{N_t}\left( {{N_r} - {N_t}} \right)}}\frac{{{\mathop{\rm sgn}} \left( {\bf{n}} \right)\Delta \left( {{{\bf{R}}_r}} \right)}}{{\det \left( {{{\bf{R}}_r}^{{N_t}}} \right)}}.
\end{align}
By plugging \eqref{eqn:out_semi_small_iden} into \eqref{eqn:psemi_small_asy_sw}, it follows that
\begin{align}\label{eqn:p_semi_small_sub}
p_{out}^{{\rm{semi}}} &= \frac{{{\rho ^{ - {N_t}{N_r}}}}}{{\det \left( {{{\bf{R}}_r}^{{N_t}}} \right)\prod\nolimits_{i = 1}^{{N_t}} {\left( {{N_t} - i} \right)!\left( {{N_r} - i} \right)!} }}\notag\\
&\quad \times \sum\limits_{{\bf{n}} \in {\Theta _{{N_t}}}} {{\mathop{\rm sgn}} \left( {\bf{n}} \right)\prod\limits_{i = 1}^{{N_t}} {\Gamma \left( {i + {n_i}} \right)} {\mathfrak g_{{\bf{n}} - \tau {\bf{1}}}}\left( {{2^R}} \right)}  + o\left( {{\rho ^{ - {N_t}{N_r}}}} \right).
\end{align}
By redefining ${\bs{\sigma }} = {\bf{n}} - \tau {\bf{1}}$ in \eqref{eqn:p_semi_small_sub}, the outage probability of the semi-correlated Rayleigh MIMO channels for $N_t < N_r$ can be derived as \eqref{eqn:out_single_side_corr_asy} consequently.

\section{Proof of \eqref{eqn:mellin_fullc_tricomi}}\label{app:full_mellin}
By substituting \eqref{eqn:eig_pdf_full} into \eqref{eqn:mellin_G}, the Mellin transform of $f_G(x)$ under full-correlated Rayleigh MIMO channels can be obtained as
\begin{align}\label{eqn:f_G_fullcorr_mt}
\varphi_{\rm full}\left( s \right) = &\int\nolimits_0^\infty  { \cdots \int\nolimits_0^\infty  {\prod\limits_{i = 1}^{N_r} {{{\left( {1 + \rho{\lambda _i}} \right)}^{s - 1}}} } } \notag\\
&\times \sum\limits_{{{\bf{k}}_{N_r}}} {\frac{{{{\left( { - 1} \right)}^{\frac{{{N_r}\left( {{N_r} - 1} \right)}}{2}}}{\cal A}}}{{{N_r}!\Delta \left( {\bf{K}} \right)}}\Delta \left( {\bs{\lambda }} \right)\det \left( {{{\left\{ {{\lambda _i}^{{k_j} + {N_t} - {N_r}}} \right\}}_{ {i,j} }}} \right)} d{\lambda _1} \cdots d{\lambda _{N_r}}.
\end{align}
By using the identity
\begin{equation}\label{eqn:delta_lambda}
\Delta \left( {\bs{\lambda }} \right) = \det \left( {{{\left\{ {{\lambda _i}^{j-1}} \right\}}_{{i,j}}}} \right) = \det \left( {{{\left\{ {{{\left( {\frac{{{\lambda _i}}}{{1 + \rho{\lambda _i}}}} \right)}^{j-1}}} \right\}}_{{i,j}}}} \right)\prod\limits_{i = 1}^{N_r} {{{\left( {1 + \rho{\lambda _i}} \right)}^{{N_r} - 1}}}
\end{equation}
together with the determinant expansion, (\ref{eqn:f_G_fullcorr_mt}) can be rewritten as
\begin{align}\label{eqn:mellin_fullc_rew}
\varphi_{\rm full}\left( s \right) = &\int\nolimits_0^\infty  { \cdots \int\nolimits_0^\infty  {\prod\limits_{i = 1}^{N_r} {{{\left( {1 + \rho{\lambda _i}} \right)}^{s + {N_r} - 2}}} } } \notag\\
&\times \sum\limits_{{{\bf{k}}_{N_r}}} {\frac{{{{\left( { - 1} \right)}^{\frac{{{N_r}\left( {{N_r} - 1} \right)}}{2}}}{\cal A}}}{{{N_r}!\Delta \left( {\bf{K}} \right)}}\sum\limits_{{{\bs{\sigma }}_1} \in {S_{N_r}}} {{\mathop{\rm sgn}} \left( {{{\bs{\sigma }}_1}} \right)\prod\limits_{i = 1}^{N_r} {{{\left( {\frac{{{\lambda _i}}}{{1 + \rho{ \lambda _i}}}} \right)}^{{\sigma _{1,i}}-1}}} } } \notag\\
&\times \sum\limits_{{{\bs{\sigma }}_2} \in {S_{N_r}}} {{\mathop{\rm sgn}} \left( {{{\bs{\sigma }}_2}} \right)\prod\limits_{i = 1}^{N_r} {{\lambda _i}^{{k_{{\sigma _{2,i}}}} + {N_t} - {N_r}}} } d{\lambda _1} \cdots d{\lambda _{N_r}}.
\end{align}
By interchanging the order of integration and summation in (\ref{eqn:mellin_fullc_rew}), it follows that
\begin{align}\label{eqn:mell_fullC_rea}
\varphi_{\rm full}\left( s \right) =& \sum\limits_{{{\bf{k}}_{N_r}}} {\frac{{{{\left( { - 1} \right)}^{\frac{{{N_r}\left( {{N_r} - 1} \right)}}{2}}}{\cal A}}}{{{N_r}!\Delta \left( {\bf{K}} \right)}}\sum\limits_{{{\bs{\sigma }}_1},{{\bs{\sigma }}_2} \in {S_{N_r}}} { {{\mathop{\rm sgn}} \left( {{{\bs{\sigma }}_1}} \right){\mathop{\rm sgn}} \left( {{{\bs{\sigma }}_2}} \right)} } } \notag\\
&\times \prod\limits_{i = 1}^{N_r} {\int\nolimits_0^\infty  {{{\left( {1 + \rho\lambda } \right)}^{s + N_r - {\sigma _{1,i}} - 1}}{\lambda ^{{k_{{\sigma _{2,i}}}} + {\sigma _{1,i}} + \tau}}d\lambda } }.
\end{align}
According to Lemma \ref{the:leb_for}, (\ref{eqn:mell_fullC_rea}) can be simplified as
\begin{align}\label{eqn:mellin_trans_full_leibniz}
\varphi_{\rm full}\left( s \right) &= \sum\limits_{{{\bf{k}}_{N_r}}} {\frac{{{{\left( { - 1} \right)}^{\frac{{{N_r}\left( {{N_r} - 1} \right)}}{2}}}{\cal A}}}{{\Delta \left( {\bf{K}} \right)}}\sum\limits_{{\bs{\sigma }} \in {S_{N_r}}} {{\mathop{\rm sgn}} \left( {\bs{\sigma }} \right)\prod\limits_{i = 1}^{N_r} {\int\nolimits_0^\infty  {{{\left( {1 + \rho\lambda } \right)}^{s +N_r-i - 1}}{\lambda ^{{k_{{\sigma _i}}} + i + {\tau}}}d\lambda } } } } \notag\\
 &= \sum\limits_{{{\bf{k}}_{N_r}}} {\frac{{{{\left( { - 1} \right)}^{\frac{{{N_r}\left( {{N_r} - 1} \right)}}{2}}}{\cal A}}}{{\Delta \left( {\bf{K}} \right)}}\det \left( {{{\left\{ {\int\nolimits_0^\infty  {{{\left( {1 + \rho\lambda } \right)}^{s + N_r - i - 1}}{\lambda ^{{k_j} + i + {\tau}}}d\lambda } } \right\}}_{{i,j}}}} \right)},
\end{align}
where the second equality holds by using the Leibniz formula for the determinant expansion \cite{horn2012matrix}. By the change of variable $x = 1/{{(1 + \rho \lambda )}}$, the integral in \eqref{eqn:mellin_trans_full_leibniz} can be expressed in terms of Beta function ${\rm B}(\alpha,\beta)$ as
\begin{equation}\label{eqn:beta_fun}
  \int_0^\infty  {{{\left( {1 + \rho \lambda } \right)}^{s + {N_r} - i - 1}}{\lambda ^{{k_j} + i + \tau }}d\lambda } = {\rho ^{ - {k_j} - i - \tau  - 1}}{\rm{B}}\left( { - s - {N_r} - {k_j} - \tau ,{k_j} + i + \tau  + 1} \right).
\end{equation}
By using the relationship between beta function and Gamma function as ${\rm B}(\alpha,\beta) = \Gamma(\alpha)\Gamma(\beta)/\Gamma(\alpha+\beta)$, the following identity holds
\begin{align}\label{eqn:beta_fundeter}
&\det \left( {{{\left\{ {\int\nolimits_0^\infty  {{{\left( {1 + \rho\lambda } \right)}^{s + N_r - i - 1}}{\lambda ^{{k_j} + i + {\tau}}}d\lambda } } \right\}}_{{i,j}}}} \right) \notag\\
&= \prod\limits_{i = 1}^{{N_r}} {\frac{{{\rho ^{ - {k_i} - i - \tau  - 1}}\Gamma \left( { - s - {N_r} - {k_i} - \tau } \right)\Gamma \left( {{k_i} + \tau  + 2} \right)}}{{\Gamma \left( { - s - {N_r} + i + 1} \right)}}} \det \left( {{{\left\{ {{{\left( {{k_j}  + \tau  + 2} \right)}_{i - 1}}} \right\}}_{i,j}}} \right).
\end{align}
Notice that $\det \left( {{{\left\{ {{{\left( {{k_j} + \tau  + 2} \right)}_{i - 1}}} \right\}}_{i,j}}} \right) = \Delta \left( {\bf{K}} \right)$, putting \eqref{eqn:beta_fundeter} into \eqref{eqn:mellin_trans_full_leibniz} leads to
\begin{equation}\label{eqn:gamma_vect_varphi_semi_full}
\varphi_{\rm full}\left( s \right) = \sum\limits_{{{\bf{k}}_{{N_r}}}} {{{\left( { - 1} \right)}^{\frac{{{N_r}\left( {{N_r} - 1} \right)}}{2}}}\mathcal A\prod\limits_{i = 1}^{{N_r}} {\frac{{{\rho ^{ - {k_i} - i - \tau  - 1}}\Gamma \left( { - s - {N_r} - {k_i} - \tau } \right)\Gamma \left( {{k_i} + \tau  + 2} \right)}}{{\Gamma \left( { - s - {N_r} + i + 1} \right)}}} }.
\end{equation}

According to the definition of $\mathcal A$, \eqref{eqn:gamma_vect_varphi_semi_full} can be further simplified by using the generalized Cauchy-Binet formula \cite[Lemma 4]{ghaderipoor2012application} as
\begin{align}\label{eqn:mel_fc_defa_sub}
\varphi_{\rm full}\left( s \right) 
& = {\left( { - 1} \right)^{\frac{{{N_r}\left( {{N_r} - 1} \right)}}{2}}}\frac{{\prod\nolimits_{i = 1}^{{N_r}} {{a_i}^{{N_t}}} \prod\nolimits_{j = 1}^{{N_t}} {{b_j}^{{N_r}}} }}{{\Delta \left( {\bf{A}} \right)\Delta \left( {\bf{B}} \right)}}\prod\limits_{i = 1}^{{N_r}} {\frac{{{\rho ^{ - i}}}}{{\Gamma \left( { - s - {N_r} + i + 1} \right)}}}\notag \\
&\quad \times\sum\limits_{{{\bf{k}}_{{N_r}}}} {\det \left( {{{\left\{ {{{\left( { - {a_i}} \right)}^{{k_j}}}} \right\}}_{i,j}}} \right)\det \left( {{{\left\{ {{b_i}^{{k_j} + N - M}} \right\}}_{i,1 \le j \le M}},{{\left\{ {{b_i}^{N - j}} \right\}}_{i,M + 1 \le j \le N}}} \right) }\notag\\
&\quad\times \prod\limits_{i = 1}^{{N_r}} {{\rho ^{ - \left( {{k_i} + {N_t} - {N_r}} \right)}}\Gamma \left( { - s - {N_r} - \left( {{k_i} + {N_t} - {N_r}} \right) + 1} \right)}\notag\\
& = \frac{{{{\left( { - 1} \right)}^{{N_r}\left( {{N_t} - {N_r}} \right)}}{\rho ^{ - \frac{1}{2}{N_r}\left( {{N_r} + 1} \right)}}\prod\nolimits_{i = 1}^{{N_r}} {{a_i}^{{N_r}}} \prod\nolimits_{j = 1}^{{N_t}} {{b_j}^{{N_r}}} }}{{\Delta \left( {\bf{A}} \right)\Delta \left( {\bf{B}} \right)\prod\nolimits_{i = 1}^{{N_r}} {{{\left( {s + i - 2} \right)}^{i - 1}}} }}\notag\\
&\quad \times \det \left( \begin{array}{l}
{\left\{ {\sum\limits_{k = 0}^\infty  {\frac{{\Gamma \left( { - s - {N_r} - k + 1} \right)}}{{\Gamma \left( { - s - {N_r} + 2} \right)}}{{\left( {-\frac{{{a_i}{b_j}}}{\rho }} \right)}^k}} } \right\}_{1 \le i \le {N_r},j}}\\
{\left\{ {{b_j}^{{N_t} - i}} \right\}_{{N_r} + 1 \le i \le {N_t},j}}
\end{array} \right).
\end{align}
By using $\mathrm B\left( {\alpha ,\beta } \right) = \int_0^1 {{x^{\alpha  - 1}}{{\left( {1 - x} \right)}^{\beta  - 1}}dx} $, the infinite series inside the determinant in \eqref{eqn:mel_fc_defa_sub} can be rewritten as
\begin{align}\label{eqn:inte_tric}
\sum\limits_{k = 0}^\infty  {\frac{{\Gamma \left( { - s - {N_r} - k + 1} \right)}}{{\Gamma \left( { - s - {N_r} + 2} \right)}}{{\left( { - \frac{{{a_i}{b_j}}}{\rho }} \right)}^k}} &= \sum\limits_{k = 0}^\infty  {{\rm{B}}\left( { - s - {N_r} - k + 1,k + 1} \right)\frac{1}{{k!}}} {\left( { - \frac{{{a_i}{b_j}}}{\rho }} \right)^k}\notag\\
 &=\int_0^\infty  {{{\left( {1 + y} \right)}^{s + {N_r} - 2}}{e^{ - \frac{{{a_i}{b_j}}}{\rho }y}}dy}\notag\\
 &=\Psi \left( {1,s + {N_r};\frac{{{a_i}{b_j}}}{\rho }} \right).
\end{align}
where the last equality holds by using the integral representation of the Tricomi's confluent hypergeometric function.
%
%
By substituting \eqref{eqn:inte_tric} into (\ref{eqn:mel_fc_defa_sub}), we finally arrive at \eqref{eqn:mellin_fullc_tricomi}. 

\section{Proof of Theorem \ref{the:full_cdf}}\label{app:full_cdf}
By applying the determinant expansion to \eqref{eqn:mellin_fullc_tricomi}, we get
\begin{align}\label{eqn:cdf_G_fullc}
&{F_G}\left( x \right) 
 = \frac{{{{\left( { - 1} \right)}^{{N_r}\left( {{N_t} - {N_r}} \right) + \frac{1}{2}{{{N_r}\left( {{N_r} - 1} \right)}}}}{\rho^{-\frac{1}{2}{{{N_r}\left( {{N_r} + 1} \right)}}}}\prod\nolimits_{i = 1}^{N_r} {{a_i}^{N_r}} \prod\nolimits_{j = 1}^{N_t} {{b_j}^{N_r}} }}{{\Delta \left( {\bf{A}} \right)\Delta \left( {\bf{B}} \right)}}\sum\limits_{{\bs{\sigma }} \in {S_{N_t}}} {{\mathop{\rm sgn}} \left( {\bs{\sigma }} \right){\prod\limits_{i = {N_r} + 1}^{N_t} {{b_{{\sigma _i}}}^{{N_t} - i}} }} \notag\\
&\times\frac{1}{{2\pi {\rm{i}}}}\int\nolimits_{c - {\rm i}\infty }^{c + {\rm i}\infty } {\frac{{\Gamma \left( { - s} \right)}}{{\Gamma \left( {1 - s} \right)}}\prod\limits_{j = 1}^{{N_r}-1} {\frac{{\Gamma \left( { - s - {N_r} + 1} \right)}}{{\Gamma \left( { - s - j + 1} \right)}}\prod\limits_{i = 1}^{N_r}\Psi \left( {1,s + 1 + {N_r};\frac{{{a_{{i}}}{b_{\sigma _i}}}}{\rho }} \right){x^{ - s}}} ds}.
\end{align}
By using the definition of $\Xi \left( {a,\alpha ,A,\varphi } \right) = {A^{\varphi  + a + \alpha s - 1}}\Psi \left( {\varphi ,\varphi  + a + \alpha s;A} \right)$, (\ref{eqn:cdf_G_fullc}) can be rewritten as
\begin{multline}\label{eqn:cdf_G_fullc_intro_xi}
{F_G}\left( x \right)
 = \frac{{{{\left( { - 1} \right)}^{{N_r}\left( {{N_t} - {N_r}} \right) + \frac{1}{2}{{{N_r}\left( {{N_r} - 1} \right)}}}}}{{{{\rho }}}^{\frac{1}{2}{{{N_r}\left( {{N_r} - 1} \right)}}}}}{{\Delta \left( {\bf{A}} \right)\Delta \left( {\bf{B}} \right)}}\sum\limits_{{\bs{\sigma }} \in {S_{N_t}}} {{\mathop{\rm sgn}} \left( {\bs{\sigma }} \right)\prod\limits_{i = {N_r} + 1}^{N_t} {{b_{{\sigma _i}}}^{{N_t} + {N_r} - i}} }   \\
\times \frac{1}{{2\pi {\rm{i}}}}\int\nolimits_{c - {\rm i}\infty }^{c + {\rm i}\infty } {\frac{{\Xi \left( {0, - 1,0,1} \right)}}{{\Xi \left( {1, - 1,0,1} \right)}}\prod\limits_{j = 1}^{{N_r}-1} {\frac{{\Xi \left( { - {N_r} + 1, - 1,0,1} \right)}}{{\Xi \left( { - j + 1, - 1,0,1} \right)}}\prod\limits_{i = 1}^{N_r} {\Xi \left(  {{N_r},1,\frac{{{a_i}{b_{{\sigma _i}}}}}{\rho },1} \right)} } } \\
\times{{\left( {x {\prod\limits_{i = 1}^{N_r} {\frac{{{a_i}{b_{{\sigma _i}}}}}{\rho }} } } \right)}^{ - s}}ds.
\end{multline}
Accordingly, ${F_G}\left( x \right)$ can be expressed in terms of the generalized Fox's H function as \eqref{eqn:F_G_CDF_fullcorr}.

\section{Proof of Lemma \ref{the:asy_y_3}}\label{app:asy_y_3}
Applying property 2 in \cite{shi2017asymptotic} to (\ref{eqn:F_G_CDF_fullcorr}) gives rise to
\begin{align}\label{eqn:F_G_CDF_fullcora_p2}
&{\mathcal Y_{ {{\bs\sigma}}}^{(3)}(x)} \notag\\
&=Y_{2{N_r},{N_r}}^{{N_r},{N_r}}\left[ {\left. {\begin{array}{*{20}{c}}
{{{\left( {1 - {N_r},1,\frac{{{a_{{i}}}{b_{\sigma_i}}}}{\rho },1} \right)}_{i = 1, \cdots ,{N_r}}},\left( {1,1,0,1} \right),{{\left( {1 - j,1,0,1} \right)}_{j = 1, \cdots ,{N_r}-1}}}\\
{\left( {0,1,0,1} \right),{{\left( {2 - {N_r},1,0,1} \right)}_{j = 1, \cdots ,{N_r}-1}}}
\end{array}} \right|\frac{x^{-1}}{{{{ \prod\nolimits_{i = 1}^{N_r} {\frac{{{a_i}{b_{{\sigma _i}}}}}{\rho }} } }}}} \right]\notag\\
&={{{{\rho }}}^{-{N_r}^2}}\prod\limits_{i = 1}^{N_r} {{a_i}^{N_r}}{{b_{\sigma_i}}^{N_r}}\notag\\
&\quad \times\frac{1}{{2\pi {\rm{i}}}}\int\nolimits_{-c - {\rm i}\infty }^{-c + {\rm i}\infty } {\frac{{\Gamma \left( s \right)}}{{\Gamma \left( {1 + s} \right)}}\prod\limits_{j = 1}^{{N_r}-1}{\frac{{\Gamma \left( {s - {N_r} + 1} \right)}}{{\Gamma \left( {s - j + 1} \right)}}}\prod\limits_{i = 1}^{N_r} {\Psi \left( {1,1 + {N_r} - s;{\frac{{{a_i}{b_{{\sigma _i}}}}}{\rho }}} \right)}{x^s} ds}.
\end{align}
By using \cite[eq.(9.210.2)]{gradshteyn1965table}, we have
\begin{align}\label{eqn:F_G_full_gd}
{\mathcal Y_{ {{\bs\sigma}}}^{(3)}(x)} =& {{{{\rho }}}^{-{N_r}^2}}\prod\limits_{i = 1}^{N_r} {{a_i}^{N_r}}{{b_{\sigma_i}}^{N_r}} \frac{1}{{2\pi {\rm{i}}}}\int\nolimits_{-c - {\rm i}\infty }^{-c + {\rm i}\infty } {\frac{{\Gamma \left( s \right)}}{{\Gamma \left( {1 + s} \right)}}\prod\limits_{j = 1}^{{N_r}-1}{\frac{{\Gamma \left( {s - {N_r} + 1} \right)}}{{\Gamma \left( {s - j + 1} \right)}}}} \notag\\
&\times \prod\limits_{i = 1}^{N_r} {\left( \begin{array}{l}
\frac{{\Gamma \left( {s - {N_r}} \right)}}{{\Gamma \left( {s - {N_r} + 1} \right)}}{}_1{F_1}\left( {1,1 + {N_r} - s;\frac{{{a_{{ i}}}{b_{\sigma_i}}}}{\rho }} \right)
 + \\
 \Gamma \left( {{N_r} - s} \right){\left( {\frac{{{a_{{ i}}}{b_{\sigma_i}}}}{\rho }} \right)^{s - {N_r}}}{}_1{F_1}\left( {s - {N_r} + 1,s - {N_r} + 1;\frac{{{a_{{ i}}}{b_{\sigma_i}}}}{\rho }} \right)
\end{array} \right){x^s}ds}.
\end{align}

Similar to \eqref{eqn:zeta_grad_rew_cset}, we set $c < -{N_t}{N_r}+\frac{1}{2}{{N_r}({N_r}-1)}$. Thus ignoring the higher order terms $o\left( {{\rho ^{ - {N_t}{N_r}-\frac{1}{2}N_r(N_r+1)}}} \right)$ in \eqref{eqn:F_G_full_gd} yields
\begin{align}\label{eqn:Cdf_G_asb_fullc}
{\mathcal Y_{ {{\bs\sigma}}}^{(3)}(x)} =& {{{{\rho }}}^{-{N_r}^2}}\prod\limits_{i = 1}^{N_r} {{a_i}^{N_r}}{{b_{\sigma_i}}^{N_r}}\notag\\
&\times \frac{1}{{2\pi {\rm{i}}}}\int\nolimits_{-c - {\rm i}\infty }^{-c + {\rm i}\infty } {\frac{{\Gamma \left( s \right)}}{{\Gamma \left( {1 + s} \right)}}\prod\limits_{i = 1}^{N_r} {\frac{{\Gamma \left( {s - {N_r}} \right)}}{{\Gamma \left( {s - i + 1} \right)}}}  {{}_1{F_1}\left( {1,1 + {N_r} - s;\frac{{{a_{{ i}}}{b_{\sigma_i}}}}{\rho }} \right){x^s}ds} } \notag\\
 &+ o\left( {{\rho ^{ - {N_t}{N_r}-\frac{1}{2}N_r(N_r+1)}}} \right).
\end{align}
With the series expansion of ${}_1{F_1}\left( {\alpha ,\beta ;x} \right)$, \eqref{eqn:mellin_cdf_g_fullcsers_exp} follows.

\section{Proof of Theorem \ref{the:asy_out_full}}\label{app:asy_out_full}
Substituting \eqref{eqn:mellin_cdf_g_fullcsers_exp} into \eqref{eqn:out_excfull} and then swapping the orders of integration and summation produces
\begin{align}\label{eqn:cdfg_fullc_interch}
p_{out}^{\rm full} =& \frac{{{\left( { - 1} \right)}^{{N_r}\left( {{N_t} - {N_r}} \right) + \frac{1}{2}{{{N_r}\left( {{N_r} - 1} \right)}}}}{{{{{\rho }}}^{-\frac{1}{2}{{{N_r}\left( {{N_r} + 1} \right)}}}}\prod\nolimits_{i = 1}^{N_r} {{a_i}^{N_r}} \prod\nolimits_{j = 1}^{N_t} {{b_j}^{N_r}} }}{{\Delta \left( {\bf{A}} \right)\Delta \left( {\bf{B}} \right)}}\notag\\
&\times \sum\limits_{{n_1}, \cdots ,{n_{N_r}} = 0}^\infty  {\sum\limits_{{\bs{\sigma }} \in {S_{N_t}}} {{\mathop{\rm sgn}} \left( {\bs{\sigma }} \right){\prod\limits_{i = {N_r} + 1}^{N_t} {{b_{{\sigma _i}}}^{{N_t} - i}} }} } \notag\\
&\times \frac{1}{{2\pi {\rm{i}}}}\int\nolimits_{-c - {\rm i}\infty }^{-c + {\rm i}\infty } {\frac{{\Gamma \left( s \right)}}{{\Gamma \left( {1 + s} \right)}}\prod\limits_{i = 1}^{N_r} {\frac{{\Gamma \left( {s - {N_r}} \right)}}{{\Gamma \left( {s - i + 1} \right)}}}  {\frac{{{{\left( {\frac{{{a_{{ i}}}{b_{\sigma_i}}}}{\rho }} \right)}^{{n_{i}}}}}}{{{{\left( {1 + {N_r} - s} \right)}_{{n_{i}}}}}}} {2^{Rs}}ds}
 + o\left( {{\rho ^{ - {N_t}{N_r}}}} \right).
\end{align}
After some basic algebraic manipulations, (\ref{eqn:cdfg_fullc_interch}) can be further derived as
\begin{align}\label{eqn:F_G_mellin_trans_meijerg}
p_{out}^{\rm full} =& \frac{{{\left( { - 1} \right)}^{{N_r}\left( {{N_t} - {N_r}} \right) + \frac{1}{2}{{{N_r}\left( {{N_r} - 1} \right)}}}}{{{{{\rho }}}^{-\frac{1}{2}{{{N_r}\left( {{N_r} + 1} \right)}}}}\prod\nolimits_{i = 1}^{N_r} {{a_i}^{N_r}} \prod\nolimits_{j = 1}^{N_t} {{b_j}^{N_r}} }}{{\Delta \left( {\bf{A}} \right)\Delta \left( {\bf{B}} \right)}}\notag\\
 &\times \sum\limits_{{n_1}, \cdots ,{n_{N_r}} = 0}^\infty  {{\frac{1}{{2\pi {\rm{i}}}}\int\nolimits_{-c - {\rm i}\infty }^{-c + {\rm i}\infty } {\frac{{\Gamma \left( s \right)}}{{\Gamma \left( {1 + s} \right)}}\prod\limits_{i = 1}^{N_r} {\frac{{\Gamma \left( {s - {N_r} - {n_i}} \right)}}{{\Gamma \left( {s - i + 1} \right)}}}  {2^{Rs}}ds} }}\notag\\ %
 &\times \sum\limits_{{\bs{\sigma }} \in {S_{N_t}}} {{\mathop{\rm sgn}} \left( {\bs{\sigma }} \right)\prod\limits_{i = 1}^{N_r} {{{\left( { - \frac{{{a_{{ i}}}{b_{\sigma_i}}}}{\rho }} \right)}^{{n_{i}}}}} \prod\limits_{i = {N_r} + 1}^{N_t} {{b_{\sigma_i}}^{{N_t} - { i}}} }  + o\left( {{\rho ^{ - {N_t}{N_r}}}} \right).
\end{align}
By expressing the contour integral in \eqref{eqn:F_G_mellin_trans_meijerg} in terms of Meijer G-function and using the determinant expansion, $p_{out}^{\rm full}$ is given by
\begin{align}\label{eqn:F_G_mellin_trans_det}
p_{out}^{\rm full} =& \frac{{{{ {{\rho }} }^{-\frac{1}{2}{{{N_r}\left( {{N_r} + 1} \right)}}}}{{\left( { - 1} \right)}^{{N_r}\left( {{N_t} - {N_r}} \right) + \frac{1}{2}{{{N_r}\left( {{N_r} - 1} \right)}}}}\prod\nolimits_{i = 1}^{N_r} {{a_i}^{N_r}} \prod\nolimits_{j = 1}^{N_t} {{b_j}^{N_r}} }}{{\Delta \left( {\bf{A}} \right)\Delta \left( {\bf{B}} \right)}} \notag\\
 &\times \sum\limits_{{\bf{n}} \in {\mathbb {N} ^{N_r}}} {G_{{N_r} + 1,{N_r} + 1}^{0,{N_r} + 1}\left( {\left. {\begin{array}{*{20}{c}}
{1,1 + {N_r} + {n_{\rm{1}}}, \cdots ,1 + {N_r} + {n_{N_r}}}\\
{0,1, \cdots ,{N_r}}
\end{array}} \right|2^R} \right)} \notag\\
 & \times {\prod\limits_{i = 1}^{N_r} {{{\left( { - \frac{{{a_i}}}{\rho }} \right)}^{{n_i}}}} }\det \left( {\begin{array}{*{20}{c}}
{{{\left\{ {{b_j}^{{n_i}}} \right\}}_{1 \le i \le {N_r},j}}}\\
{{{\left\{ {{b_j}^{{N_t} - i}} \right\}}_{{N_r} + 1 \le i \le {N_t},j}}}
\end{array}} \right) + o\left( {{\rho ^{ - {N_t}{N_r}}}} \right),
\end{align}
where ${\bf n} = (n_1,\cdots,n_{N_r})$. Notice that any term with $n_i=0,\cdots,{N_t}-{N_r}-1$ is equal to zero thanks to the basic property of determinant, the dominant terms with $\bf n$ belonging to the set of the permutations of ${\Omega _{N_r}} = \left\{ {{N_t} - {N_r}, \cdots ,{N_t} - 1} \right\}$ can produce non-zero determinants. Ignoring the terms with both zero value of the determinant and the order of $\rho$ larger than $N_tN_r$, $p_{out}^{\rm full}$ can be asymptotically expanded as
\begin{align}\label{eqn:F_G_mellin_trans_detzero}
p_{out}^{\rm full} =&
 \frac{{{{{{\rho }} }^{-{N_t}{N_r}}}{\left( { - 1} \right)}^{\frac{1}{2}{{{N_r}\left( {{N_r} - 1} \right)}}}\prod\nolimits_{i = 1}^{N_r} {{a_i}^{N_r}} \prod\nolimits_{j = 1}^{N_t} {{b_j}^{N_r}} }}{{\Delta \left( {\bf{A}} \right)}}G_{{N_r} + 1,{N_r} + 1}^{0,{N_r} + 1}\left( {\left. {\begin{array}{*{20}{c}}
{1,{N_t} + 1, \cdots ,{N_t} + {N_r}}\\
{0,1, \cdots ,{N_r}}
\end{array}} \right|2^R} \right)\notag\\
&\times \sum\limits_{{\bf{n}} \in {\Omega _{N_r}}} {{\mathop{\rm sgn}} \left( {{\bf n}} \right)\prod\limits_{i = 1}^{N_r} {{a_i}^{{n_i}}} }  + o\left( {{\rho ^{ - {N_t}{N_r}}}} \right),
\end{align}
where the equality holds by using the following identity
\begin{equation}\label{eqn:det_full_sim}
\det \left( {\begin{array}{*{20}{c}}
{{{\left\{ {{b_j}^{{n_i}}} \right\}}_{1 \le i \le {N_r},j}}}\\
{{{\left\{ {{b_j}^{{N_t} - i}} \right\}}_{{N_r} + 1 \le i \le {N_t},j}}}
\end{array}} \right) = {\mathop{\rm sgn}} \left( {{\bf n}} \right) (-1)^{N_r(N_t-N_r)}{\Delta \left( {\bf{B}} \right)}.
\end{equation}
and the Meijer-G function can be extracted from the summation as a common factor due to the fact that its value is independent of the order of the elements of $\bf n$. Since the following equality holds
\begin{equation}\label{eqn:detea_identi}
\sum\limits_{{\bf{n}} \in {\Omega _{N_r}}} {{\mathop{\rm sgn}} \left( {{\bf n}} \right)\prod\limits_{i = 1}^{N_r} {{a_i}^{{n_i}}} }  = \det \left( {{{\left\{ {{a_i}^{{N_t} - j}} \right\}}_{{i,j}}}} \right) = {\left( { - 1} \right)}^{\frac{1}{2}{{{N_r}\left( {{N_r} - 1} \right)}}}\prod\limits_{i = 1}^{N_r} {{a_i}^{{N_t} - {N_r}}} \Delta \left( {\bf{A}} \right),
\end{equation}
the asymptotic expression of $p_{out}^{\rm full}$ can be finally derived as \eqref{eqn:F_G_mellin_trans_detzerofina}.

\section{Proof of Remark \ref{the:rem_g}}\label{app:rem_g}
By using the singular value decomposition of $\bf H$ as $\mathbf{H}=\mathbf{U}{\bs \Sigma} {\mathbf{V}^{\rm H}}$ and the Jacobian of the coordinate change, the joint PDF of the unordered strictly positive eigenvalues of $\mathbf{H}\mathbf{H}^{\rm H}$ for full-correlated Rayleigh MIMO channels can also be written as \cite[eq.(26)]{simon2006capacity}, \cite[eq.(6)]{ghaderipoor2012application}, \cite{zheng2002communication}
\begin{align}\label{eqn:lambda_joint}
{f_{\bs \lambda} }\left( {{\lambda _1}, \cdots ,{\lambda _{N_r}}} \right) = \frac{{{{\left( {\Delta \left( {\bs{\Lambda }} \right)} \right)}^2}\prod\nolimits_{j = 1}^{N_r} {{\lambda _j}^{{N_t} - {N_r}}} }}{{N_r}!{\prod\nolimits_{i = 1}^{N_r} {\left( {{N_r} - i} \right)!\left( {{N_t} - i} \right)!} }}\int {D{\bf{V}}\int {D{\bf{U}}p{\left(\mathbf{H}\right)}} },
\end{align}
where $\bs \Sigma = {\rm diag}(\sqrt{\lambda_1},\cdots,\sqrt{\lambda_{N_r}})$, $\mathbf{U} \in \mathcal U(N_t)$ and $\mathbf{V} \in \mathcal U(N_r)$ are unitary matrices, $D{\bf{U}}$ and $D{\bf{V}}$ represent the standard Haar integration measures of $\mathcal U(N_t)$ and $\mathcal U(N_r)$, respectively, $p{\left(\mathbf{H}\right)}$ is the joint PDF of the elements of $\bf H$ given by
\begin{equation}\label{eqn:joint_pdf_pH}
p{\left(\mathbf{H}\right)}=\frac{{\rm{etr}}\left( { - {\bf{U\Sigma }}{{\bf{V}}^{\rm H}}{{{\bf{R}}_t}^{ - 1}}{\bf{V}}{{\bf{\Sigma }}^{\rm{H}}}{{\bf{U}}^{\rm{H}}}{{{\bf{R}}_r}^{ - 1}}} \right)}{\det \left( {{{{\bf{R}}_r}^{N_t}}} \right)\det \left( {{{{\bf{R}}_t}^{N_r}}} \right)}.
\end{equation}
and the abbreviation $\rm{etr}{\left(\cdot\right)}$ denotes $\rm{etr}{\left(\mathbf X \right)} = exp\{\rm{tr}{\left(\mathbf X\right)}\}$.

With \eqref{eqn:out_def}, the outage probability can be represented by a multi-fold integral as
\begin{align}\label{eqn:out_prob_uv}
{p_{out}} =& \int\nolimits_{\prod\nolimits_{j = 1}^{N_r} {\left( {1 + \rho{\lambda _j}} \right)}  \le {2^R}} {{f_{\bs{\lambda }}}\left( {{\lambda _1}, \cdots ,{\lambda _{N_r}}} \right)d{\lambda _1} \cdots d{\lambda _{N_r}}} \notag\\
 =& \frac{{\rho^{-N_t N_r}}}{{N_r}!{\prod\nolimits_{i = 1}^{N_r} {\left( {{N_r} - i} \right)!\left( {{N_t} - i} \right)!} }}\int\nolimits_{\prod\nolimits_{j = 1}^{N_r} {\left( {1 + {\lambda _j}} \right)}  \le {2^R}} \frac{{{{\left( {\Delta \left( {\bs{\Lambda }} \right)} \right)}^2}\prod\nolimits_{j = 1}^{N_r} {{\lambda _j}^{{N_t} - {N_r}}} }}{{\det \left( {{{{\bf{R}}_r}^{N_t}}} \right)\det \left( {{{{\bf{R}}_t}^{N_r}}} \right)}}\notag\\
 &\times \int {D{\bf{V}}\int {D{\bf{U}}{{\rm{etr}}\left( { - {\rho}^{-1}{\bf{U\Sigma }}{{\bf{V}}^{\rm H}}{{{\bf{R}}_t}^{ - 1}}{\bf{V}}{{\bf{\Sigma }}^{\rm{H}}}{{\bf{U}}^{\rm{H}}}{{{\bf{R}}_r}^{ - 1}}} \right)}} } d{\lambda _1} \cdots d{\lambda _{N_r}}.
\end{align}
By applying ${\rm etr}{(\bf X)}=\sum\nolimits_{k=0}^{\infty} {({\rm tr}{\bf X})^k}/{k!} $ and $\int {D{\bf{V}}}  = \int {D{\bf{U}}}  = 1$ to \eqref{eqn:out_prob_uv}, ${p_{out}}$ is asymptotic to
\begin{align}\label{eqn:out_asy_fc_tayl}
{p_{out}} =& \frac{{\rho^{-N_t N_r}}}{{N_r}!{\prod\nolimits_{i = 1}^{N_r} {\left( {{N_r} - i} \right)!\left( {{N_t} - i} \right)!} }{\det \left( {{{{\bf{R}}_r}^{N_t}}} \right)\det \left( {{{{\bf{R}}_t}^{N_r}}} \right)}}\notag\\
&\times {\int\nolimits_{\prod\nolimits_{j = 1}^{N_r} {\left( {1 + {\lambda _j}} \right)}  \le {2^R}} {{{\left( {\Delta \left( {\bs{\Lambda }} \right)} \right)}^2}\prod\limits_{j = 1}^{N_r} {{\lambda _j}^{{N_t} - {N_r}}} d{\lambda _1} \cdots d{\lambda _{N_r}}}} + o\left( {{\rho ^{ - {N_t}{N_r}}}} \right).
\end{align}
Comparing (\ref{eqn:out_asy_fc_tayl}) and (\ref{eqn:F_G_mellin_trans_detzerofina}) yields \eqref{eqn:g_meijer_relation1}. By using the determinant expansion to \eqref{eqn:g_meijer_relation1} and Lemma \ref{the:leb_for}, ${g_{\bf{0}}}(R)$ can be obtained as
\begin{equation}\label{eqn:g0_deter_exp}
{g_{\bf{0}}}(R) =  \frac{\sum\nolimits_{{\bs\sigma} \in {S_{N_r}}} {{\mathop{\rm sgn}} \left( {{\bs\sigma }} \right)} \int\nolimits_{\prod\nolimits_{j = 1}^{N_r} {\left( {1 + {\lambda _j}} \right)}  \le {2^R}} \prod\nolimits_{i = 1}^{N_r} {{\lambda _i}^{{{\sigma _{i}}} + i +\tau - 1}} d{\lambda _1} \cdots d{\lambda _{N_r}}}{{\prod\nolimits_{j = 1}^{N_r} {\left( {{N_t} - j} \right)!\left( {{N_r} - j} \right)!} }}.
\end{equation}
By applying \cite[eq.(26)-(27)]{shi2017asymptotic} to \eqref{eqn:g0_deter_exp}, (\ref{eqn:g_meijer_relation2}) finally follows.

\section{Proof of Theorem \ref{the:conv_g0}}\label{app:conv_g0}
From the integral representation of ${g_{\bf{0}}}(R)$ in (\ref{eqn:g_meijer_relation1}), it is readily found that the outage probability is a monotonically increasing function of transmission rate $R$, because the region of integration expands as $R$ increases. Clearly, it is easily concluded that ${g_{\bf{0}}}(R)$ is a convex function if $N_t= N_r = 1$, because ${g_{\bf{0}}}(R)$ reduces to the convex function of $\mathfrak g_{(1)}(2^R)$.

To prove the convexity of ${g_{\bf{0}}}(R)$ for the cases of $N_t\ge 2$, it suffices to show that the second derivative of ${g_{\bf{0}}}(R)$ with respect to $R$ is larger than or equal to zero, i.e., ${g_{\bf{0}}}^{\prime\prime}(R) \ge 0$. Note that the contour integral representation of ${g_{\bf{0}}}(R)$ can be expressed as
\begin{align}\label{eqn:def_g_R_fullred}
g_{\bf n}(R) =
{\frac{1}{{2\pi {\rm{i}}}}\int\nolimits_{-c - {\rm i}\infty }^{-c + {\rm i}\infty } {\frac{1}{s}\prod\limits_{i = 1}^{N_r} {\frac{{\Gamma \left( {s  - {N_t} -i - {n_i}+1} \right)}}{{\Gamma \left( {s - i + 1} \right)}}}  {2^{Rs}}ds} }.
\end{align}
On the basis of \eqref{eqn:def_g_R_fullred}, taking the second derivative of $g_{\bf 0}(R)$ with respect to $R$ leads to
\begin{align}\label{eqn:def_g_R_fullred_second}
{g_{\bf{0}}}^{\prime\prime}(R)  
=&(\ln2)^2{\frac{1}{{2\pi {\rm{i}}}}\int\nolimits_{-c - {\rm i}\infty }^{-c + {\rm i}\infty } {s\prod\limits_{i = 1}^{N_r} {\frac{{\Gamma \left( {s  - {N_t} -i+1} \right)}}{{\Gamma \left( {s - i + 1} \right)}}}  {2^{Rs}}ds} }\notag\\
=&(\ln2)^2{\frac{1}{{2\pi {\rm{i}}}}\int\nolimits_{-c - {\rm i}\infty }^{-c + {\rm i}\infty } {{\frac{{\Gamma \left( {s  - {N_t}+1 } \right)}}{{\Gamma \left( {s} \right)}}}\prod\limits_{i = 2}^{N_r} {\frac{{\Gamma \left( {s  - {N_t} -i+1} \right)}}{{\Gamma \left( {s - i + 1} \right)}}}  {2^{Rs}}ds} }\notag\\
&+(\ln2)^2 N_t{\frac{1}{{2\pi {\rm{i}}}}\int\nolimits_{-c - {\rm i}\infty }^{-c + {\rm i}\infty } {\prod\limits_{i = 1}^{N_r} {\frac{{\Gamma \left( {s  - {N_t} -i+1} \right)}}{{\Gamma \left( {s - i + 1} \right)}}}  {2^{Rs}}ds} },
\end{align}
where the second equality holds by using $s{\Gamma \left( {s  - {N_t} } \right)} = {\Gamma \left( {s  - {N_t} + 1 } \right)} + N_t{\Gamma \left( {s  - {N_t} } \right)}$. Since the first derivative of $g_{\bf n}(R)$ can be expressed as
\begin{align}\label{eqn:def_g_R_fullredfirstd}
{g_{\bf n}}^{\prime}(R) =
\ln2{\frac{1}{{2\pi {\rm{i}}}}\int\nolimits_{-c - {\rm i}\infty }^{-c + {\rm i}\infty } {\prod\limits_{i = 1}^{N_r} {\frac{{\Gamma \left( {s  - {N_t} -i - {n_i}+1} \right)}}{{\Gamma \left( {s - i + 1} \right)}}}  {2^{Rs}}ds} },
\end{align}
\eqref{eqn:def_g_R_fullred_second} can be expressed in terms of ${g_{\bf n}}^{\prime}(R)$ as
\begin{align}\label{eqn:def_g_R_fullred_secrepfir}
{g_{\bf{0}}}^{\prime\prime}(R) =&\ln2{g_{\tilde{\bf 0}}}^{\prime}(R)+\ln2 N_t{g_{\bf 0}}^{\prime}(R),
\end{align}
where $\tilde{\bf 0}$ is all-zero $1\times N_r$ vector except its first element being $-1$, i.e., $\tilde{\bf 0} = (-1,0,\cdots,0)$. Moreover, the following lemma reveals the increasing monotonicity of $g_{\bf n}(R)$.
\begin{lemma}\label{the:inc_gn}
  $g_{\bf n}(R)$ is an increasing function of $R$ if $N_t + n_i > 0$ for $i \in [1, N_r]$, i.e., ${g_{\bf n}}^{\prime}(R) \ge 0$.
\end{lemma}
\begin{proof}
By using the relationship between beta function and Gamma function as ${\rm B}(\alpha,\beta) = \Gamma(\alpha)\Gamma(\beta)/\Gamma(\alpha+\beta)$ for $\Re{(\alpha)},\Re{(\beta)}>0$, \eqref{eqn:def_g_R_fullred} can be rewritten as
\begin{align}\label{eqn:def_g_R_fullredbeta}
g_{\bf n}(R) =
\prod\limits_{i = 1}^{N_r} \frac{1}{{\Gamma \left( {{N_t} + {n_i}} \right)}}{\frac{1}{{2\pi {\rm{i}}}}\int\nolimits_{-c - {\rm i}\infty }^{-c + {\rm i}\infty } {\frac{1}{s}\prod\limits_{i = 1}^{N_r} {{\rm{B}}\left( {s - {N_t} - i - {n_i} + 1,{N_t} + {n_i}} \right)}  {2^{Rs}}ds} }.
\end{align}
By using the integral representation of Beta function as ${\rm{B}}(\alpha,\beta)=\int\nolimits_0^1x^{\alpha-1}(1-x)^{\beta-1}dx$, \eqref{eqn:def_g_R_fullredbeta} can be expressed as
\begin{align}\label{eqn:def_g_R_fullredbetainte}
g_{\bf n}(R) =& \prod\limits_{i = 1}^{{N_r}} {\frac{1}{{\Gamma \left( {{N_t} + {n_i}} \right)}}} \int\nolimits_{0 \le {x_1},\cdots,{x_{{N_r}}} \le 1} {\prod\limits_{i = 1}^{{N_r}} {{x_i}^{ - {N_t} - i - {n_i}}{{(1 - {x_i})}^{{N_t} + {n_i} - 1}}} d{x_1} \cdots d{x_{{N_r}}}}\notag\\
&\times \frac{1}{{2\pi {\rm{i}}}}\int_{ - c - {\rm{i}}\infty }^{ - c + {\rm{i}}\infty } {\frac{1}{s}{e^{\left( {R\ln 2 + \sum\nolimits_{i = 1}^{{N_r}} {\ln {x_i}} } \right)s}}ds} \notag\\
 =& \prod\limits_{i = 1}^{{N_r}} {\frac{1}{{\Gamma \left( {{N_t} + {n_i}} \right)}}} \int\nolimits_{\scriptstyle0 \le {x_1},\cdots,{x_{{N_r}}} \le 1\hfill\atop
\scriptstyle\sum\nolimits_{i = 1}^{{N_r}} {\ln {x_i}}  \ge  - R\ln 2\hfill} {\prod\limits_{i = 1}^{{N_r}} {{x_i}^{ - {N_t} - i - {n_i}}{{(1 - {x_i})}^{{N_t} + {n_i} - 1}}} d{x_1} \cdots d{x_{{N_r}}}},
\end{align}
where the last step holds by using the inverse Laplace transform of the step unit function as $u(x) = \frac{1}{{2\pi {\rm{i}}}}\int_{ - c - {\rm{i}}\infty }^{ - c + {\rm{i}}\infty } {{1}/{s}{e^{sx}}ds}$. Clearly, the integrand in \eqref{eqn:def_g_R_fullredbetainte} is larger than or equal to zero, and the range of the integration enlarges as $R$ increases. Thus we conclude that $g_{\bf n}(R)$ is an increasing function of $R$.
\end{proof}
According to Lemma \ref{the:inc_gn}, \eqref{eqn:def_g_R_fullred_secrepfir} shows ${g_{\bf{0}}}^{\prime\prime}(R) \ge 0$. The proof is therefore accomplished.

\bibliographystyle{ieeetran}
\bibliography{mimo}

\end{document}